\documentclass[12pt]{article}
\usepackage[T1]{fontenc}
\usepackage{lmodern}

\usepackage{sectsty}
\allsectionsfont{\textsf}
\sectionfont{\centering\sffamily}
\subsectionfont{\centering\sffamily}
\subsubsectionfont{\centering\sffamily}

\usepackage{etoolbox}
\usepackage{mathptmx}
\patchcmd{\abstract}{\sfshape\abstractname}{{\sf \abstractname}}{}{}

\usepackage{amssymb}
\usepackage{bm}
\usepackage[hiresbb]{graphicx}
\usepackage{tikz}
\usepackage[cmex10]{amsmath}
\usepackage{empheq}
\usepackage{amsthm}
\usepackage{mathtools}
\usepackage[longnamesfirst]{natbib}
\usepackage{xcolor}
\usepackage[%
 setpagesize=false,%
 bookmarks=true,%
 bookmarksdepth=tocdepth,%
 bookmarksnumbered=true,%
 colorlinks=true,
 citecolor=black,
 urlcolor=blue,
 linkcolor=blue,%
 pagebackref=true, %
 backref=page,
 pdftitle={Normative Ceilings and Stochastic Giving},%
 pdfsubject={Stochastic choice, social preferences, menu-dependent perturbed utility},%
 pdfauthor={Yosuke Hashidate and Keisuke Yoshihara},%
 pdfkeywords={stochastic choice; social preferences; menu-dependent perturbed utility; personal norms; deliberate randomization; regularity}%
]{hyperref}
\usepackage{hypernat}
\usepackage[margin= 1in]{geometry}
\usepackage{latexsym}
\usepackage{booktabs}
\usepackage{multirow}
\def\qed{\hfill $\Box$}
\theoremstyle{plain}
\newtheorem{theorem}{Theorem}[section]
\theoremstyle{plain}
\newtheorem{proposition}{Proposition}[section]
\theoremstyle{plain}
\newtheorem{lemma}{Lemma}[section]
\theoremstyle{plain}
\newtheorem{claim}{Claim}[section]
\theoremstyle{plain}

\theoremstyle{plain}
\newtheorem{corollary}{Corollary}[section]
\theoremstyle{definition}
\newtheorem{remark}{Remark}[section]
\theoremstyle{definition}
\newtheorem{definition}{Definition}[section]
\theoremstyle{definition}
\newtheorem{axiom}{Axiom}[section]
\theoremstyle{definition}
\newtheorem{example}{Example}[section]
\theoremstyle{plain}

\renewcommand*{\backref}[1]{}
\renewcommand*{\backrefalt}[4]{[%
    \ifcase #1 Not cited.%
          \or Cited on pp.~#2.%
          \else Cited on pp. #2.%
    \fi%
    ]}

\begin{document}
\title{\textsf{Normative Ceilings and Stochastic Giving}\thanks{First Draft: April 27, 2020. Earlier versions of this study were presented under the titles ``Stochastic Choice and Social Preferences: Inequity Aversion versus Shame Aversion'' and ``Social Preferences and Deliberately Stochastic Behavior.'' We thank Ryota Iijima, Tetsuya Kawamura, Fabrice Le Lec, Yusuke Osaki, and the participants at the CTW Summer Camp (Hirosaki University), ABEF (Nagoya University of Commerce and Business), the Game Theory Workshop (Komazawa University and Online), ASSA 2021 (Virtual), Kansai University, Hitotsubashi University, and EWMES 2021(virtual) for their helpful comments and suggestions. This research was supported by a grant-in-aid from the Tokyo Center for Economic Research (TCER) and JSPS (KAKENHI Grant Number JP19J01049 and JP20K13457) (Hashidate). Part of this research was completed while Yoshihara was a research associate at the Center for Mathematics and Data Science, Gunma University. Yoshihara is grateful for their hospitality. We thank Takumi Yamaguchi for his excellent expert research assistance. The usual disclaimer applies. All remaining errors are ours. During the preparation of this work, the authors used Claude (Anthropic), through Claude Code, in order to edit and restructure the English text, check proofs and numerical results, and search for and check the literature. After using this tool, the authors reviewed and edited the content as needed and take full responsibility for the content of the published article.}} \author{\textsc{Yosuke Hashidate}\thanks{Affiliation: Faculty of Economics, Sophia University; Address: 7-1, Kioi-cho, Chiyoda-ku, Tokyo 102-8554, Japan; Email: \texttt{hashidate@sophia.ac.jp}} \and \textsc{Keisuke Yoshihara}\thanks{Email: \texttt{ksk0110@gmail.com}}} \date{Current Draft: \today}
\maketitle

\begin{abstract}
People randomize when they allocate money to others, and what else is available changes how they randomize: adding an alternative can raise the probability that an existing one is chosen, in violation of Regularity. We characterize a model of stochastic choice in which a selfish decision maker faces norm-driven social pressure and the menu matters only through one scalar, the \emph{normative ceiling}, defined as the most pro-social alternative available. We introduce two conditions on cross-menu comparisons: comparisons line up through a single rescaling attached to each menu’s most normative alternative, and higher normative levels are more selective. Together with standard positivity, continuity, acyclicity, order and richness requirements, these conditions force the randomization cost to depend on the menu only through the ceiling. Raising the ceiling makes the cost uniformly more curved. Random utility and baseline perturbed utility models imply Regularity; models of deliberate randomization necessarily violate it; ours does neither. Instead, our model predicts that Regularity can fail only when the ceiling rises across behaviorally distinguishable levels, and that it then fails at the \emph{least} selfish alternatives. This prediction is a sign restriction with no free parameters, and we show how nested menus test it.
\end{abstract}

\vspace{0.2in} \noindent \textbf{Keywords}: Stochastic Choice, Social Preferences, Menu-Dependent Perturbed Utility, Personal Norms, Deliberate Randomization, Regularity. \\ \noindent \textbf{JEL Classification Numbers}: D63; D64; D81.

\begin{flushright} 
Man is by nature a social animal; an individual who is unsocial naturally and not accidentally is either beneath our notice or more than human. \\
Aristotle \\
Politics, c. 328 BC
\end{flushright}

\section{Introduction}
\label{intro}

\subsection{Motivation}
\emph{Regularity}---the property that adding an alternative to a menu cannot raise the probability that an existing one is chosen---holds in every random utility model, and it is the first thing a theory of stochastic choice is asked about. It also fails. Decision makers who randomize \emph{deliberately}---rather than because a shock happened to fall one way---are known to violate it, both theoretically \citep{M_1985} and experimentally \citep{AO_2017}. In social environments, where the decision is how much of a fixed sum to give away, the violations are neither rare nor universal: in the Probabilistic Dictator Game of \citet{FOOV_2025}, $42\%$ of subjects violate Regularity and $58\%$ satisfy it. The question we take up is what a violation tells us: not whether a model can accommodate one---many can---but which menu changes produce one and which alternatives it hits.

Consider a charity’s donation page. It offers a few suggested amounts, and the organization is deciding whether to add a conspicuously large tier at the top---one that almost nobody selects. On the usual accounts the addition is either irrelevant or draws weight away from the amounts below it. Our model predicts something else: the new tier raises the share going to the \textit{middle} amount and lowers the share going to the smallest, even though it is itself rarely chosen. The work is done not by the money, since the tier is not taken, but by the fact that it changes which option is the most generous on the page. And the prediction has content because it also says when no violation should occur: adding a tier \textit{below} the most generous one already offered should move weight in the ordinary direction.

The Probabilistic Dictator Game of \citet{FOOV_2025} recreates the donation page in the laboratory, where the shares are observed directly. Write an allocation as (own, recipient) and take their four menus,
\begin{equation*}
    A_1=\{(0,10),(5,5),(10,0)\}, \quad
    A_2=\{(5,5),(10,0)\}, \quad
    A_3=\{(0,10),(10,0)\}, \quad
    A_4=\{(0,10),(5,5)\},
\end{equation*}
so that $A_2,A_3,A_4$ are the three binary subsets of $A_1$. Throughout this illustration we use the canonical model of Example \ref{ex:base_model}, with $u(x_1)=x_1$, $\varphi({\bm x})=x_2$ and $\gamma(t)=e^{t/5}$.

Facing $A_2$, a decision maker who is fundamentally selfish but feels the pull of the equal split need not commit. The equal split is the most pro-social alternative on the table---the menu's \emph{normative ceiling}---and committing to $(10,0)$ in its presence is costly in a way that committing to it in private is not. Randomizing softens that cost at some sacrifice of money, and in the canonical model the resulting probabilities are $\rho((10,0),A_2)=0.715$ and $\rho((5,5),A_2)=0.285$.

Now add $(0,10)$, the option of giving everything away, to $A_2$. The ceiling rises from the equal split to $(0,10)$, and a more demanding benchmark makes randomization harder to adjust at the margin---formally, the cost of randomizing becomes more \emph{curved}, so choice probabilities respond less to money. Compression of this kind does not move probability in one direction: it shifts probability from the alternatives that were likely toward those that were unlikely. The selfish option was the likely one and the equal split the unlikely one, so the equal split is the one that gains,
\begin{equation*}
    \rho\bigl((5,5),A_1\bigr)=0.322 \;>\; 0.285 \;=\; \rho\bigl((5,5),A_2\bigr),
    \qquad
    \rho\bigl((10,0),A_1\bigr)=0.483 \;<\; 0.715,
\end{equation*}
while the new option itself takes $0.195$. Regularity fails, and it fails at the equal split---never at $(10,0)$. An option that almost nobody takes has made the compromise more likely.

The other two expansions are the control, and they are what makes the prediction testable. In $A_3$ and $A_4$ the ceiling is \emph{already} $(0,10)$, so adjoining the missing alternative does not raise it, and the model requires Regularity to hold strictly: $\rho\bigl((10,0),\cdot\bigr)$ falls from $0.663$ to $0.483$ and $\rho\bigl((0,10),\cdot\bigr)$ from $0.337$ to $0.195$ in the first case, and $\rho\bigl((5,5),\cdot\bigr)$ falls from $0.584$ to $0.322$ in the second. So across these four menus the model permits a violation in exactly one of the three nested comparisons, and, within that comparison, at only one of the two incumbents, with no free parameters and nothing estimated. Proposition \ref{prop:regularity} turns this pattern into a theorem: the entire menu effect is disciplined by a single index---the normative ceiling---and the direction of the resulting violation is pinned by where each alternative sat in the original distribution.

\paragraph*{The gap: models that violate Regularity do not say where.} Models that violate Regularity are not scarce, so it is worth saying precisely what is missing. The recent experimental evidence of \citet{FOOV_2025} maps the gap sharply. In their second experiment---a Probabilistic Dictator Game in which $104$ dictators assign weights to allocations of \$10 across the four nested menus above---the ex-ante/ex-post model of \citet{S_2013} explains $27\%$ of subjects and the quadratic social welfare model of \citet{ES_1992} $41\%$, because both predict either a deterministic choice or randomization over at most two allocations, whereas subjects spread weight broadly. Only $17\%$ choose deterministically in all four problems. Turning to stochastic choice models, they find that $58\%$ of subjects satisfy Regularity and $42\%$ violate it, and they report the consequences for the perturbed utility family: the baseline and allocation-dependent versions of \citet{FIS_2015} \emph{imply} Regularity and so cannot accommodate the violators, while the menu-dependent version explains $78.85\%$ precisely because it can.

The explanatory rates leave a well-defined hole. On one side sit random utility, baseline APU and allocation-dependent APU, which force Regularity and are therefore refuted for two subjects in five. On the other side sits \citet{CDOR_2019}, where stochasticity reflects a complementarity among the elements of the menu and Regularity violations are \emph{necessary}---too strong for the three subjects in five who satisfy it. In between sits the menu-dependent APU of \citet{FIS_2015}, which fits both groups, but only because it places no restriction whatever on how the cost $C_A$ varies with $A$: it buys its explanatory power with silence about which menu changes produce violations. A $42\%$ violation rate calls for neither extreme, but for structure.

Our contribution is to supply the missing structure, and we impose no functional form to do so. What we impose is that cross-menu comparisons of common alternatives line up through a single rescaling attached to the menu's ceiling-setting alternative (\textsf{Normative Independence}), and that a higher normative level is \emph{more selective} in \citeauthor{FIS_2015}'s sense (\textsf{Normative Selectivity}). Together with standard positivity, continuity, acyclicity and richness conditions, these imply that the menu enters only through the normative ceiling. The compression property that the example above describes is not an assumption but a consequence (Proposition \ref{prop:penalty_implied}), and so is the technical richness condition that \citet{FIS_2015} impose (Proposition \ref{prop:richness_redundant}). The resulting model accommodates broad randomization and accommodates Regularity violations, yet predicts exactly when they occur and whom they hit---in \citeauthor{FOOV_2025}'s design, the sign restriction illustrated above. Section \ref{sec:regularity} develops this, together with a design that separates the normative ceiling from the competing reading under which what matters is the decision maker's own exposure to a low payoff.

\paragraph*{The motive is the cost of committing, not ex-ante fairness.} The motive at work here should be distinguished from ex-ante fairness, which is the usual reason to expect deliberate randomization in social settings: a decision maker who randomizes in order to give the recipient a fair \emph{opportunity} \citep{S_2013}. The distinction is clean and it is testable. Ex-ante fairness is a property of the chosen lottery---it evaluates the distribution induced over the recipient's payoffs---and is therefore indifferent to what else was on the table. The normative-ceiling motive is a property of the \emph{menu} from which the lottery was chosen. Only the latter predicts that an alternative which is never chosen can change behavior merely by being available, which is why nested-menu designs, and Regularity in particular, identify which motive is operative. Section \ref{sec:social_scope} makes both halves of this precise: the two families do not nest as representations (Proposition \ref{prop:exante}), and the \emph{pattern} of Regularity violations separates them, because an ex-ante fairness motive can raise the probability of the most selfish allocation while ours provably cannot. A further consequence, easy to miss, is that the framework does not identify a personal norm with inequality aversion. The norm is recovered from behavior rather than assumed, so whether a decision maker's norm is inequality-averse is an identified empirical question.

The donation page of the opening illustration makes the contrast concrete. Almost nobody selects the top tier, so the distribution over the recipient's payoff conditional on any given choice is untouched and ex-ante fairness has nothing to say; yet its presence changes how weight is spread over the remaining amounts. The force is the cost of \emph{committing} under a visible norm, not the value of the opportunity conferred; the motive is related to the one documented by \citet{DWK_2007}, where decision makers manipulate what they know in order to avoid a definite verdict on themselves, rather than manipulating what the recipient receives.

\subsection{Overview}
Our approach builds on the menu-dependent additive perturbed utility (APU) framework of \citet{FIS_2015},
\begin{align*}
	\rho(A) = \underset{\rho \in \Delta(A)} {\operatorname{argmax}} \sum_{{\bm x} \in A} \Bigl(u({\bm x}) \rho({\bm x}) - C_A(\rho({\bm x})) \Bigr),
\end{align*}
in which the decision maker trades the expected utility of a stochastic choice against a menu-dependent cost $C_A$ of randomizing. The framework is permissive by design: $C_A$ may vary with $A$. Our question is what disciplines that variation when the menu carries a \emph{personal} norm---a standard the decision maker feels bound by, which need not be shared by anyone else and which we recover from her choices rather than impose.

We make three contributions. The first contribution is a reduction. Seven behavioral conditions collapse the menu-specific cost functions to a single-index form,
\begin{align*}
C_A(\cdot) = c \Bigl(\cdot, \max_{{\bm x} \in A} \varphi({\bm x}) \Bigr),
\end{align*}
so that the whole of the menu's influence runs through one scalar: the \emph{normative ceiling} $\varphi(A)=\max_{{\bm x}\in A}\varphi({\bm x})$, the level of the most normative alternative available. The personal-norm index $\varphi$ is not imposed exogenously. It is derived from choice by a calibrated comparison (Definition \ref{norm}) and is unrestricted beyond continuity, so whether a decision maker’s norm is inequality-averse, increasing in what the recipient receives, or of some other form is a question about the data rather than a modeling choice. Two properties are consequences of the reduction rather than inputs to it: the compression property (Proposition \ref{prop:penalty_implied}) and the richness condition that \citet{FIS_2015} impose alongside their axioms (Proposition \ref{prop:richness_redundant}).

The second contribution concerns identification. Raising the ceiling makes the randomization cost uniformly more \emph{curved}, and it is the curvature, not the level, of the marginal cost that behavior pins down: level comparisons across ceilings are a normalization rather than a property of $\rho$ (Remark \ref{rmk:level_corollary}). The index itself is ordinally unique, up to the merging of levels that induce the same behavior in every menu (Theorem \ref{unique_norm}).

The third contribution is the model's testable content. Random utility and baseline perturbed utility models imply Regularity; models of deliberate randomization necessarily violate it; ours does neither. We show that violations can occur only when the normative ceiling rises across behaviorally distinguishable levels, and that they then hit the \emph{least} selfish alternatives in the menu---raising the probabilities that were low and lowering those that were high (Proposition \ref{prop:regularity}). We also delimit the framework: a norm that is inequality-averse is inside it, while ex-ante fairness is provably outside and is separated from it by the same pattern of violations (Section \ref{sec:social_scope}).

\subsection{Related Literature}
Deliberate randomization---choosing a lottery over alternatives rather than committing to one of them---has been documented across settings. \citet{M_1985} raised the possibility on theoretical grounds, \citet{AO_2017} established it experimentally, \citet{DKW_2019} found it in the field, where applicants to German universities randomized over the schools they ranked, and \citet{AO_2022} survey the revealed-preference literature it has generated. The benchmark model is \citet{CDOR_2019}'s: there randomization is valuable because the elements of a menu are complementary, and Regularity violations follow \emph{necessarily}. That is the wrong prediction for the three subjects in five who satisfy Regularity in \citeauthor{FOOV_2025}'s data, which is the gap this paper fills.

In social settings the usual explanation for randomizing is ex-ante fairness: a decision maker randomizes to equalize the recipients' \emph{opportunities} rather than their payoffs \citep{S_2013, BLO_2013, MZ_2018}. In earlier work we characterized the stochastic choice behavior that \citeauthor{S_2013}'s model generates, and found that it violates both first-order stochastic dominance and Regularity \citep{HY_2021}. Our explanation here is different in kind. It locates the force in the cost of \emph{committing} under a visible norm, which makes it a property of the menu rather than of the chosen lottery, and Section \ref{sec:social_scope} shows that the two are non-nested as representations and separated by the pattern of Regularity violations.

Our model is a special case of the item-invariant APU of \citet{FIS_2015}, which places no restriction on how the cost $C_A$ varies across menus.\footnote{The first step of our proof architecture invokes Proposition 8 of \citet{FIS_2014} as an existential baseline; that proposition appears in the working paper and not in the published version, which is why both are cited here. The contribution of this paper begins immediately after this step.} We restrict that variation with two axioms: cross-menu comparisons are tied to the menu’s ceiling-setting alternative (Axiom \ref{ax:norm_independence}), and higher normative levels are more selective (Axiom \ref{ax:selectivity}). The remaining axioms, such as \textsf{Selfishness-Acyclicity}, recover the full model of \citet{FIS_2015}; the two new conditions are exactly what pin the menu dependence to a single index.

Figure \ref{fig:relatiohship} places the model. It sits strictly inside the menu-dependent APU universe of \citet{FIS_2015}, and it straddles the Regularity boundary---which is the point. Random utility and the baseline perturbed utility models lie wholly on the Regularity side and so cannot accommodate the violators; the unrestricted menu-dependent model crosses it too, but says nothing about where the boundary runs. Ours crosses it \emph{and} locates it.

\begin{figure}[htbp]
	\begin{center}
		\begin{tikzpicture}[scale=0.80]
		  \footnotesize
		  \draw[dashed, thick, gray!70] (1.3,-3.3) -- (1.3,3.3);
		  \node[gray!70, anchor=south east] at (1.1,3.0) {Regularity holds};
		  \node[gray!70, anchor=south west] at (1.5,3.0) {Regularity fails};
		  \draw[thick] (0,0) ellipse (5.1 and 2.6);
		  \node at (-0.3,-2.05) {Menu-Dependent APU};
		  \fill[red!14] (1.5,-0.3) ellipse (2.3 and 1.3);
		  \draw[thick, red!60!black] (1.5,-0.3) ellipse (2.3 and 1.3);
		  \node[align=center, red!45!black] at (1.5,-0.3) {Norm-Dependent APU\\(this paper)};
		  \draw (-3.4,1.3) ellipse (2.2 and 1.4);
		  \node at (-3.4,1.95) {Random Utility};
		  \draw (-2.3,-0.4) ellipse (1.9 and 0.9);
		  \node at (-2.6,-0.85) {Baseline APU};
		  \node[font=\scriptsize\itshape] at (-3.0,0.15) {logit};
		\end{tikzpicture}
		\caption{Where the model sits. Random utility and the baseline (menu- and allocation-independent) perturbed utility models satisfy Regularity; their intersection is non-empty and contains logit, which is the APU with cost $c(p)=\alpha p\ln p+\gamma p+\delta$ \citep[Theorem 2]{FIS_2015} and, as is well known, equally the random utility model with i.i.d.\ Gumbel shocks. Neither class contains the other: \citet[\S5.2]{FIS_2015} show that item-invariant APU---the menu-dependent model, with a pair $(u,(C_A)_{A\in\mathcal A})$---nests every \emph{symmetric} random utility model, while a baseline APU can violate the Block--Marschak conditions \citep{BlM_1960} and so fail to be a random utility model at all. Because it places no restriction on how $C_A$ varies with $A$, the menu-dependent model straddles the Regularity boundary without predicting where that boundary lies. The norm-dependent APU model straddles it too, but Proposition \ref{prop:regularity} locates it: Regularity can fail only when the normative ceiling rises across behaviorally distinguishable levels. See also \citet{S_2025}.}
		\label{fig:relatiohship}
	\end{center}
\end{figure}

The idea that the menu sets a normative standard, and that failing to meet it is costly, is \citet{DS_2012}’s. Their decision maker is ashamed to be selfish when a fairer allocation was available, and the normatively best allocation in the set---the Nash solution of an induced bargaining game---is what she is measured against. Our normative ceiling carries that object into stochastic choice. Their decision maker responds to the standard by choosing differently; ours responds by randomizing. The move has two payoffs. The primitive is weaker: choice frequencies from menus, which experiments produce, rather than a ranking of menus, which they do not. And the prediction is new: their decision maker would pay to avoid the larger menu, while ours reallocates probability within it (Proposition \ref{prop:regularity}). We are also more permissive about the norm and more restrictive about how it enters: $\varphi$ is unrestricted beyond continuity, but it acts on behavior only through a single scalar. An earlier note of ours \citep{HY_2019} contains the mechanism informally: adding a normatively better alternative to a two-element menu raises the probability of the compromise and lowers that of the selfish option. What is new here is what makes that observation testable---which menu changes produce the effect, which alternatives it hits, and how the pattern is told apart from ex-ante fairness. On the last point the present paper also corrects that note, which suggested that making the weight on the ex-ante term menu-dependent would bring ex-ante fairness inside the perturbed-utility form; Proposition \ref{prop:exante} shows that it does not.

That the standard is social, and that the decision maker would rather not face it, are documented in the experimental literature on giving. Dictators take costly quiet exits when the recipient cannot observe them \citep{DCD_2006}, use risk as an excuse for selfishness \citep{E_2016}, and give differently once the interpretation of the gift changes \citep{L_2007}; image concerns and incentives interact in the way \citet{BT_2006} describe. Closest among our own work is \citet{H_2021}, which studies reference points formed by social image rather than a standard set by the menu. The menu-level index is what is new here.

\subsection{Outline}
The remainder of the paper is organized as follows. Section \ref{axioms} formalizes the domain and introduces our core behavioral axioms, culminating in the main characterization result (Theorem \ref{representation_norm}); Theorem \ref{unique_norm} establishes the ordinal uniqueness of the personal norm function. Section \ref{discussion} derives the model’s behavioral implications: Proposition \ref{prop:regularity} characterizes exactly when, and for which alternatives, Regularity fails, and Section \ref{sec:social_scope} delimits what the framework includes, showing in particular that ex-ante fairness lies outside it (Proposition \ref{prop:exante}). Section \ref{conclusion} concludes. All proofs are in Appendix \ref{app:proofs}.

\section{Axiomatization}
\label{axioms}
Let $I = \lbrace 1, 2 \rbrace$ be a set of individuals, where $1$ is the decision maker and where $2$ is the other (passive) agent.\footnote{The analysis extends to the case of $n$ agents.} We assume that the set of payoffs is $\mathbb{R}$: A vector $\bm{x} = (x_1, x_2) \in \mathbb{R}^2$ is called an \textit{allocation} of payoffs among individuals, yielding a payoff $x_i \in \mathbb{R}$ for each $i \in I$. Let $X \subseteq \mathbb{R}^2$ be a convex set of allocations, which we do not assume to be bounded (this is what allows the \textsf{Richness} condition of Definition \ref{richness} (Appendix \ref{app:fis}) to hold; the specific fixed-endowment examples used for illustration, such as the \$10 dictator game in Section \ref{intro}, concern a single finite menu drawn from this larger $X$, not $X$ itself). A choice set, the \textit{menu}, is a nonempty finite subset of $X$. Let $\mathcal{A}$ be the collection of all such menus. The elements in $\mathcal{A}$ are denoted by $A, B, C \in \mathcal{A}$.

As the primitive of this paper, we study a \textit{stochastic choice function} $\rho$ that maps a \textit{menu} $A$ to a probability distribution over the \textit{allocations} in menu $A$, denoted by $\rho(A)$. Formally, we denote the stochastic choice rule by $\rho: \mathcal{A} \rightarrow \Delta(X)$, where $\Delta(X)$ is the set of probability distributions over $X$ with finite support; for a menu $A\in\mathcal A$ we write $\Delta(A)$ for those whose support is contained in $A$. Given a menu $A \in \mathcal{A}$ with ${\bm x} \in A$, we denote the probability that an allocation ${\bm x}$ is chosen from menu $A$ by a nonnegative number $\rho({\bm x}, A) \geq 0$ and $\sum_{{\bm x} \in A} \rho({\bm x}, A) = 1$. $\rho(A)$ defines the probability distribution over menu $A$. For example, consider a menu $A = \lbrace {\bm x}, {\bm y} \rbrace$. Then, $\rho(A) = ({\bm x}, \rho({\bm x}, A); {\bm y}, \rho({\bm y}, A))$, where the allocation ${\bm x}$ is realized with probability $\rho({\bm x}, A)$ and the allocation ${\bm y}$ is realized with probability $\rho({\bm y}, A)$.

\paragraph*{Model} The model has three ingredients. First, let $u: \mathbb{R} \rightarrow \mathbb{R}$ be a continuous and monotonic function with full range, $u(\mathbb{R}) = \mathbb{R}$. We call $u$ a \textit{selfish} utility function if, for each $x_1, y_1 \in \mathbb{R}$, $x_1 > y_1 \Leftrightarrow u(x_1) > u(y_1)$, and $x_1 = y_1 \Leftrightarrow u(x_1) = u(y_1)$. The utility function $u$ captures the decision maker's ranking. The full-range requirement is the representation-side counterpart of the richness conditions (Definition \ref{richness} in Appendix \ref{app:fis}, and Axiom \ref{norm_richness} below): behaviorally, it is what those conditions assert, and recording it here lets the necessity arguments in Appendix \ref{app:necessity} invoke it directly rather than circularly through the very axioms being derived.

Second, let $\succsim_n^{\rho}$ over $X$ be a personal norm ranking over allocations, formally introduced in Definition \ref{norm}. The binary relation $\succsim_n^{\rho}$ captures the decision maker's personal norm ranking, which is generally more altruistic than the decision maker's personal ranking. We call a continuous function $\varphi: X \rightarrow \mathbb{R}$ a \emph{personal norm utility function}. The relation $\succsim_n^\rho$ is derived from $\rho$, while $\rho$ is in turn generated by the tuple $(u,\varphi,c)$; so ``$\varphi$ represents $\succsim_n^\rho$'' is not a property one can read off the ingredients of the tuple, and we do \emph{not} build it into the model. It is instead a \emph{theorem}: Part (i) of the necessity argument in Appendix \ref{app:necessity} shows that ${\bm x}\succsim_n^\rho{\bm y}\iff\varphi({\bm x})\ge\varphi({\bm y})$ for every normalized representation, and this is what gives $\varphi$ its interpretation and makes \textsf{Normative Weak Order} necessary. Example \ref{ex:eta_failure} shows that the conclusion genuinely needs condition (iv) below: without it there are tuples satisfying every other requirement whose induced $\succsim_n^\rho$ is not represented by $\varphi$---indeed is not even transitive.

Third, $c: [0,1] \times \varphi(X) \rightarrow \mathbb{R} \cup \lbrace + \infty \rbrace$ is a strictly convex function in the first argument that satisfies
	\begin{enumerate}
		\renewcommand{\labelenumi}{(\roman{enumi})}
			\item $C^2$ over $(0,1)$ (so that $c_1=c'$ and $c_{11}=c''$ exist and are continuous on $(0,1)$);

			\item (Inada) $c_1(p,t)\to-\infty$ as $p\to 0^+$ and $c_1(p,t)\to+\infty$ as $p\to 1^-$, for
			every $t\in\varphi(X)$; equivalently, $c_1(\cdot,t)$ maps $(0,1)$ onto $\mathbb R$; and

			\item if $\varphi({\bm x})<\varphi({\bm y})$, then $c_{11}(p,\varphi({\bm x})) \le c_{11}(p,\varphi({\bm y}))$ for every $p\in(0,1)$; and

			\item (Level activity) for any two distinct levels $s<s'$ in $\varphi(X)$, the
			difference $D_{s,s'}:=c_1(\cdot,s')-c_1(\cdot,s)$---which is non-decreasing by
			(iii)---is either constant on $(0,1)$ or non-constant on $(0,\tfrac12)$.
	\end{enumerate}
\begin{definition}[\textsf{Norm-Dependent Additive Perturbed Utility}]
\label{norm-APU}
$\rho$ is a \textit{norm-dependent additive perturbed utility} if there exists a tuple $(u, \varphi, c)$, where $u: \mathbb{R} \rightarrow \mathbb{R}$ is a selfish utility function, $\varphi: X \rightarrow \mathbb{R}$ is a personal norm utility function (that is, a continuous function; that it represents $\succsim_n^\rho$ is derived, not assumed), and $c: [0,1] \times \varphi(X) \rightarrow \mathbb{R} \cup \lbrace \infty \rbrace$ is a norm-dependent cost function such that $\rho$ is represented by
\begin{equation}
\label{model_norm}
	\rho(A) = \underset{\rho \in \Delta(A)} {\operatorname{argmax}} \sum_{{\bm x} \in A} \Bigl(u(x_1) \rho({\bm x}) - c(\rho({\bm x}), \varphi(A)) \Bigr),
\end{equation}
where $\varphi(A) := \max_{{\bm x} \in A} \varphi({\bm x})$.
\end{definition}

\paragraph*{On the selfishness of $u$.} Making the decision maker's own utility depend on $x_1$ alone is the paper's most visible restriction, and it should be stated as a scope condition rather than slipped past. It rules out, by construction, a decision maker whose \emph{ranking} of allocations is altruistic, inequity-averse, or maximin: within any menu the first-order conditions make the choice probability strictly increasing in $u(x_1)$, so the most selfish alternative available is always the modal one. That is a sharp and falsifiable restriction. It is borne out in two of the three binary menus of \citet{FOOV_2025}, where dictators ``overwhelmingly selected the allocation yielding a higher private payoff''; the menu $\{(5,5),(10,0)\}$, where the equal split is the salient alternative, is the borderline case and the sharpest test of the model we know of. Two further questions---why the restriction is what makes $\varphi$ identifiable, and why it does not amount to assuming a social preference---are taken up in Section \ref{sec:social_scope}, where they belong with the comparison to the inequality-aversion literature.

Condition (i) is what makes $c_{11}$ available. Condition (ii) is the Inada requirement of the perturbed utility literature, imposed at \emph{both} endpoints and on the \emph{marginal} cost rather than on $c$ itself---the convention $c(0,t)=c(1,t)=+\infty$ would be strictly stronger and would exclude the entropy cost $c(p,t)=\gamma(t)[p\ln p+(1-p)\ln(1-p)]$, which is finite at both endpoints while its derivative diverges. Two consequences are used throughout: divergence at $p=1$ rules out deterministic choice from a menu with at least two elements, and divergence at $p=0$ forces full support, which is the behavioral content of \textsf{Positivity} (Axiom \ref{positivity}); together they make $c_1(\cdot,t)$ a bijection onto $\mathbb R$, whose inverse the calibration arguments of the Appendix use. Condition (iii) states monotonicity for the \emph{curvature} rather than the \emph{level} of the marginal cost, which is the version that is invariant across representations (Remark \ref{rmk:level_corollary}).

Condition (iv) is less familiar, and it is not a technical convenience: without it the model is not characterized by the axioms. By (iii) a pair of levels falls into exactly one of three cases, and (iv) excludes the middle one, in which the two levels differ behaviorally and yet no calibrated comparison can detect the difference---so that Definition \ref{norm} ranks some pairs as indifferent \emph{vacuously} and the resulting relation can fail transitivity. Excluding it costs nothing in applications, since every \emph{scale} family $c_1(p,t)=\gamma(t)\psi(p)$ satisfies it. Appendix \ref{app:iv_free} sets out the three cases, exhibits the pathology, and shows that (iv) narrows the class of admissible cost families without narrowing the class of behaviors. Definition \ref{def:level_separation} below returns to (iv) from the identification side.

\subsection{Axioms}
\label{sec:axioms_list}
Seven conditions characterize the model. \textsf{Positivity} is maintained throughout; the other six are the content of Theorem \ref{representation_norm}(a). Each does one job in the argument, as the following table summarizes.

\begin{center}
\footnotesize
\begin{tabular}{@{}p{42mm}p{86mm}@{}}
\toprule
Condition & What it does \\
\midrule
\textsf{Positivity} & Puts every choice probability in $(0,1)$, so the interior first-order conditions are available. Behavioral counterpart of the Inada condition in Definition \ref{norm-APU}(ii). \\[1mm]
\textsf{Continuity} & Makes choice probabilities respond continuously to the allocations, and so makes the cost's dependence on the ceiling continuous. \\[1mm]
\textsf{Normative Weak Order} & Makes the derived ranking $\succsim_n^\rho$ (Definition \ref{norm}) representable by a continuous index $\varphi$. \\[1mm]
\textsf{Selfishness-Acyclicity} & Lets the selfish utility $u$ be recovered from choice. \\[1mm]
\textsf{Normative Independence} & Forces the menu to enter through a single index, and makes that index cardinally comparable across menus of different ceilings. \\[1mm]
\textsf{Normative Selectivity} & Makes the randomization cost more curved at higher normative levels---Definition \ref{norm-APU}(iii). \\[1mm]
\textsf{Normative Richness} & Supplies, at or below any prescribed normative level, allocations attaining any selfish payoff. \\
\bottomrule
\end{tabular}
\end{center}

Three of the seven are not ours. \textsf{Positivity} and \textsf{Continuity} are the conditions \citet[Definitions 7 and 9]{FIS_2015} impose, and \textsf{Selfishness-Acyclicity} adapts their Item Acyclicity to a decision maker who cares about one coordinate of the allocation. \textsf{Normative Selectivity} adapts their notion of one perturbed utility being ``more selective'' than another (their Definition 12) from a comparison of two choice rules to a comparison of two normative levels within one rule. \textsf{Normative Weak Order} is a standard condition---completeness, transitivity and closed contour sets---but imposed on a relation that is itself derived from $\rho$, and that relation is new. \textsf{Normative Richness} is a richness condition stated relative to that derived order rather than to choice probabilities. \textsf{Normative Independence} has no antecedent we are aware of, and together with \textsf{Normative Selectivity} it is what the characterization runs on. We impose no structure on the cost function directly; its properties are derived from these conditions.

First, we consider \textit{Positivity}. This axiom states that all probabilities are positive. Even though \textit{positivity} seems to be an extreme assumption in social preferences, this axiom is helpful for the identification of the model.

Two remarks on its standing. Technically, \textsf{Positivity} is what makes the interior first-order conditions available, and every argument in this paper runs through them; it is the behavioral counterpart of the Inada condition in Definition \ref{norm-APU}(ii). Substantively, it says that the decision maker randomizes over the whole menu, including alternatives that are worse for her both selfishly and normatively, and it is fair to ask whether that describes anything. The evidence suggests it does. In the Probabilistic Dictator Game of \citet{FOOV_2025}, only about $17\%$ of dictators chose deterministically in all four decision problems, and $44.2\%$ placed positive weight on \emph{all three} allocations of the ternary menu---including the one that is worst for them on both dimensions. Full support is closer to a description of these data than to a technical convenience.
\begin{axiom}[\textsf{Positivity}]
\label{positivity}
For any menu $A \in \mathcal{A}$ and allocation ${\bm x} \in A$, $\rho({\bm x}, A) > 0$.
\end{axiom}

Next, we consider \textit{Continuity}, which guarantees that the choice probabilities respond continuously to the allocations.
\begin{axiom}[\textsf{Continuity}]
\label{continuity}
For any menu $\lbrace {\bm x}^1, \cdots, {\bm x}^m \rbrace$ with allocation sequences $\lim_{n \rightarrow \infty} {\bm x}^k_n = {\bm x}^k$ for each $k = 1, \cdots, m$,
	\begin{align*}
	\lim_{n \rightarrow \infty} \rho({\bm x}^k_n, \lbrace {\bm x}^1_n, \cdots, {\bm x}^m_n \rbrace) = \rho({\bm x}^k, \lbrace {\bm x}^1, \cdots, {\bm x}^m \rbrace).
	\end{align*}
\end{axiom}

We induce a personal-norm ranking over allocations, denoted by $\succsim_n^{\rho}$ on $X$, based on their impact on choice probabilities in larger menus.

\begin{definition}[\textsf{Personal-Norm Ranking}]
\label{norm}
For any ${\bm x}, {\bm y} \in X$, we say that ${\bm x}$ is normatively weakly better than ${\bm y}$, denoted by ${\bm x} \succsim_n^\rho {\bm y}$, if for every menu $A \in \mathcal{A}$ with ${\bm x}, {\bm y} \notin A$ and all ${\bm z}, {\bm w} \in A$ such that
\begin{align*}
    \rho({\bm w}, A \cup \lbrace {\bm x} \rbrace) \;=\; \rho({\bm w}, A \cup \lbrace {\bm y} \rbrace)
    \qquad\text{and}\qquad z_1 > w_1,
\end{align*}
it holds that
\begin{align*}
    \rho({\bm z}, A \cup \lbrace {\bm x} \rbrace) \leq \rho({\bm z}, A \cup \lbrace {\bm y} \rbrace).
\end{align*}
Strict normative preference (${\bm x} \succ_n^\rho {\bm y}$) and indifference (${\bm x} \sim_n^\rho {\bm y}$) are defined as the asymmetric and symmetric parts of $\succsim_n^\rho$, respectively.
\end{definition}

This definition behaviorally captures the idea that a normatively better allocation imposes a stronger penalty on the probabilities of choosing \emph{more selfish} alternatives in the same menu, \emph{relative to a calibration point}. The calibration hypothesis $\rho({\bm w},A\cup\{{\bm x}\})=\rho({\bm w},A\cup\{{\bm y}\})$ is what makes the induced ranking behave like a normative one, and it plays the same role in the compression property recorded as Definition \ref{ax:norm_penalty} in Appendix \ref{sec:necessity_norm}.

\begin{remark}[\textsf{Why the calibration is needed}]
\label{rmk:why_calibration}
The uncalibrated version of this definition---``${\bm x}\succsim_n^\rho{\bm y}$ if $\rho({\bm z},A\cup\{{\bm x}\})\le\rho({\bm z},A\cup\{{\bm y}\})$ for every $A$ and every ${\bm z}\in A$''---does not work, because adding an allocation to a menu suppresses the other alternatives through \emph{two} distinct channels: (i) the added allocation's own attractiveness, i.e., its selfish payoff, which draws probability away from every other alternative; and (ii) the normative ceiling it may impose, which compresses the distribution. These channels can point in opposite directions. Take ${\bm x}=(0,10)$ normatively better and ${\bm y}=(5,0)$ more selfish. Where the ceiling already exceeds both, only (i) operates and adding ${\bm y}$ suppresses others \emph{more}; where the ceiling is low, adding ${\bm x}$ raises it and (ii) dominates, reversing the comparison. Under the uncalibrated definition neither ${\bm x}\succsim_n^\rho{\bm y}$ nor ${\bm y}\succsim_n^\rho{\bm x}$ then holds, so the induced relation is incomplete and \textsf{Normative Weak Order} fails---not for exotic $\rho$, but for the canonical norm-dependent APU itself. Worse, since channel (ii) \emph{compresses} rather than uniformly suppresses---raising the probability of alternatives whose payoff is below the menu's calibration point---a definition quantified over \emph{all} ${\bm z}\in A$ cannot recover $\varphi$ at all; it can only recover the selfish order. Restricting attention to ${\bm z}$ that are more selfish than a calibrated reference ${\bm w}$ isolates channel (ii) exactly, which is what the definition above does.
\end{remark}

Based on this, we introduce the following elementary axioms to characterize the role of personal norms.

\begin{axiom}[\textsf{Normative Weak Order}]
\label{ax:norm_weak_order}
$\succsim_n^\rho$ is complete, transitive, and continuous: for every ${\bm y}\in X$, the sets $\{{\bm x}\in X: {\bm x}\succsim_n^\rho{\bm y}\}$ and $\{{\bm x}\in X: {\bm y}\succsim_n^\rho{\bm x}\}$ are closed in $X$.
\end{axiom}

Axiom \ref{ax:norm_weak_order} requires that the induced normative ranking is a continuous weak order. Completeness and transitivity are standard conditions in revealed preference analysis; continuity (closed contour sets) is the usual requirement for a \emph{continuous} numerical representation, and it is what Lemma \ref{norm_result} below uses. We state closedness explicitly rather than deriving it from \textsf{Continuity} (Axiom \ref{continuity}) because Definition \ref{norm} is an \emph{implication} whose hypothesis is a knife-edge equality: the set of ${\bm x}$ satisfying it is an intersection of sets of the form (calibration fails) $\cup$ (inequality holds), which need not be closed even when $\rho$ is continuous in the allocations. Under a representation the requirement is immediate, since $\succsim_n^\rho$ is then represented by the continuous function $\varphi$ (Appendix \ref{app:necessity}).

We emphasize that this property is not implied by Definition \ref{norm} alone; rather, it is imposed here to ensure that the normative ranking can be consistently represented by a real-valued index.

Moreover, we adapt \citet{FIS_2015}'s Definition 12 (``more selective'') and Proposition 3 from their setting---comparing two \emph{different} APU rules $P_1,P_2$ sharing a common utility function $u$---to ours, where there is a single $\rho$ but a family of cost functions $c(\cdot,t)$ indexed by the ceiling $t$, again sharing a common $u$ (Definition \ref{norm-APU} does not let $u$ depend on $t$). The role of ``rule 1 vs.\ rule 2'' is played by ``ceiling ${\bm b}$'s level vs.\ ceiling ${\bm a}$'s level,'' where ${\bm a},{\bm b}\in X$ are two allocations compared via $\succsim_n^\rho$---\emph{not} by two real numbers $t_A>t_B\in\varphi(X)$: unlike $\varphi$ itself, $\succsim_n^\rho$ is defined directly from $\rho$ (Definition \ref{norm}), so stating the comparison this way keeps both the Definition and the Axiom fully primitive, exactly as \textsf{Normative Penalty} (Definition \ref{ax:norm_penalty}) already is.

For $A\in\mathcal A$ and ${\bm x}\in A$, say ${\bm x}$ \emph{sets the ceiling of} $A$ if ${\bm x}\succsim_n^\rho {\bm y}$ for every ${\bm y}\in A$ (such an element exists and is unique up to $\sim_n^\rho$, since $\succsim_n^\rho$ is complete and transitive by \textsf{Normative Weak Order} and $A$ is finite).

\begin{definition}[\textsf{Selectivity Comparison}]
\label{def:selectivity}
Fix ${\bm a},{\bm b}\in X$ with ${\bm a}\succ_n^\rho{\bm b}$. Say $\rho$ is \emph{more selective at ${\bm b}$'s level than at ${\bm a}$'s level} if: for every ${\bm x},{\bm x}',{\bm y}\in X$ with $x_1<x_1'$, and every ${\bm z}_A,{\bm z}_A',{\bm z}_B,{\bm z}_B'\in X$ such that \[ A:=\{{\bm x},{\bm y},{\bm z}_A,{\bm a}\}, \qquad A':=\{{\bm x}',{\bm y},{\bm z}_A',{\bm a}\} \;\in\;\mathcal A \] (all with interior probabilities, and with ${\bm a}$ setting the ceiling of both), and \[ B:=\{{\bm x},{\bm y},{\bm z}_B,{\bm b}\}, \qquad B':=\{{\bm x}',{\bm y},{\bm z}_B',{\bm b}\} \;\in\;\mathcal A \] (with ${\bm b}$ setting the ceiling of both), satisfying \[ \rho({\bm y},A)=\rho({\bm y},A'), \qquad \rho({\bm y},B)=\rho({\bm y},B'), \qquad \rho({\bm x},A)=\rho({\bm x},B), \] it holds that \[ \rho({\bm x}',B')-\rho({\bm x},B) \;\ge\; \rho({\bm x}',A')-\rho({\bm x},A). \]
\end{definition}

The three matching conditions play the same role as in \citeauthor{FIS_2015}'s Definition 12: the first two hold ${\bm y}$'s own choice probability fixed when moving from the ``$x$''-menu to the ``$x'$''-menu, within each ceiling level separately (so that the change is attributable to $x$ alone, not to a change in the menu's overall pressure); the third calibrates the two ceiling levels to a common starting probability for ${\bm x}$, so that the two induced probability \emph{increments} (from ${\bm x}$ to ${\bm x}'$) are genuinely comparable across ${\bm a}$'s level and ${\bm b}$'s level.

\begin{axiom}[\textsf{Normative Selectivity}]
\label{ax:selectivity}
For every ${\bm a},{\bm b}\in X$ with ${\bm a}\succ_n^\rho{\bm b}$, $\rho$ is more selective at ${\bm b}$'s level than at ${\bm a}$'s level (Definition \ref{def:selectivity}).
\end{axiom}

In words: starting from a common baseline probability for ${\bm x}$, raising ${\bm x}$'s payoff to ${\bm x}'$ moves probability toward ${\bm x}'$ \emph{by at least as much} under the weaker normative environment (${\bm b}$'s level) as under the stronger one (${\bm a}$'s level)---stronger normative pressure dampens (weakly) how responsive choice is to payoff differences. This is stated purely in terms of $\mathcal A$, $\rho$, $\succsim_n^\rho$ (itself defined from $\rho$ alone, Definition \ref{norm}), and $x_1$: no reference to $u,\varphi,c,\lambda$ is needed to state it, and in particular \emph{not} to the real-valued index $\varphi$ that only exists once a representation has been constructed (Lemma \ref{norm_result}, Step 2 of Appendix \ref{app:sufficiency}).

Finally, we provide an acyclicity condition, called \textsf{Selfishness-Acyclicity}. This acyclic condition adapts \textit{Item Acyclicity} in \citet{FIS_2015} to our setting, requiring consistency between revealed choice-probability comparisons and the selfishness ranking $>$ on $x_1$.
\begin{axiom}[\textsf{Selfishness-Acyclicity}]
\label{ax:selfish_acyclicity}
There do not exist $n \geq 2$, allocations ${\bm x}^1, \ldots, {\bm x}^n \in X$, and menus $A^1, \ldots, A^n \in \mathcal{A}$ with ${\bm x}^k, {\bm x}^{k+1} \in A^k$ for each $k = 1, \ldots, n$ (indices modulo $n$, so that ${\bm x}^{n+1} := {\bm x}^1$), such that for every $k = 1, \ldots, n$, at least one of
\begin{align*}
	\rho({\bm x}^k, A^k) \geq \rho({\bm x}^{k+1}, A^{k}) \qquad \text{or} \qquad x^k_1 \geq x^{k+1}_1
\end{align*}
holds, and at least one of these $n$ inequalities is strict (i.e., $\rho({\bm x}^k, A^k) > \rho({\bm x}^{k+1}, A^{k})$ or $x^k_1 > x^{k+1}_1$ for some $k$).
\end{axiom}

Axiom \ref{ax:selfish_acyclicity} adapts \textit{Item Acyclicity} from the item-invariant APU framework of \citet{FIS_2014, FIS_2015} to our setting. There, Item Acyclicity requires that the union of the revealed choice-probability order and a given exogenous item order contains no cycle with at least one strict link; this is precisely the condition needed to invoke the Farkas-type Lemma \ref{apu_farkas} for a utility function consistent with \textit{both} orders simultaneously. Here, the role of the exogenous item order is played by the decision maker's selfish payoff ranking $x_1$. Consequently, applying Lemma \ref{apu_farkas} to the combined relation $\succsim \, := \, (\text{revealed order}) \cup (x_1\text{-order})$ directly yields a utility function $u$ that is simultaneously consistent with observed choice probabilities \textit{and} strictly increasing in $x_1$---i.e., a \textit{selfish} utility function---rather than merely the abstract, unlabeled ordinal index recovered in Lemma \ref{apu_menu} (see Claim \ref{claim:selfish_repn} in Appendix \ref{app:sufficiency}).

This condition ensures that the underlying selfish utility can be consistently recovered from observed stochastic choice behavior, providing the foundation for the baseline utility representation.

Recall from above that ${\bm a}\in A$ \emph{sets the ceiling of} $A$ if ${\bm a}\succsim_n^\rho{\bm y}$ for every ${\bm y}\in A$ (such an element exists and is unique up to $\sim_n^\rho$, since $\succsim_n^\rho$ is complete and transitive by \textsf{Normative Weak Order} and $A$ is finite). We say $A,B\in\mathcal A$ \emph{share a ceiling} if their ceiling-setting elements are $\sim_n^\rho$-equivalent. Our next axiom asks that choice probabilities, once passed through a rescaling tied to a menu's own ceiling-setting element, admit \emph{comparisons of common elements that are consistent across menus, whether or not those menus share a ceiling}---adapting the \textsf{Ordinal IIA} condition of \citet{FIS_2015} (their Definition 5, which characterizes their baseline, menu-\emph{independent} APU model) to our setting. As with \textsf{Normative Selectivity} above, the family of rescalings is indexed by \emph{allocations} ${\bm a}\in X$---via whichever allocation sets a given menu's ceiling---rather than by the index $\varphi$, which exists only once the axioms have been used to construct it; each rescaling depends on the ceiling-setter only through its $\sim_n^\rho$-class, so that ceiling-setters at the same level are assigned the same rescaling.

\begin{axiom}[\textsf{Normative Independence}]
\label{ax:norm_independence}
There exists a family $\{g_{\bm a}\}_{{\bm a}\in X}$ of $C^1$, strictly increasing functions $g_{\bm a}:(0,1)\to\mathbb R$, with $g_{\bm a}=g_{\bm b}$ whenever ${\bm a}\sim_n^\rho{\bm b}$ and with the boundary behavior
\begin{align*}
    \lim_{p\to 0^+} g_{\bm a}(p) = -\infty, \qquad \lim_{p\to 1^-} g_{\bm a}(p) = +\infty
    \qquad \text{for every } {\bm a}\in X,
\end{align*}
such that for any $A,B\in\mathcal A$, any ${\bm a}\in A$ setting the ceiling of $A$, any ${\bm b}\in B$ setting the ceiling of $B$, and any ${\bm x},{\bm y}\in A\cap B$,
\begin{align*}
    g_{\bm a}(\rho({\bm x},A)) - g_{\bm a}(\rho({\bm y},A)) \;=\; g_{\bm b}(\rho({\bm x},B)) - g_{\bm b}(\rho({\bm y},B)).
\end{align*}
\end{axiom}

The axiom is the ceiling-indexed counterpart of \citeauthor{FIS_2015}'s \textsf{Ordinal IIA}, and the modification that recovers theirs is a single clause: require $g_{\bm a}=g_{\bm b}$ for \emph{all} ${\bm a},{\bm b}\in X$ rather than only for $\sim_n^\rho$-equivalent ones. That single-$g$ version---one rescaling applying identically to every menu, with no ceiling structure---is equivalent to the fully menu-\emph{independent} APU representation \citep[Theorem 1]{FIS_2015}, so it would collapse the ceiling-dependent curvature that \textsf{Normative Selectivity} delivers. That is not what we impose: Axiom \ref{ax:norm_independence} keeps a separate rescaling for every $\sim_n^\rho$-class of ceiling-setting allocations. Appendix \ref{app:fis} spells out the difference.

The boundary requirement on $g_{\bm a}$ is the observable counterpart of the Inada-type divergence of $c_1$ at both endpoints in Definition \ref{norm-APU}(ii): under a representation, $g_{\bm a} = c_1(\cdot,\varphi({\bm a}))$, so the two conditions say the same thing. It is what guarantees, in the sufficiency direction, that the cost function built from $g_{\bm a}$ in Step 3 of Appendix \ref{app:sufficiency} blows up at both endpoints---nothing else in the axiom list delivers this---and, with Axiom \ref{richness_binary_below_ceiling} below, that the selfish utility constructed there has full range $\mathbb R$.

Axiom \ref{ax:norm_independence} requires no calibration hypothesis: unlike the compression property of Definition \ref{ax:norm_penalty} (Appendix \ref{sec:necessity_norm}), it does not single out a pre-matched reference allocation, and it applies to \emph{every} pair of common elements in \emph{every} pair of menus simultaneously, whether or not those menus share a ceiling. When $A,B$ share a ceiling, their ceiling-setters ${\bm a},{\bm b}$ satisfy ${\bm a}\sim_n^\rho{\bm b}$, so $g_{\bm a}=g_{\bm b}$ automatically and the statement reduces to a single shared rescaling. The content of the general statement is that, once probabilities are passed through the rescaling attached to a menu's own ceiling-setter, the resulting \emph{differences} between common elements no longer depend on which menu realized them---exactly the property one would expect if a single selfish utility $u$, shared across menus, governed the first-order conditions throughout $\mathcal A$, even though $c(\cdot,t)$ is free to vary in shape across ceilings.

The richness condition our axiom system does use is the following.\footnote{It is not \citeauthor{FIS_2015}'s \textsf{Richness} (Definition \ref{richness}), despite the name; Appendix \ref{app:fis} sets out the difference.} Definition \ref{richness} guarantees allocations realizing any target \emph{probability} within a menu, but several arguments below (Claim \ref{claim:well_defined_c}, Lemma \ref{lem:eps_perturbation}, and Proposition \ref{prop:curvature_char}) instead need an allocation realizing a target \emph{selfish payoff}, while keeping its normative ceiling below a prescribed bound. This is a statement about the normative order $\succsim_n^\rho$ directly, not about $\rho$'s realized probabilities, and does not follow from Definition \ref{richness} alone.
\begin{axiom}[\textsf{Normative Richness}]
\label{norm_richness}
\label{richness_below_ceiling}
\label{richness_binary_below_ceiling}
For every ${\bm a}\in X$:
\begin{enumerate}
    \item[(1)] for any $v \in \mathbb{R}$ and any finite $S\subset X$, there exists
    ${\bm w} \in X\setminus S$ with $w_1 = v$ and ${\bm a} \succsim_n^\rho {\bm w}$;
    \item[(2)] the choice in (1) can be made continuously: for any finite $S\subset X$ there is a
    continuous map $\zeta:\mathbb R\to X\setminus S$ with
    \begin{align*}
        \zeta(v)_1 = v \qquad\text{and}\qquad {\bm a}\succsim_n^\rho\zeta(v) \qquad\text{for every } v\in\mathbb{R};
    \end{align*}
    \item[(3)] for any $p\in(0,1)$ there exists ${\bm y}\in X$ with
    ${\bm a}\succsim_n^\rho{\bm y}$ and $\rho({\bm y},\{{\bm a},{\bm y}\}) = p$.
\end{enumerate}
\end{axiom}
That is, at every normative level (weakly below that of any given ${\bm a}$) there is an allocation attaining any prescribed selfish payoff $v$ exactly, one may always take it outside any prescribed finite set of allocations already in play, and these allocations may be chosen along a continuous curve. This captures the geometric fact that $X$ extends unboundedly in the selfish-payoff direction at every normative level; in the canonical model of Example \ref{ex:base_model} one may simply take $\zeta(v) := (v,c)$ for a sufficiently small constant $c$.

The three clauses serve different purposes. The plain existence statement, \emph{with} the finite-set avoidance, is what the necessity arguments use: menus are sets, so a construction that places several allocations of equal selfish payoff in the same menu (as in Lemma \ref{lem:calibrated_config} and the tie configurations of Remark \ref{rmk:scale_shortcut}) needs them to be genuinely distinct, and needs them distinct from the allocations already designated. The \emph{continuous} clause is used only in the sufficiency direction, and there only once---in the intermediate value argument of Proposition \ref{prop:curvature_char}, which moves an allocation continuously below a fixed ceiling until a prescribed probability is hit. (It reappears once more outside the proof of the theorem, in the redundancy discussion of Appendix \ref{app:independence}.) Necessity never needs it, because its intermediate value arguments run in the payoff variable rather than in the allocation space.

Clauses (1) and (2) control \emph{payoffs} below a prescribed ceiling; the arguments in Appendix \ref{app:sufficiency} additionally need control of \emph{probabilities} below a prescribed ceiling, in a \emph{binary} menu, and that is clause (3). It does not follow from the others: not from Definition \ref{richness} (whose target menus are ternary, with no control over their ceiling), and not from clause (1) (which fixes a payoff, not a probability).

In words: holding a designated allocation ${\bm a}$ fixed, every interior probability is realized by \emph{some} binary menu $\{{\bm a},{\bm y}\}$ whose second element is normatively no better than ${\bm a}$ (so that ${\bm a}$ sets the menu's ceiling). Under a representation this is exactly the statement that $u$ has full range on the allocations at or below ${\bm a}$'s normative level: by the interior first-order conditions, $\rho({\bm y},\{{\bm a},{\bm y}\})=p$ requires $u(y_1)-u(a_1) = c_1(p,t)-c_1(1-p,t)$, and as $p$ sweeps $(0,1)$ the right-hand side sweeps all of $\mathbb{R}$ by Definition \ref{norm-APU}(ii). In the sufficiency direction it is what makes the selfish utility constructed in Appendix \ref{app:sufficiency} surjective onto $\mathbb{R}$, hence continuous.\footnote{A property one might expect to find among the axioms is absent: that stronger normative pressure \emph{compresses} choice probabilities toward a calibration point, suppressing the alternatives more selfish than it. We do not assume it. Appendix \ref{sec:necessity_norm} states it (Definition \ref{ax:norm_penalty}) and shows that the axioms above deliver it (Proposition \ref{prop:penalty_implied}), so the compression property is a prediction of the model rather than an input to it.}

\subsection{Results}
\label{results}
We are ready to present the main results of this study.
\begin{theorem}
\label{representation_norm}
Suppose that $\rho$ satisfies \textit{Positivity}. Then, the following statements are equivalent:
	\begin{enumerate}
		\renewcommand{\labelenumi}{(\alph{enumi})}
			\item $\rho$ satisfies \textit{Continuity},
            \textit{Normative Weak Order},
            \textit{Selfishness-Acyclicity},
            \textit{Normative Independence},
            \textit{Normative Selectivity}, and \textit{Normative Richness}.
            
			\item $\rho$ is a \textit{norm-dependent additive perturbed utility}.
	\end{enumerate}
\end{theorem}

\paragraph*{The three steps of the proof.}
\label{outline_proof}
The argument runs in three steps; the details are in Appendix \ref{app:sufficiency}.

First, we establish the baseline representation. Using \textsf{Selfishness-Acyclicity} (Axiom \ref{ax:selfish_acyclicity}), we construct a selfish utility function $u$ that monotonically tracks the decision maker's material payoff $x_1$. We show that $u$ is continuous via \textsf{Continuity} (Axiom \ref{continuity}). Applying Proposition 8 of \citet{FIS_2014} (the working-paper precursor to \citealp{FIS_2015}; see Lemma \ref{apu_menu} below for our self-contained statement and proof), we initially obtain an unstructured, menu-dependent APU representation with a pair $(u, (C_A)_{A \in \mathcal{A}})$. For each menu $A$, the interior first-order conditions yield $u(x_1) - C'_A(\rho({\bm x}, A)) - \lambda_A = 0$, where $C_A$ acts as a menu-specific cost function.

Second, we construct the personal norm utility function $\varphi$. By \textsf{Normative Weak Order} (Axiom \ref{ax:norm_weak_order}), the behaviorally induced personal-norm ranking $\succsim_n^\rho$ is complete and transitive. Since the choice probabilities are continuous, this ranking preserves the topological structure and admits a continuous ordinal representation $\varphi : X \to \mathbb{R}$.

The third and most crucial step is to discipline the unstructured menu-dependent costs $(C_A)_{A \in \mathcal{A}}$ into the single-index parameterization $c(\cdot, \max_{{\bm z} \in A} \varphi({\bm z}))$. We utilize \textsf{Normative Independence} (Axiom \ref{ax:norm_independence}) to tightly connect the normative ceiling of a menu to the choice probabilities. Its cross-menu identities make the gap between two common elements depend only on the rescaling attached to the ceiling-setting allocation, and we prove that $C_A'(p) = C_B'(p)$ up to a constant whenever $\max_{{\bm y} \in A}\varphi({\bm y}) = \max_{{\bm z} \in B}\varphi({\bm z})$. This structurally guarantees that the norm-dependent cost function $c$ is perfectly well-defined.

Before turning to uniqueness, we record a concrete model, both to fix ideas and to establish that the axiom system is not vacuous.

\begin{example}[\textsf{A canonical norm-dependent APU}]
\label{ex:base_model}
Let $X=\mathbb{R}^2$, $u(x_1)=x_1$, and $\varphi({\bm x}) = x_2$. Let
\begin{align*}
    \kappa(p) \;:=\; p\ln p + (1-p)\ln(1-p), \qquad p\in(0,1),
\end{align*}
be the (negative) entropy kernel, and set $c(p,t) := \gamma(t)\,\kappa(p)$ for a continuous, \emph{strictly} increasing $\gamma:\mathbb{R}\to(0,\infty)$ (for instance $\gamma(t)=e^{t}$). Then $(u,\varphi,c)$ satisfies every requirement of Definition \ref{norm-APU}: $\kappa$ is $C^\infty$ and strictly convex on $(0,1)$ with $\kappa'(p) = \ln\frac{p}{1-p}$ strictly increasing from $-\infty$ to $+\infty$ and $\kappa''(p) = \frac{1}{p(1-p)}>0$, giving (i) and (ii); and $c_{11}(p,t)=\gamma(t)\kappa''(p)$ is strictly increasing in $t$, giving (iii). Condition (iv) holds because $D_{s,s'}=(\gamma(s')-\gamma(s))\kappa'$ with $\gamma(s')>\gamma(s)$, and $\kappa'$ is non-constant on $(0,\tfrac12)$. Note that $\kappa$ itself is finite (indeed $0$) at both endpoints---this is precisely why Definition \ref{norm-APU}(ii) is imposed on $c_1$ rather than on $c$. The induced $\rho$ is the interior maximizer of \eqref{model_norm}, characterized by the logistic form
\begin{align*}
    \rho({\bm x},A) \;=\; \sigma\!\left(\frac{x_1-\lambda_A}{\gamma(\varphi(A))}\right),
    \qquad \sigma(z):=\frac{1}{1+e^{-z}},
    \qquad \text{$\lambda_A$ the unique solution of } \sum_{{\bm v}\in A}\sigma\!\left(\frac{v_1-\lambda_A}{\gamma(\varphi(A))}\right)=1,
\end{align*}
in which $\gamma(\varphi(A))$ plays the role of a \emph{temperature} that rises with the menu's normative ceiling. It is worth resisting the temptation to write this as multinomial logit, because it is not: the perturbation in \eqref{model_norm} applies $\kappa$ to each probability separately, so the first-order condition is $x_1-\lambda_A=\gamma(\varphi(A))\ln\frac{\rho({\bm x},A)}{1-\rho({\bm x},A)}$---a logistic in the \emph{odds} of each allocation against its own complement, not a softmax over the menu. The softmax form would arise from the Shannon kernel $\sum_{\bm x}\rho({\bm x})\ln\rho({\bm x})$, whose marginal cost $\ln p+1$ stays bounded as $p\to1^-$ and therefore violates Definition \ref{norm-APU}(ii). By the Necessity part of Theorem \ref{representation_norm} (Appendix \ref{app:necessity}), this $\rho$ satisfies \emph{all} of the axioms in part (a); in particular the axiom system is consistent. Strict monotonicity of $\gamma$ is essential here and not merely convenient: if $\gamma$ were constant on an interval containing $x_2$ and $y_2$, the menus $A\cup\{{\bm x}\}$ and $A\cup\{{\bm y}\}$ would share a temperature, the calibration in Definition \ref{norm} would force $x_1=y_1$ and hence identical choice probabilities, and one would obtain ${\bm x}\sim_n^\rho{\bm y}$ with $x_2<y_2$---so $\varphi$ would fail to represent $\succsim_n^\rho$.
\end{example}

Nothing in this example is special to the entropy kernel. Let $c_1(p,t)=\gamma(t)\psi(p)$ for any $\psi$ with $\psi'>0$ continuous on $(0,1)$, antiderivative continuous on $[0,1]$ and $\psi(0^+)=-\infty$, $\psi(1^-)=+\infty$, and any continuous strictly increasing $\gamma:\mathbb R\to(0,\infty)$. Every requirement of Definition \ref{norm-APU} again holds---(iv) automatically, since $D_{s,s'}=(\gamma(s')-\gamma(s))\psi$ is strictly increasing---and so does every axiom of Theorem \ref{representation_norm}(a). The entropy kernel is one member of this family, so the axiom system is not tailored to logit; what it pins down is that the menu enters through a single increasing index, not the functional form of the cost.

Theorem \ref{unique_norm} below is stated for \emph{normalized} representations, and the next definition is what ``normalized'' means. It is needed because two normative levels can differ as numbers while being behaviorally indistinguishable: when that happens no condition stated in terms of $\rho$ can separate them, so $\varphi$ cannot be pinned down any more finely than up to the merging of such levels. Normalization removes exactly that slack, and nothing else.

\begin{definition}[\textsf{Level Separation}]
\label{def:level_separation}
Levels $s<s'$ in $\varphi(X)$ are \emph{indistinguishable} if $D_{s,s'}:=c_1(\cdot,s')-c_1(\cdot,s)$ is constant on $(0,1)$, and \emph{separated} if $D_{s,s'}$ is non-constant on $(0,\tfrac12)$. A representation $(u,\varphi,c)$ is \emph{normalized} if every pair of distinct levels in $\varphi(X)$ is separated.
\end{definition}
Since $D_{s,s'}$ is non-decreasing by Definition \ref{norm-APU}(iii), the two cases are mutually exclusive, and condition \ref{norm-APU}(iv) says that no third case occurs. For a \emph{normalized} representation the first alternative is vacuous, so (iv) reduces to the requirement that $D_{s,s'}$ be non-constant on $(0,\tfrac12)$ for every pair of distinct levels; this is the form used in the proofs. The dichotomy has to be imposed rather than derived---the excluded third case is a pair of levels that can neither be merged nor separated---and Appendix \ref{app:iv_free} sets out why, together with Example \ref{ex:eta_failure}, which shows that the case is not idle: there a cost family satisfying (i)--(iii) induces an intransitive $\succsim_n^\rho$, so it is a norm-dependent APU in the sense of (i)--(iii) whose behavior violates \textsf{Normative Weak Order}. Condition (iv) removes that case, and with it every representation can be normalized by merging; the uniqueness statement below is for normalized representations.

\begin{remark}[\textsf{Condition \ref{norm-APU}(iv) is not an extra restriction}]
\label{rmk:iv_free}
Condition (iv) narrows the class of admissible cost families, but it does not narrow the class of \emph{behaviors} characterized in Theorem \ref{representation_norm}: if $\rho$ satisfies the axioms of part (a) and $(u,\varphi,c)$ satisfies Definition \ref{norm-APU}(i)--(iii) and represents $\rho$, then $(u,\varphi,c)$ automatically satisfies (iv). The argument is given in Appendix \ref{app:iv_free}; it uses one richness proviso, automatic in the canonical setting of Example \ref{ex:base_model} and in any $X$ in which $\varphi$ depends only on the non-selfish coordinates. So (iv) excludes exactly those tuples $(u,\varphi,c)$ whose induced behavior violates the axioms, and nothing else.
\end{remark}

Next, we state the uniqueness result.
\begin{theorem}
\label{unique_norm}
Suppose $(u, \varphi, c)$ and $(\hat{u}, \hat{\varphi}, \hat{c})$ both represent the same stochastic choice function $\rho$. Then, the following conditions hold:
\begin{enumerate}
    \item[(i)] There exist a real number $\alpha > 0$ and a constant $\beta_u \in \mathbb{R}$ such that:
    \begin{equation*}
        \hat{u}(x) = \alpha u(x) + \beta_u \quad \text{for all } x \in X.
    \end{equation*}
    \item[(ii)] There exists a strictly increasing continuous function $f: \varphi(X) \to \mathbb{R}$ such that:
    \begin{equation*}
        \hat{\varphi}(x) = f(\varphi(x)) \quad \text{for all } x \in X.
    \end{equation*}
    \item[(iii)] There exist functions $\gamma(t)$ and $\delta(t)$ such that for all $p \in (0,1)$ and all $t \in \hat{\varphi}(X)$:
    \begin{equation*}
        \hat{c}(p, t) = \alpha c(p, f^{-1}(t)) + \gamma(t)p + \delta(t).
    \end{equation*}
\end{enumerate}
\end{theorem}

\begin{remark}
Unlike the selfish utility function, which is cardinally identified up to a positive affine transformation, the personal norm function is only ordinally identified. Because our representation depends on $\varphi$ exclusively through the normative ceiling (a maximum operator), ordinal transformations perfectly preserve the identity of the normatively maximal element. Restricting $\varphi$ to a parametric family does not change this. Such a family may pin down its parameters---the exponent of a Nash product is determined by the indifference curves it generates---but it pins them down from the \emph{ordering}, and the numerical levels remain without content because $c(\cdot,t)$ absorbs any relabelling of them. Cardinal content would require restricting $c$ instead: in the scale form $c_1(p,t)=\gamma(t)\psi(p)$ with $\gamma$ fixed rather than free, relabelling the levels changes $\gamma$ and is therefore detectable.
\end{remark}

\begin{proposition}
\label{prop:marginal_cost}
Suppose that $\rho$ satisfies the axioms in Theorem \ref{representation_norm}. Then, the curvature of the randomization cost, $c_{pp}(p,t) = c_{11}(p,t)$, is non-decreasing in the normative ceiling $t$: for any $p \in (0,1)$ and any $t_A > t_B$ in $\varphi(X)$, $c_{pp}(p,t_A) \geq c_{pp}(p,t_B)$.
\end{proposition}

Proposition \ref{prop:marginal_cost} formally captures the behavioral implication of \textsf{Normative Selectivity} (Axiom \ref{ax:selectivity}), which is the axiom that delivers it. The proof is given in Appendix \ref{app:closing_marginal_cost} (Proposition \ref{prop:curvature_char} and Corollary \ref{cor:prop_marginal_cost_curvature}).

\begin{remark}
Proposition \ref{prop:marginal_cost} pins down the \emph{curvature} of the randomization cost, not its \emph{level}: the statement ``$c_p(p,t_A) \geq c_p(p,t_B)$ for every $p$'' is generally \emph{not} implied by the axioms, since, by Theorem \ref{unique_norm}(iii), the level of $c_p(\cdot,t)$ can be shifted by an arbitrary function of $t$ alone without changing $\rho$ at all, whereas the curvature $c_{pp}(\cdot,t)$ cannot be (see Remark \ref{rmk:level_corollary} in the Appendix for the precise obstruction). Curvature monotonicity is nonetheless enough to drive the paper's comparative statics: for instance, under a multiplicatively separable cost structure $c(p,t) = a(t)g(p)$ with $g$ strictly convex, Proposition \ref{prop:marginal_cost} gives $a(t)$ non-decreasing in $t$, so that $\rho$ becomes uniformly less sensitive to payoff differences as the normative ceiling rises, in the sense used throughout Sections \ref{intro} and \ref{discussion}.
\end{remark}

\section{Discussion}
\label{discussion}

\subsection{Menu Effects and Regularity}
\label{sec:regularity}
The model's most directly testable implication concerns \emph{Regularity}, the property that enlarging a menu cannot raise a choice probability.

\begin{definition}[\textsf{Regularity}]
\label{def:regularity}
$\rho$ satisfies \emph{Regularity} if $\rho({\bm x},B)\le\rho({\bm x},A)$ whenever $A\subseteq B$ and ${\bm x}\in A$.
\end{definition}

In the model of \citet{CDOR_2019}, stochasticity reflects a complementarity among the elements of the menu and therefore necessarily entails violations of Regularity. In our model, violations are not automatic: they are tied to movements of the normative ceiling, and when they occur the model says which alternatives they hit. The following statement is for a single added alternative, which suffices, since Regularity for arbitrary $A\subseteq B$ follows by adding elements one at a time.

\begin{proposition}[\textsf{Regularity and the normative ceiling}]
\label{prop:regularity}
Let $\rho$ be a norm-dependent additive perturbed utility with representation $(u,\varphi,c)$, let $A\in\mathcal A$ with $|A|\ge2$ and ${\bm z}\in X\setminus A$, and put $B:=A\cup\{{\bm z}\}$, $t:=\varphi(A)$, $t':=\varphi(B)$, $D:=c_1(\cdot,t')-c_1(\cdot,t)$ and $\Lambda:=\lambda_A-\lambda_B$, where $\lambda_A$ and $\lambda_B$ are the multipliers of the two budget constraints. Then:
\begin{enumerate}
    \item[(i)] \emph{(Criterion.)} For every ${\bm x}\in A$,
    \begin{equation*}
        \rho({\bm x},B)>\rho({\bm x},A) \iff D\bigl(\rho({\bm x},A)\bigr)<\Lambda .
    \end{equation*}
    
    \item[(ii)] \emph{(Who violates.)} Call ${\bm x}\in A$ a violator at this expansion if $\rho({\bm x},B)>\rho({\bm x},A)$. The violators form a lower set in the selfish payoff: if ${\bm x}$ is a violator and ${\bm y}\in A$ has $y_1\le x_1$, then ${\bm y}$ is a violator too. At least one element of $A$ is not a violator, and the most selfish element of $A$ never is. Quantitatively, with $\pi^\ast:=\sup\{p\in(0,1): D(p)<\Lambda\}$ (and $\pi^\ast:=0$ when that set is empty), ${\bm x}$ is a violator if $\rho({\bm x},A)<\pi^\ast$ and is not one if $\rho({\bm x},A)>\pi^\ast$.
    
    \item[(iii)] \emph{(No ceiling rise, no violation.)} If $\varphi({\bm z})\le\varphi(A)$, then
    $\rho({\bm x},B)<\rho({\bm x},A)$ for every ${\bm x}\in A$; Regularity holds strictly.
    
    \item[(iv)] \emph{(Indistinguishable levels, no violation.)} If the levels $t\le t'$ are
    indistinguishable in the sense of Definition \ref{def:level_separation}, then
    $\rho({\bm x},B)\le\rho({\bm x},A)$ for every ${\bm x}\in A$.
\end{enumerate}
\end{proposition}
A violator is defined relative to a single pair $({\bm x},A\subset B)$ at which Regularity fails; Regularity itself is a property of $\rho$, quantified over all such pairs.

The proof is in Appendix \ref{app:regularity}; part (i) is two lines of first-order conditions, and the rest follows from the monotonicity of $D$ and a counting argument. We add two remarks. The first concerns part (ii): it is stated without any hypothesis on the ceilings, because the threshold structure needs only the monotonicity of $D$; it is parts (iii) and (iv) that say when the violator set is \emph{empty}, so that violations require both a strict rise in the ceiling and behaviorally distinguishable levels. The second concerns the lower-set formulation: it is the robust one, because an element sitting exactly at $\rho({\bm x},A)=\pi^\ast$ may fall on either side depending on whether the supremum is attained, so a payoff threshold $v^\ast$ with ``violators $=\{x_1<v^\ast\}$’’ exists but is not canonically pinned at the boundary.

Part (ii) is the economically substantive statement: when an added alternative raises the normative ceiling, any gain in choice probability accrues to the less selfish incumbents, and the more selfish ones lose. The violation is not a diffuse menu effect but a compression around a threshold, and both the direction and the identity of the affected alternatives are pinned by observables. This is exactly the pattern described informally for the \$10 dictator game in Section \ref{intro}, now with the mechanism isolated: the ceiling rises, the cost of randomization becomes more curved, and curvature pulls probability from the likely to the unlikely.

Notice that the criterion in part (i) is the \emph{same} inequality that governs the calibrated personal-norm ranking. In the proof of Lemma \ref{lem:calibrated_config} an element of a common menu gains probability under the higher ceiling precisely when $D$ evaluated at its share falls below the multiplier gap, and $\Lambda$ there plays the role it plays here. Regularity violations and the identification of $\varphi$ are thus two readings of one piece of algebra: the comparison of $D$ with a multiplier gap at a given share. Nothing further is needed to connect them.

The example of the Introduction is an instance; its threshold can be read off directly. Take the canonical model of Example \ref{ex:base_model} with $\gamma(t)=e^{t/5}$, let $A=A_2=\{(5,5),(10,0)\}$ and adjoin $(0,10)$ to obtain $B=A_1$. The ceiling rises from $5$ to $10$, and
\begin{equation*}
    \rho\bigl((5,5),\cdot\bigr): 0.285 \longrightarrow 0.322,
    \qquad
    \rho\bigl((10,0),\cdot\bigr): 0.715 \longrightarrow 0.483 ,
\end{equation*}
with $(0,10)$ itself taking $0.195$. The less selfish incumbent gains and the more selfish one loses, as part (ii) requires. Here $D=(e^{2}-e)\,\kappa'$ and $\Lambda=\lambda_A-\lambda_B\approx -2.996$, so the threshold is
\begin{equation*}
    \pi^\ast \;=\; D^{-1}(\Lambda) \;\approx\; 0.345 ,
\end{equation*}
and indeed $\rho((5,5),A)=0.285<\pi^\ast<0.715=\rho((10,0),A)$. The threshold sits just above the uniform share $\tfrac13$ of $B$, as one would expect in this regime: as $\gamma(t')\to\infty$ compression drives every probability in $B$ toward $1/|B|$, and the alternatives that gain are those whose original share lay below that value.

\paragraph*{A test, and a design that separates the mechanism.} The proposition is restrictive enough to be tested on nested menus without estimating anything. In the Probabilistic Dictator Game of \citet{FOOV_2025}, subjects assign weights to allocations of \$10 over the four menus $A_1,\dots,A_4$ of Section \ref{intro}, so that $A_2,A_3,A_4$ are the three binary subsets of $A_1$. Adding the missing alternative raises the normative ceiling in exactly one of the three expansions---namely the one that restores the $\varphi$-maximal element of $A_1$---so by part (iii) at most one of the three comparisons can exhibit a violation at all, and by part (ii) the violator there is the less selfish incumbent. We do not need to know the shape of $\varphi$ to say this, and the three cases are informative precisely because they differ:
\begin{table}[htbp]
\centering
\begin{tabular}{lll}
\toprule
$\varphi$-maximal allocation & Violation possible only in & and only at \\
\midrule
$(0,10)$ \;(norm increasing in giving) & $A_2\subset A_1$ & $(5,5)$ \\
$(5,5)$ \;(inequity-averse norm) & $A_3\subset A_1$ & $(0,10)$ \\
$(10,0)$ \;(norm favoring the selfish act) & $A_4\subset A_1$ & $(0,10)$ \\
\bottomrule
\end{tabular}
\caption{Which nested comparison can exhibit a Regularity violation, and at which incumbent, for each of the three candidate normative rankings over \citeauthor{FOOV_2025}'s menus. The pattern is a sign restriction with no free parameters.}
\label{tab:norm_test}
\end{table}
The model therefore permits a violation in one of three nested comparisons and at one of the two incumbents there, with no free parameters---and \emph{which} comparison exhibits it identifies which allocation carries the highest personal norm. This identification shows that the framework does not presuppose a social preference: $\varphi$ is recovered from behavior, it is unrestricted beyond continuity, and an inequity-averse norm peaked at the equal split is as admissible as a norm increasing in what the recipient gets. The pattern of Regularity violations is what tells the two apart. \citet{FOOV_2025} report an aggregate violation rate of $42\%$; which comparison and which allocation those violations occur at is exactly what the proposition restricts, and is testable on their existing data.

One competing reading should be separated out. \citet{FOOV_2025} suggest informally that randomization may be cheaper in $A_2$ than in $A_3$ or $A_4$ because any mixture in $A_2$ still guarantees the dictator at least \$5, whereas mixtures in $A_3$ and $A_4$ risk a payoff of zero. That is also a menu-level scalar, but a different one: the decision maker's own downside exposure rather than the normative ceiling. In this design the two indices happen to rank the menus identically, so the design cannot tell them apart. They can be separated by raising the recipient's maximum while holding the dictator's own payoff range fixed---for instance by expanding $\{(5,0),(8,2)\}$ with $(5,10)$, which lifts the normative ceiling while leaving the dictator's worst case at \$5 in both menus. The ceiling reading permits the less selfish incumbent to gain; the downside-exposure reading predicts no menu effect at all. Menus of this kind require an endowment that is not fixed across allocations, which our domain permits.

\subsection{Social Preferences: What the Framework Does and Does Not Include}
\label{sec:social_scope}
We discuss how our model relates to the social-preference literature, and in particular whether inequality aversion and ex-ante fairness are special cases. A distributional concern could enter in three places, and the answer is different in each. One distinction has to be made first, because it is the one most easily lost.
\begin{table}[htbp]
\centering
\begin{tabular}{p{55mm}p{38mm}l}
\toprule
Where a distributional concern enters & Status & Where shown \\
\midrule
In the norm $\varphi$ (e.g., an inequality-averse norm) & Included and identified & Table \ref{tab:norm_test} \\
In the selfish utility $u$ & Not identified from $\rho$ & --- \\
In the evaluation of the lottery (ex-ante fairness) & Excluded & Proposition \ref{prop:exante} \\
\bottomrule
\end{tabular}
\caption{Three places a distributional concern could enter the model, and the status of each.}
\label{tab:scope}
\end{table}

\paragraph*{Inequality aversion in the norm: included.} The normative ceiling $\varphi(A)=\max_{{\bm x}\in A}\varphi({\bm x})$ is a property of the \emph{menu}, not of the alternative chosen from it. So when $\varphi$ is peaked at the equal split, what the model says is that \emph{the presence of a fair option on the menu raises the cost of committing}---not that the decision maker dislikes unequal outcomes. An inequality-averse utility in the sense of \citet{FS_1999} says the latter: it evaluates each allocation and enters the maximand. Here the maximand is $u(x_1)$, and $\varphi$ enters only the curvature of the randomization cost. The two readings are separated by nested menus, because only the former predicts that an alternative that is \emph{never chosen} changes behavior merely by being available. This is why we speak of a personal norm rather than a social preference: the object is a standard the decision maker feels bound by, not a taste over distributions.

\paragraph*{Which allocation carries the norm: identified.} The framework places no restriction on $\varphi$ beyond continuity, so a norm peaked at the equal split is exactly as admissible as one increasing in what the recipient receives, or one favoring the selfish act. Which allocation carries the highest norm is not a modeling choice but an identified feature of behavior: Table \ref{tab:norm_test} shows that the \emph{pattern} of Regularity violations across the nested menus of \citet{FOOV_2025} pins it down, with no free parameters. A reader who expects a social-preference model to commit to inequality aversion at the outset will find that commitment absent here, and that is deliberate.

\paragraph*{Inequality aversion in the utility: not identified.} A natural extension replaces the selfish utility $u(x_1)$ with an inequality-averse utility of the whole allocation; this changes $u$ but not the mechanism, and Section \ref{axioms} discusses it. The obstacle is that the calibrated comparisons of Definition \ref{norm} rest on the exogenously observable comparison $z_1>w_1$; a general social utility on $X$ would leave $u$ and $\varphi$ as two unrestricted functions on the same domain, and nothing in $\rho$ would separate them. Extending $u$ is the natural next step, but it is a question about identification, not about generality.

\paragraph*{Ex-ante fairness: excluded, and separated by the same test.} A decision maker who randomizes in order to equalize the recipient's \emph{opportunity} evaluates the lottery through the distribution it induces, which adds to the objective a term $w$ evaluated at the induced mean allocation. Such a term cannot be absorbed into the randomization cost. A recent generalization sharpens this claim about scope rather than threatening it. \citet{T_2025} axiomatizes preferences in which the weight on ex-ante relative to ex-post fairness varies across recipients, and shows that the resulting behavior cannot be accommodated by \citeauthor{S_2013}'s model. That generalization concerns \emph{several} recipients, however; with a single recipient the two representations coincide \citep[footnote 2]{T_2025}, and a single recipient is our setting.

Let us call $w$ \emph{absorbable} if there exist a selfish utility $\tilde u$ and a kernel $\tilde c$, and write $g_t:=c(\cdot,t)-\tilde c(\cdot,t)$ (a symbol local to this subsection and to Appendix \ref{app:exante}: it is not the rescaling $g_{\bm a}$ of Axiom \ref{ax:norm_independence}), such that
\begin{equation}
\label{eq:absorb}
    w\Bigl(\sum_{{\bm x}\in A}\rho({\bm x}){\bm x}\Bigr)
    = \sum_{{\bm x}\in A}\bigl[\tilde u(x_1)-u(x_1)\bigr]\rho({\bm x})
      + \sum_{{\bm x}\in A}\bigl[c(\rho({\bm x}),\varphi(A))-\tilde c(\rho({\bm x}),\varphi(A))\bigr]
\end{equation}
for every $A\in\mathcal A$ and every $\rho\in\Delta(A)$. 

\begin{proposition}[\textsf{Ex-ante fairness is not absorbable}]
\label{prop:exante}
Let $u$, $\varphi$ and $c$ be as in Definition \ref{norm-APU} and let $w:\mathbb R^2\to\mathbb R$ be continuous. Suppose there are ${\bm x}_0\in X$, a unit vector ${\bm d}$ and $\bar s>0$ such that
\begin{equation}
\label{eq:exante_ray}
    {\bm x}_0-s{\bm d}\in X
    \quad\text{and}\quad
    \varphi({\bm x}_0-s{\bm d})<\varphi({\bm x}_0)
    \qquad\text{for every } s\in(0,\bar s],
\end{equation}
and put $t:=\varphi({\bm x}_0)$, so that every binary menu $\{{\bm x}_0,{\bm x}_0-s{\bm d}\}$ has normative ceiling $t$. If $w$ is absorbable, then $w$ is affine along the whole segment from ${\bm x}_0$ to ${\bm x}_0-\bar s{\bm d}$---equivalently, when $w$ is twice continuously differentiable,
\begin{equation*}
    {\bm d}^{\top}H_w({\bm x}_0-r{\bm d}){\bm d}=0
    \qquad\text{for every } r\in[0,\bar s],
\end{equation*}
where $H_w$ denotes the Hessian of $w$. Moreover $\Gamma_t:=g_t(\cdot)+g_t(1-\cdot)$ is affine on $(0,1)$, so on binary menus the ex-ante term is absorbed into the selfish utility and leaves the randomization cost unchanged. Finally, if $w$ is twice continuously differentiable:
\begin{enumerate}
    \item[(a)] if \eqref{eq:exante_ray} holds for three pairwise non-parallel directions, then
    $H_w({\bm x}_0)={\bm 0}$; and
    
    \item[(b)] if it holds for \emph{every} direction with a common $\bar s>0$---as it does when
    ${\bm x}_0$ maximizes $\varphi$ strictly on a ball $B({\bm x}_0,\bar s)\subseteq X$---then $w$
    is affine on $B({\bm x}_0,\bar s)$.
\end{enumerate}
\end{proposition}

The proof is in Appendix \ref{app:exante}; the main clause uses no differentiability, only the continuity of $w$ along the segment. That matters here, because the leading ex-ante criteria are piecewise linear rather than smooth: \citeauthor{S_2013}'s is of the \citet{FS_1999} form, and the extension he gives admits any concave, constant-linear criterion---\citet{CR_2002}'s among them. Part (a) cannot be strengthened to a statement about a neighborhood of ${\bm x}_0$: three rays constrain the Hessian only at the anchor. For instance $w(x,y)=x^{2}y^{2}(x-y)^{2}$ satisfies the ray conditions identically along $(1,0)$, $(0,1)$ and $(1,1)$ from ${\bm x}_0={\bm 0}$, and $H_w({\bm 0})={\bm 0}$, yet $H_w(t,2t)\neq{\bm 0}$ for \emph{every} $t\neq0$---its $(1,1)$ entry is $-16t^{4}$---so the Hessian fails to vanish at points arbitrarily close to the anchor. This is why (b) quantifies over all directions and is proved without reference to the Hessian.

We add three remarks on the statement. The first concerns the hypothesis. Condition \eqref{eq:exante_ray} is a richness condition on $X$ and not a restriction on $\varphi$: it asks only that some allocation be $\varphi$-maximal along a ray of allocations lying strictly below its own ceiling, and it costs nothing beyond the non-degeneracy the comparison already presupposes. Take any ${\bm a}\succ_n^\rho{\bm y}$, which exist unless $\varphi$ is constant and every pair of levels is indistinguishable. The segment $[{\bm a},{\bm y}]$ lies in $X$ because $X$ is convex, and $\varphi$ is continuous on it by the closedness clause of \textsf{Normative Weak Order}, so $\varphi$ attains a maximum there; let ${\bm x}_0$ be the maximizer nearest ${\bm y}$ and put ${\bm d}:=({\bm x}_0-{\bm y})/\lVert{\bm x}_0-{\bm y}\rVert$, $\bar s:=\lVert{\bm x}_0-{\bm y}\rVert$. Every point strictly between ${\bm x}_0$ and ${\bm y}$ has $\varphi$ strictly below $\varphi({\bm x}_0)$---otherwise it would be a maximizer nearer ${\bm y}$---and so does ${\bm y}$, which is \eqref{eq:exante_ray}. What \textsf{Normative Richness}~(1) adds is that such pairs remain available at any prescribed own payoff. The form of the hypothesis matters. A natural alternative, comparing two menus that share a midpoint with one twice the spread of the other, would give a shorter route to the same conclusion but cannot be satisfied: for the canonical $\varphi({\bm x})=x_2$ of Example \ref{ex:base_model} the two ceilings are $m_2+|d_2|/2$ and $m_2+|d_2|$, equal only if $d_2=0$, and in a constant-sum environment every spread is parallel to $(1,-1)$. Anchoring at a $\varphi$-maximal element instead makes the common ceiling automatic.

The second concerns the constant-sum environment, which is the one the \$10 dictator game and \citet{FOOV_2025}’s design live in. There the allocations the design uses lie on a line, so a single direction ${\bm d}$ is available and the three-direction clause is unusable; but taking ${\bm x}_0$ to be the $\varphi$-maximal allocation of the design and ${\bm d}$ the direction pointing into the rest of the line, \eqref{eq:exante_ray} holds and $w$ is affine along the whole segment of induced means. One direction is all the model ever evaluates there, so this is the full conclusion and not a weakening of it. Parts (a) and (b) are correspondingly unusable in that environment---$X$ lies on a line and contains no ball, and all directions in it are parallel---which is why the single-direction statement is the one the proposition leads with.

The third concerns fixed menus: no fixed-menu argument can do the work, which is why the quantification over scales matters. Within a \emph{single} binary menu an ex-ante term is often absorbable: for $A=\{(0,10),(10,0)\}$ and $w({\bm m})=m_1^2$, the choice $g_t(p)=50p^2$ leaves the residual $100p-50$, which is affine in $p$. What rules absorbability out is the requirement that one kernel and one $w$ serve all menus at a given ceiling.

The sharper point is that the two motives are separated \emph{observationally}, and by part (ii) of Proposition \ref{prop:regularity} itself. We compare against the hybrid in which an ex-ante fairness term is added to a perturbed objective, since a pure ex-ante fairness model randomizes over at most two allocations and is on that ground already rejected by \citet{FOOV_2025}. Take $u(x_1)=x_1$, the kernel $c(p,t)=\gamma(t)\,\kappa(p)$ of Example \ref{ex:base_model} with $\gamma\equiv\tfrac12$, and $w({\bm m})=-\tfrac{\nu}{2}(m_1-5)^2$ with $\nu=\tfrac32$---a decision maker whose ex-ante ideal is exactly the equal split. Maximizing the resulting strictly concave objective over $\Delta(A)$ on the two menus $A_2\subset A_1$ above gives
\begin{table}[htbp]
\centering
\begin{tabular}{lccc}
\toprule
 & $(0,10)$ & $(5,5)$ & $(10,0)$ \\
\midrule
$\rho(\cdot,A_2)$ & --- & $0.825$ & $0.175$ \\
$\rho(\cdot,A_1)$ & $0.270$ & $0.331$ & $0.399$ \\
\bottomrule
\end{tabular}
\caption{Choice probabilities for the hybrid with an ex-ante fairness term, $u(x_1)=x_1$, $c(p,t)=\gamma(t)\kappa(p)$ with $\gamma\equiv\tfrac12$, and $w({\bm m})=-\tfrac{\nu}{2}(m_1-5)^2$ with $\nu=\tfrac32$. The violator set is $\{(10,0)\}$, an upper set in the selfish payoff.}
\label{tab:exante_numbers}
\end{table}
Adjoining $(0,10)$ to $A_2$ therefore raises the probability of the \emph{most selfish} allocation from $0.175$ to $0.399$, while the equal split falls from $0.825$ to $0.331$. The violator set is $\{(10,0)\}$---an upper set in the selfish payoff, not a lower one---and part (ii) of Proposition \ref{prop:regularity} rules exactly this out for \emph{every} norm-dependent additive perturbed utility. The mechanism is transparent: in $A_2$ the only way to be ex-ante fair is to put weight on $(5,5)$, whereas once $(0,10)$ is available a mixture of the two extremes delivers the same ex-ante allocation at a higher private payoff, so the dictator takes the selfish option far more often. Nor is the pattern a knife-edge: it obtains throughout a region of the parameter space, while for a weak enough fairness concern the same specification satisfies Regularity, so the test has content in both directions.
Nor is the \emph{observational} separation confined to smooth criteria, as Proposition \ref{prop:exante} is not: take \citeauthor{S_2013}'s own piecewise-linear criterion, which on the constant-sum line has a single kink, at the equal split. Keep $u(x_1)=x_1$ as the ex-post term, and write $\delta$ for the weight on the ex-ante term, $\beta$ for the guilt parameter, and $\theta:=1+\delta(1-2\beta)$ for the effective slope above the kink. Every attainable mean in $A_2$ lies above the kink, so $\rho\bigl((10,0),A_2\bigr)=\sigma\bigl(5\theta/2\gamma\bigr)$; in $A_1$ the optimum is pinned \emph{at} the kink whenever $\theta\le0$, and the maximum-entropy law with mean $5$ on $\{0,5,10\}$ is uniform, so $\rho\bigl((10,0),A_1\bigr)=\tfrac13$. The violator set is therefore $\{(10,0)\}$ again---an upper set---whenever $\beta>\tfrac12$ and $\delta(2\beta-1)>1+\tfrac{2\ln 2}{5}\gamma$; and $\beta>\tfrac12$ is the very threshold by which \citet{FOOV_2025} classify \citeauthor{S_2013} types in $A_2$.

The implication runs one way only. A weak enough ex-ante fairness concern generates data on $A_1,\dots,A_4$ that a norm-dependent additive perturbed utility can match with a suitably chosen kernel; absorbability of the objective and observational equivalence of the choice function are different questions. Proposition \ref{prop:exante} says that the two families do not nest as representations. Proposition \ref{prop:regularity}(ii) refutes the ex-ante fairness reading whenever the observed violator set fails to be a lower set in the selfish payoff, and is silent otherwise.

\section{Conclusion}
\label{conclusion}
We have developed a decision-theoretic framework, the norm-dependent additive perturbed utility model, to explain why fundamentally selfish decision makers deliberately randomize in social situations. The menu-dependent model of \citet{FIS_2015} leaves the dependence of randomization costs on the menu unrestricted; by imposing that cross-menu comparisons line up through a single rescaling attached to each menu's most normative alternative (Axiom \ref{ax:norm_independence}) and that higher normative levels are more selective (Axiom \ref{ax:selectivity}), on top of a normative weak order (Axiom \ref{ax:norm_weak_order}), we have shown that menu-dependent stochastic choice can be characterized by a single normative index---the normative ceiling of the menu. The technical richness condition that \citet{FIS_2015} assume is then implied rather than imposed (Proposition \ref{prop:richness_redundant}), so the characterization rests on seven conditions, one of which---\textsf{Positivity}---is maintained throughout rather than being part of the equivalence. The compression property is likewise implied rather than imposed (Proposition \ref{prop:penalty_implied}).

The structure pays off in identification and in testable implications. The curvature of the randomization cost rises with the normative ceiling, so choice probabilities become less responsive to payoff differences; it is the curvature, and not the level, of the marginal cost that behavior identifies, and the personal norm function is ordinally unique (Theorem \ref{unique_norm}). The sharpest consequence concerns menu effects. Proposition \ref{prop:regularity} shows that Regularity can fail only when the normative ceiling rises across behaviorally distinguishable levels, and that any resulting gain in choice probability accrues to the less selfish incumbents and not to the more selfish ones. This is a sign restriction with no free parameters, and Section \ref{sec:regularity} shows how nested-menu designs test it and how to separate the normative ceiling from the competing reading under which what matters is the decision maker’s own downside exposure. Appendix \ref{app:psychometric} recasts the binary choice probabilities as a menu-dependent Fechnerian psychometric function.

Three questions remain open. The first is whether the uniform increment comparison holds for all pairs: that inequality is the entire residual content of \textsf{Normative Richness}~(3), and settling it would determine whether the clause is an axiom or a theorem (Appendix \ref{app:independence}). The second is whether a social utility of the whole allocation can be identified alongside $\varphi$, which is the obstacle to extending $u$ (Section \ref{sec:social_scope}). The third is how the normative ceiling itself responds to the features a designer can vary---anonymity, social distance, whether the recipient observes that the choice was randomized. The third is the most immediately testable, because Proposition \ref{prop:regularity} is a parameter-free sign restriction: a nested-menu design can look for the predicted pattern before anything is estimated.

\section*{Appendix}
Appendix \ref{app:proofs} contains the proofs. Appendix \ref{app:fis} relates our conditions to those of \citet{FIS_2015}, Appendix \ref{app:independence} establishes the independence of the axioms, and Appendix \ref{app:psychometric} recasts the binary choice probabilities as a menu-dependent psychometric function.

\appendix
\section{Proofs and Supporting Examples}
\label{app:proofs}
The three steps of the argument are outlined in Section \ref{results}, immediately after the statement of the theorem.

\subsection{Theorem \ref{representation_norm}: Sufficiency}
\label{app:sufficiency}
We present the sufficiency part.

\subsubsection{Step 1}
We start from an allocation-invariant additive perturbed utility (APU) representation, i.e., a menu-dependent APU representation.
\begin{definition}[\textsf{Menu-Dependent Additive Perturbed Utility}]
We say that $\rho$ has a \textit{menu-dependent APU} if there is a pair $(u, (c_{A})_{A \in \mathcal{A}})$, where $u: X \rightarrow \mathbb{R}$ is a utility function, and $c_{A} : [0, 1] \rightarrow \mathbb{R} \cup \lbrace + \infty \rbrace$ is a cost function for each $A \in \mathcal{A}$, such that $\rho$ has the form
\begin{align*}
\rho(A) = \underset{\rho \in \Delta(A)} {\operatorname{argmax}} \sum_{{\bm x} \in A} \Bigl(u({\bm x}) \rho({\bm x}, A) - c_{A}(\rho({\bm x}, A)) \Bigr).
\end{align*}
\end{definition}

We show the following lemma, following Proposition 8 in \citet{FIS_2014}. Define the \textit{revealed allocation relations} $\succ^\rho, \sim^\rho \subset X \times X$ by ${\bm x} \succ^\rho {\bm y}$ if $\rho({\bm x},A) > \rho({\bm y},A)$ for some common menu $A \ni {\bm x},{\bm y}$, and ${\bm x} \sim^\rho {\bm y}$ if $0 < \rho({\bm x},A) = \rho({\bm y},A) < 1$ for some common menu $A$. We say $\rho$ satisfies \textit{Allocation Acyclicity} if there is no cycle in $\succ^\rho \cup \sim^\rho$ with at least one link in $\succ^\rho$ (i.e., condition (i) of Lemma \ref{apu_farkas} holds for $\succ = \succ^\rho, \sim = \sim^\rho$).

\begin{lemma}[\textsf{Menu-Dependent Additive Perturbed Utility}]
\label{apu_menu}
Suppose that $\rho$ satisfies axiom \ref{positivity} (\textit{Positivity}). The following statements are equivalent:
\begin{enumerate}
\renewcommand{\labelenumi}{(\roman{enumi})}
	\item $\rho$ is a menu-dependent APU.
	\item There exists $u: X \rightarrow \mathbb{R}$ such that $u({\bm x}) > u({\bm y})$ if $\rho({\bm x}, A) > \rho({\bm y}, A)$, and $u({\bm x}) = u({\bm y})$ if $\rho({\bm x}, A) = \rho({\bm y}, A) \in (0, 1)$ for each menu $A \in \mathcal{A}$.
	\item $\rho$ satisfies Allocation Acyclicity.
\end{enumerate}
\end{lemma}

\begin{proof}
We show that $(i) \Rightarrow (ii)$, $(ii) \Rightarrow (i)$, and $(ii) \Leftrightarrow (iii)$. \paragraph*{\textsf{$(i) \Rightarrow (ii)$}:} Suppose that $\rho$ has a menu-dependent APU. Take an arbitrary menu $A \in \mathcal{A}$ with ${\bm x}, {\bm y} \in A$. Let $\lambda: \mathcal{A} \rightarrow \mathbb{R}$ be the Lagrangian multiplier associated with each menu. The first-order condition (FOC) for $\rho$ is as follows. There exists $\lambda(A)$ for each ${\bm x} \in A$ such that
\begin{equation}
\label{eq_FOC}
u({\bm x}) - c'_{A}(\rho({\bm x}, A)) + \lambda(A)
\begin{cases}
    \geq 0 & \text{if $\rho({\bm x}, A) = 1$,} \\
    = 0   & \text{if $\rho({\bm x}, A) \in (0, 1)$,} \\
    \leq 0 & \text{if $\rho({\bm x}, A) = 0$.}
  \end{cases}
\end{equation}

By the FOC for $\rho$ and the strict convexity of $c_{A}$,
\begin{align*}
\rho({\bm x}, A) > \rho({\bm y}, A) &\Rightarrow u({\bm x}) > u({\bm y}),
\end{align*}
and
\begin{align*}
\rho({\bm x}, A) = \rho({\bm y}, A) \in (0, 1) &\Rightarrow u({\bm x}) = u({\bm y}),
\end{align*}

\paragraph*{\textsf{$(ii) \Rightarrow (i)$}:} Suppose that there exists $u: X \rightarrow \mathbb{R}$ such that $u({\bm x}) > u({\bm y})$ if $\rho({\bm x}, A) > \rho({\bm y}, A)$, and $u({\bm x}) = u({\bm y})$ if $\rho({\bm x}, A) = \rho({\bm y}, A) \in (0, 1)$ for each menu $A \in \mathcal{A}$.

We construct $(c_{A})_{{A} \in \mathcal{A}}$ using $\lambda$. Take an arbitrary menu $A \in \mathcal{A}$. Let
\begin{align*}
	\underline{w}(A) :=
	\begin{cases}
	0 & \text{if $\rho({\bm x}, A) > 0$, for all ${\bm x} \in A$,} \\
	\max \lbrace u({\bm x}) + \lambda(A) \mid {\bm x} \in A, \rho({\bm x}, A) = 0 \rbrace & \text{otherwise.}
	\end{cases}
\end{align*}

We construct a function $g_{A} : [0, 1] \rightarrow \mathbb{R}$ for each $A \in \mathcal{A}$ such that
\begin{enumerate}
\renewcommand{\labelenumi}{(\roman{enumi})}
	\item $g_{A}(0) = \underline{w}(A)$; and
	\item $g_{A}(\rho({\bm x}, A)) = u({\bm x})$ if $\rho({\bm x}, A) \in (0, 1)$.
\end{enumerate}
Intuitively, in case (i), there is no randomization.

Define a cost function $c_{A}: [0,1] \rightarrow \mathbb{R}$ by
\begin{align*}
	c_{A}(q) := \int_0^q g_{A}(p)\,dp.
\end{align*}
Take ${\bm z} \in A \cap B$ for some $A,  B \in \mathcal{A}$. Suppose that $\rho({\bm z}, A), \rho({\bm z}, B) \in (0, 1)$ with $\rho({\bm z}, A) > \rho({\bm z}, B)$. Then, we have
\begin{align*}
u({\bm x}) - c'_{A}(\rho({\bm x}, A))
\begin{cases}
    \geq 0 & \text{if $\rho({\bm x}, A) = 1$,} \\
    = 0   & \text{if $\rho({\bm x}, A) \in (0, 1)$,} \\
    \leq 0 & \text{if $\rho({\bm x}, A) = 0$.}
  \end{cases}
\end{align*}
The FOC condition holds with $\lambda(A) = 0$ for each $A \in \mathcal{A}$.

We show that $c_{A}$ is $C^1$ over $(0,1)$ for each $A \in \mathcal{A}$. By definition, we obtain
\begin{align*}
	\frac{d}{dq} \int_0^q g_{A}(p)\,dp = g_{A}(q).
\end{align*}

By definition and the property of $\lambda(\cdot) = 0$, the strict convexity of $c_{A}$ follows from $g'_{A}(q) > 0$ for all $q \in (0, 1)$.

\paragraph*{\textsf{$(ii) \Leftrightarrow (iii)$}:} We apply the following lemma in \citet{FIS_2014} (see Lemma 3).
\begin{lemma}[A Version of Farkas's Lemma]
\label{apu_farkas}
Let $\mathcal{X}$ be a finite set. Suppose that $\succ, \sim \ \subset \mathcal{X} \times \mathcal{X}$ such that $\succ$ is asymmetric and $\sim$ is symmetric. Then, the following conditions are equivalent.
\begin{enumerate}
\renewcommand{\labelenumi}{(\roman{enumi})}
	\item There is no cycle, meaning that there is no sequence
		\begin{align*}
			(\chi_1, \chi_2), (\chi_2, \chi_3), \cdots, (\chi_{m-1}, \chi_m), (\chi_m, \chi_1)
		\end{align*}	
	in $\succ \cup \sim$, where at least one of them belongs to $\succ$.
	\item There exists a function $v: \mathcal{X} \rightarrow \mathbb{R}$ such that $v(\chi) > v(\chi')$ if $\chi \succ \chi'$, and $v(\chi) = v(\chi')$ if $\chi \sim \chi'$.
\end{enumerate}
\end{lemma}

\begin{remark}[Extension to Infinite $X$]
\label{rmk:farkas_infinite}
Lemma \ref{apu_farkas} is stated for a finite set $\mathcal{X}$, but $X$ is a (possibly uncountable) convex subset of $\mathbb{R}^2$, not a finite set. We use the following standard extension throughout, mirroring the argument used by \citet{FIS_2015} in the proof of their Theorem 2 for infinite $Z$. Suppose $\succ, \sim$ satisfy condition (i) (no cycle) when restricted to every finite subset of $X$, and suppose in addition that $\succsim := \succ \cup \sim$ has closed upper and lower contour sets on $X$ (which holds whenever $\succ,\sim$ are built from $\rho$, since $\rho$ is continuous by Axiom \ref{continuity}). Fix a countable dense subset $X_0 \subseteq X$ (which exists since $X \subseteq \mathbb{R}^2$ is separable). Applying Lemma \ref{apu_farkas} to every finite subset of $X_0$ and taking a consistent limit yields $v_0: X_0 \to \mathbb{R}$ representing $\succsim$ on $X_0$. Because $X$ is convex, hence connected, and $\succsim$ has closed contour sets, $v_0$ admits a continuous extension $v: X \to \mathbb{R}$ representing $\succsim$ on all of $X$, following the classical construction of \citet{D_1954}. We invoke this extended version of Lemma \ref{apu_farkas} whenever we apply it to all of $X$.
\end{remark}

Apply Lemma \ref{apu_farkas} (extended to infinite $X$ per Remark \ref{rmk:farkas_infinite}) with $\succ = \succ^\rho$ and $\sim = \sim^\rho$, the revealed allocation relations defined above. Condition (i) of Lemma \ref{apu_farkas} for this pair is, by definition, exactly \textit{Allocation Acyclicity}: there is no sequence of allocations ${\bm x}^1, {\bm x}^2, \ldots, {\bm x}^m$ with ${\bm x}^1 \succsim^\rho {\bm x}^2 \succsim^\rho \cdots \succsim^\rho {\bm x}^m \succ^\rho {\bm x}^1$ (where $\succsim^\rho := \succ^\rho \cup \sim^\rho$), with at least one strict ($\succ^\rho$) link. Condition (ii) of Lemma \ref{apu_farkas} for this pair is exactly statement (ii) of the present lemma. Hence Lemma \ref{apu_farkas} gives $(ii) \Leftrightarrow (iii)$ directly.

\textit{Remark.} Allocation Acyclicity, as just defined, is logically weaker than \textsf{Selfishness-Acyclicity} (Axiom \ref{ax:selfish_acyclicity}), which is assumed throughout the paper: the combined relation used in Axiom \ref{ax:selfish_acyclicity} (built from $\rho(\cdot,A^k) \geq \rho(\cdot,A^k)$-type comparisons \emph{and} the $x_1$-order) contains $\succ^\rho \cup \sim^\rho$ as a sub-relation, so any cycle witnessing a violation of Allocation Acyclicity would also witness a violation of Axiom \ref{ax:selfish_acyclicity}. Consequently, Axiom \ref{ax:selfish_acyclicity} implies Allocation Acyclicity, and Lemma \ref{apu_menu} applies whenever Axiom \ref{ax:selfish_acyclicity} holds; we do not need to invoke Axiom \ref{ax:selfish_acyclicity}'s full strength (in particular, the $x_1$-order component) until Claim \ref{claim:selfish_repn} below, where it is used to additionally pin down that $u$ depends on ${\bm x}$ only through $x_1$.

\end{proof}

\subsubsection*{Norm-Dependent Additive Perturbed Utility}
We extend the menu-dependent APU to a norm-dependent APU. There is a pair $(u, (C_{A})_{A \in \mathcal{A}})$ where $u: X \rightarrow \mathbb{R}$ is a utility function, and each $C_{A}: [0, 1] \rightarrow \mathbb{R} \cup \lbrace + \infty \rbrace$ is a cost function such that $\rho$ has the following form:
\begin{align*}
\rho(A) = \underset{\rho \in \Delta(A)} {\operatorname{argmax}} \sum_{{\bm x} \in A} \Bigl(u(x_1) \rho({\bm x}, A) - C(\rho({\bm x}, A), A) \Bigr).
\end{align*}
We now show that $u$ is selfish, i.e., $u({\bm x}) = u(x_1)$ for each ${\bm x} \in X$ (Claim \ref{claim:selfish_repn} below). This $u$ is provisional: as explained in the remark following the claim, it is used only to extract a monotonicity fact, and is superseded in Step 3 by a selfish utility that is, by construction, automatically continuous and globally consistent across all normative ceilings.

\subsubsection*{Utility Function}
Regarding the model's components, the selfish utility function $u$ represents a personal ranking on $X$ that maximizes only the decision maker's payoff. That is, for any ${\bm x}, {\bm y} \in X$, $x_1 \geq y_1 \Leftrightarrow u({\bm x}) \geq u({\bm y})$ holds. This building block of the model mainly stems from Axiom \ref{ax:selfish_acyclicity} (\textsf{Selfishness-Acyclicity}).
\begin{claim}[\textsf{Utility Representation}]
\label{claim:selfish_repn}
$u: X \rightarrow \mathbb{R}$ is selfish, i.e., there exists a strictly increasing $\tilde u: \mathbb{R} \rightarrow \mathbb{R}$ such that $u({\bm x}) = \tilde u(x_1)$ for each ${\bm x} \in X$.
\end{claim}
\begin{proof}
Define binary relations $\succ, \sim$ on $X$ by: ${\bm x} \succ {\bm y}$ if either (a) $\rho({\bm x}, A) > \rho({\bm y}, A)$ for some common menu $A \in \mathcal{A}$, or (b) $x_1 > y_1$; and ${\bm x} \sim {\bm y}$ if either (a) $\rho({\bm x}, A) = \rho({\bm y}, A) \in (0,1)$ for some common menu $A \in \mathcal{A}$, or (b) $x_1 = y_1$. By Axiom \ref{ax:selfish_acyclicity} (\textsf{Selfishness-Acyclicity}), $\succ \cup \sim$ contains no cycle with at least one strict ($\succ$) link. Hence, by Lemma \ref{apu_farkas} (extended to infinite $X$ per Remark \ref{rmk:farkas_infinite}), there exists $u: X \to \mathbb{R}$ such that $u({\bm x}) > u({\bm y})$ if ${\bm x} \succ {\bm y}$, and $u({\bm x}) = u({\bm y})$ if ${\bm x} \sim {\bm y}$.

By construction, this $u$ satisfies $u({\bm x}) > u({\bm y})$ whenever $\rho({\bm x},A) > \rho({\bm y}, A)$ for a common menu $A$, and $u({\bm x}) = u({\bm y})$ whenever $\rho({\bm x},A) = \rho({\bm y},A) \in (0,1)$, so it coincides (up to the ordinal identification in Lemma \ref{apu_menu}) with the utility function used in the menu-dependent APU representation of Step 1. It also satisfies, directly from clause (b), $u({\bm x}) > u({\bm y})$ whenever $x_1 > y_1$, and $u({\bm x}) = u({\bm y})$ whenever $x_1 = y_1$. Hence $u$ is constant on each level set of $x_1$ and strictly increasing across level sets, so $u({\bm x}) = \tilde u(x_1)$ for a strictly increasing function $\tilde u: \mathbb{R} \to \mathbb{R}$, i.e., $u$ is selfish.
\end{proof}

The $u$ just constructed is pinned down only ordinally: by Lemma \ref{apu_menu}, any strictly increasing transform of $u$ that preserves the same order is an equally valid choice, with the menu-specific cost functions $C_A$ reconstructed accordingly (this is exactly the $\lambda(A)=0$-WLOG freedom in the proof of Lemma \ref{apu_menu}, $(ii)\Rightarrow(i)$). All we use from $u$ going forward is the one fact this ordinal information delivers: for any menu $D$ and ${\bm z},{\bm z}'\in D$, $z_1>z_1' \Rightarrow \rho({\bm z},D) > \rho({\bm z}',D)$ (immediate from the interior first-order conditions of Lemma \ref{apu_menu} together with $u(z_1)>u(z_1')$ and strict convexity of $C_D$). In Step 3 (Claim \ref{claim:well_defined_c}), we use this fact to build a selfish utility $M$ directly from Axiom \ref{ax:norm_independence} and \textsf{Normative Richness}, and thereafter write $u:=M$; $M$ is, by its own construction, automatically continuous and valid at every normative ceiling simultaneously, so no separate continuity argument (and no comparison between $M$ and this provisional $u$) is needed.

\subsubsection{Step 2}
The next step is to consider the personal norm utility representation $\varphi: X \rightarrow \mathbb{R}$, which represents $\succsim_n^{\rho}$ on $X$. By Axiom \ref{ax:norm_weak_order}, $\succsim_n^\rho$ (defined in Definition \ref{norm}) is a weak order.

We show Lemma \ref{norm_result} in the following statement, which describes the characterization of personal norm rankings $\succsim_n^\rho$ on $X$.
\begin{lemma}[\textsf{Personal Norm Utility Representation}]
\label{norm_result}
$\succsim^\rho_n$ on $X$ is represented by a continuous function $\varphi: X \rightarrow \mathbb{R}$, i.e.,
	\begin{align*}
		{\bm x} \succsim_n^{\rho} {\bm y} \ \text{if} \ \text{and only} \ \text{if} \ \varphi({\bm x}) \geq \varphi({\bm y}).
	\end{align*}
\end{lemma}

\begin{proof}
The lemma is a single application of Debreu's representation theorem, whose three hypotheses are supplied as follows.

\emph{(i) The domain is a connected, separable metric space.} By the standing assumption of Section \ref{axioms}, $X$ is a convex subset of $\mathbb{R}^2$. Convexity gives connectedness, and any subset of $\mathbb{R}^2$ is a separable metric space in the induced topology (hence second countable). No further structure on $X$ is used.

\emph{(ii) The relation is a weak order.} By Axiom \ref{ax:norm_weak_order} (\textsf{Normative Weak Order}), $\succsim_n^\rho$---which by Definition \ref{norm} is defined from $\rho$ alone---is complete and transitive.

\emph{(iii) The relation is continuous.} Also by Axiom \ref{ax:norm_weak_order}, the weak upper and lower contour sets of $\succsim_n^\rho$ are closed in $X$. As noted when that axiom was stated, closedness is assumed rather than derived from Axiom \ref{continuity}: Definition \ref{norm} is an implication whose hypothesis---the calibration equality $\rho({\bm w},A\cup\{{\bm x}\})=\rho({\bm w},A\cup\{{\bm y}\})$---is a knife-edge condition, so continuity of $\rho$ in the allocations does not by itself make the contour sets closed.

By Debreu's theorem \citep{BM_1995}, there is a continuous $\varphi:X\to\mathbb{R}$ with ${\bm x}\succsim_n^\rho{\bm y}$ if and only if $\varphi({\bm x})\ge\varphi({\bm y})$, which is the assertion of the lemma.

We note that, unlike Step 1, this step does \emph{not} go through the Farkas-type Lemma \ref{apu_farkas}. That lemma yields a representation of an acyclic relation but says nothing about continuity, so invoking it here and then invoking Debreu separately would produce two different functions and prove nothing about either; since \textsf{Normative Weak Order} already delivers a continuous weak order, Debreu's theorem alone gives existence and continuity of one and the same $\varphi$. Step 1 cannot proceed this way, because the relation used there is not assumed complete.

This completes the proof of Lemma \ref{norm_result}.\footnote{An earlier draft additionally claimed that $\varphi$ can be chosen \emph{monotone with respect to fair allocations} (i.e., increasing along the diagonal $\{(\alpha,\alpha):\alpha\in\mathbb{R}\}$). This is a genuine economic property--it is, in effect, Normative Pareto Monotonicity (an axiom we do not impose) restricted to fair allocations--but it is neither implied by the axioms retained in Theorem \ref{representation_norm} (Normative Weak Order alone does not order the diagonal at all) nor used anywhere else in this paper, so we do not claim it here. Establishing it would require reinstating Normative Pareto Monotonicity, or a comparable axiom, as a primitive.}
\end{proof}

\subsubsection{Step 3: Construction of a Globally Selfish Utility and the Norm-Dependent Cost Function}
\label{app:step3}
Step 1 gives a selfish utility, order-consistent with $\rho$ and with $x_1$, but pinned down only ordinally, paired with menu-specific cost functions $C_A$ that need not depend on $A$ only through its normative ceiling $t=\max_{{\bm y}\in A}\varphi({\bm y})$, still less agree across different ceilings. We now use Axiom \ref{ax:norm_independence} (\textsf{Normative Independence}) together with \textsf{Normative Richness} (Axiom \ref{richness_below_ceiling}) to construct, directly from $\rho$, a single selfish utility function $M:\mathbb R\to\mathbb R$---continuous and strictly increasing by construction---together with, for every $t\in\varphi(X)$, a single-index cost function $c(\cdot,t)$, such that $M$ and $c(\cdot,t)$ jointly satisfy the interior first-order conditions of \eqref{model_norm} on \emph{every} menu, regardless of ceiling. Step 1's provisional utility is used only internally, to justify one monotonicity fact along the way (see the remark above); the final selfish utility of the representation is $M$.

\begin{claim}\label{claim:well_defined_c}
There is a continuous, strictly increasing function $M:\mathbb R\to\mathbb R$ and, for every $t\in\varphi(X)$, a $C^1$ function $c(\cdot,t):(0,1)\to\mathbb R$, strictly convex, such that for every $A\in\mathcal A$ with $\varphi(A)=t$, the quantity $M(x_1)-c_1(\rho({\bm x},A),t)$ takes a common value $\lambda_A$ across all ${\bm x}\in A$. In particular, $C_A=c(\cdot,\varphi(A))$ for every menu $A$ (so any two menus of the same ceiling induce the same cost function), with $c(\cdot,t)$ constructed directly from Axiom \ref{ax:norm_independence}'s functions and independent of any particular menu.
\end{claim}

\begin{proof}
For each $t\in\varphi(X)$, fix a reference allocation ${\bm r}_t\in X$ with $\varphi({\bm r}_t)=t$ (which exists because $t\in\varphi(X)$; unique up to $\sim_n^\rho$ by \textsf{Normative Weak Order}), and set $g_t := g_{{\bm r}_t}$, where $\{g_{\bm a}\}_{{\bm a}\in X}$ is the family given by Axiom \ref{ax:norm_independence}. Since ${\bm r}_t$ sets the ceiling of any ceiling-$t$ menu containing it, and \emph{any} allocation ${\bm a}$ that sets the ceiling of some ceiling-$t$ menu satisfies ${\bm a}\sim_n^\rho{\bm r}_t$ (both are ceiling-setting elements at the same level), the axiom's requirement $g_{\bm a}=g_{\bm b}$ whenever ${\bm a}\sim_n^\rho{\bm b}$ guarantees $g_{\bm a}=g_t$ for every such ${\bm a}$; below we simply write ``$g_t$'' throughout and invoke this fact whenever the argument calls for the rescaling attached to \emph{some} ceiling-setting element of a ceiling-$t$ menu, without needing that element to literally be ${\bm r}_t$. We first construct, separately for each fixed $t$, a transport function $m_t$ (Same-payoff transport, Monotonicity, Chaining identity below, exactly as in the within-ceiling argument); we then show the $m_t$'s across \emph{different} ceilings agree up to an additive constant (Cross-ceiling matching), and finally assemble $M$ and $c(\cdot,t)$ (Conclusion).

\textit{Same-payoff transport (fixed $t$).} Fix $t\in\varphi(X)$; write ${\bm r}:={\bm r}_t$ for brevity within this and the next two steps. We first show that
\begin{align*}
    m_t(z_1) \;:=\; g_t\bigl(\rho({\bm z},\{{\bm r},{\bm z}\})\bigr) - g_t\bigl(\rho({\bm r},\{{\bm r},{\bm z}\})\bigr)
\end{align*}
is well-defined for any ${\bm z}\neq{\bm r}$ with $\varphi({\bm z})\le t$, i.e., depends on ${\bm z}$ only through $z_1$; for ${\bm z}={\bm r}$ itself, $\{{\bm r},{\bm z}\}$ collapses to a singleton and the displayed formula is not meaningful, so we instead adopt the convention $m_t(r_1):=0$ (comparing ${\bm r}$ with itself trivially yields no information to transport). With this convention in place, the Chaining identity and Cross-ceiling matching arguments below continue to hold verbatim when one of the two designated allocations equals ${\bm r}$, so no separate treatment of that case is needed there. Now let ${\bm z},{\bm z}'$ both differ from ${\bm r}$ and satisfy $\varphi({\bm z}),\varphi({\bm z}')\le t$ and $z_1=z_1'$; if ${\bm z}={\bm z}'$ there is nothing to prove, so suppose ${\bm z}\neq{\bm z}'$ and consider $E:=\{{\bm r},{\bm z},{\bm z}'\}$, a menu of ceiling $t$ (via ${\bm r}$) with three distinct elements. By Lemma \ref{apu_menu}, $E$ admits a menu-specific cost function $C_E$ satisfying the interior first-order conditions of \eqref{model_norm} together with the (already fixed, provisional, selfish) utility $u$ of Step 1; since $u(z_1)=u(z_1')$, these conditions give $C_E'(\rho({\bm z},E))=C_E'(\rho({\bm z}',E))$, and since $C_E'$ is strictly increasing (strict convexity of $C_E$), $\rho({\bm z},E)=\rho({\bm z}',E)$. Applying Axiom \ref{ax:norm_independence} to the pair $(E,\{{\bm r},{\bm z}\})$ (which share $\{{\bm r},{\bm z}\}$, both menus having ceiling $t$) and to the pair $(E,\{{\bm r},{\bm z}'\})$ (sharing $\{{\bm r},{\bm z}'\}$) gives
\begin{align*}
    g_t(\rho({\bm z},E)) - g_t(\rho({\bm r},E)) &= g_t(\rho({\bm z},\{{\bm r},{\bm z}\})) - g_t(\rho({\bm r},\{{\bm r},{\bm z}\})), \\
    g_t(\rho({\bm z}',E)) - g_t(\rho({\bm r},E)) &= g_t(\rho({\bm z}',\{{\bm r},{\bm z}'\})) - g_t(\rho({\bm r},\{{\bm r},{\bm z}'\})).
\end{align*}
Since $\rho({\bm z},E)=\rho({\bm z}',E)$, the left-hand sides coincide, hence so do the right-hand sides: $m_t(z_1)$ does not depend on the choice of representative ${\bm z}$.

\textit{Monotonicity of $m_t$.} If $z_1>z_1'$ (with $\varphi({\bm z}),\varphi({\bm z}')\le t$), the same menu $E=\{{\bm r},{\bm z},{\bm z}'\}$ has, by the interior first-order conditions and strict convexity of $C_E$, $\rho({\bm z},E)>\rho({\bm z}',E)$ (since $u(z_1)>u(z_1')$, Step 1). Subtracting the two displayed identities above and using that $g_t$ is strictly increasing gives $m_t(z_1)-m_t(z_1') = g_t(\rho({\bm z},E))-g_t(\rho({\bm z}',E)) > 0$. Hence $m_t$ is strictly increasing on its domain. That domain is in fact all of $\mathbb R$: by \textsf{Normative Richness} (Axiom \ref{richness_below_ceiling}), applied with ${\bm a}:={\bm r}$, for every $v\in\mathbb R$ there is ${\bm z}\in X$ with $z_1=v$ and $\varphi({\bm z})\le t$.

\textit{Chaining identity.} We claim that, for \emph{every} menu $A'\in\mathcal A$ with $\varphi(A')=t$ and every ${\bm x},{\bm y}\in A'$,
\begin{align*}
    g_t(\rho({\bm x},A')) - g_t(\rho({\bm y},A')) \;=\; m_t(x_1) - m_t(y_1).
\end{align*}
If ${\bm r}\in A'$, this follows exactly as in the Same-payoff transport step (with $E$ replaced by $A'$, and without requiring $x_1=y_1$). If ${\bm r}\notin A'$, first apply Axiom \ref{ax:norm_independence} to the pair $(A',\,A'\cup\{{\bm r}\})$, which share all of $A'$ (both menus have ceiling $t$, since $\varphi({\bm r})=t$), to get $g_t(\rho({\bm x},A'))-g_t(\rho({\bm y},A')) = g_t(\rho({\bm x},A'\cup\{{\bm r}\}))-g_t(\rho({\bm y},A'\cup\{{\bm r}\}))$, and then apply the ${\bm r}\in A'$ case to $A'\cup\{{\bm r}\}$.

\textit{Cross-ceiling matching.} We now show that the transport functions $m_t$ constructed at \emph{different} ceilings agree up to an additive constant: for any $s,t\in\varphi(X)$, there is a constant $\kappa(s,t)$ such that $m_t(z_1) = m_s(z_1) + \kappa(s,t)$ for every $z_1\in\mathbb R$. By symmetry it suffices to treat $s\le t$. By \textsf{Normative Richness} (Axiom \ref{richness_below_ceiling}), applied with ${\bm a}:={\bm r}_s$, for every $z_1\in\mathbb R$ there exists ${\bm z}\in X$ with this payoff and $\varphi({\bm z})\le s$. Let ${\bm w}_t\in X$ satisfy $\varphi({\bm w}_t)=t$ (take ${\bm w}_t:={\bm r}_t$), and consider $A:=\{{\bm r}_s,{\bm z}\}$ (ceiling $s$, via ${\bm r}_s$) and $B:=\{{\bm r}_s,{\bm z},{\bm w}_t\}$ (ceiling $t$, via ${\bm w}_t$; if $s=t$ we may take ${\bm w}_t:={\bm r}_s$ itself and $B:=A$, in which case the claim is immediate). Here ${\bm r}_s$ sets the ceiling of $A$ and ${\bm w}_t$ sets the ceiling of $B$ (so, per the convention above, $g_{{\bm r}_s}=g_s$ and $g_{{\bm w}_t}=g_t$); applying Axiom \ref{ax:norm_independence} to $A,B$ (which need not share a ceiling) with ${\bm x}:={\bm z}$, ${\bm y}:={\bm r}_s$,
\begin{align*}
    g_s(\rho({\bm z},A)) - g_s(\rho({\bm r}_s,A)) \;=\; g_t(\rho({\bm z},B)) - g_t(\rho({\bm r}_s,B)).
\end{align*}
The left-hand side equals $m_s(z_1) - m_s(r_{s,1}) = m_s(z_1)$ (definition of $m_s$ and the convention $m_s(r_{s,1})=0$). The right-hand side, by the Chaining identity at level $t$ applied to $B$ (which has ceiling $t$ and contains ${\bm z},{\bm r}_s$), equals $m_t(z_1) - m_t(r_{s,1})$. Hence $m_s(z_1) = m_t(z_1) - m_t(r_{s,1})$ for every $z_1\in\mathbb R$, i.e., $m_t = m_s + \kappa(s,t)$ with $\kappa(s,t):=m_t(r_{s,1})$, a constant not depending on $z_1$.

\textit{Conclusion.} Fix any reference level $t_0\in\varphi(X)$ and define $M:=m_{t_0}$. By Cross-ceiling matching, for every $t\in\varphi(X)$, $m_t = M + \kappa(t)$ on all of $\mathbb R$, where $\kappa(t):=\kappa(t_0,t)$ if $t_0\le t$ and $\kappa(t):=-\kappa(t,t_0)$ if $t<t_0$ (consistent with $\kappa(t_0)=0$). By Monotonicity of $m_{t_0}$, $M$ is strictly increasing with domain all of $\mathbb R$.

We claim that $M$ is in fact \emph{onto} $\mathbb R$, and hence continuous. Writing ${\bm r}:={\bm r}_{t_0}$ and, for ${\bm z}\neq{\bm r}$ with $\varphi({\bm z})\le t_0$, $p:=\rho({\bm z},\{{\bm r},{\bm z}\})$---so that $\rho({\bm r},\{{\bm r},{\bm z}\}) = 1-p$---the definition of $m_{t_0}$ reads
\begin{align*}
    M(z_1) \;=\; g_{t_0}(p) - g_{t_0}(1-p) \;=:\; G(p),
\end{align*}
where $G:(0,1)\to\mathbb R$ is strictly increasing (as $g_{t_0}$ is) and, by the boundary requirement on $g_{t_0}$ in Axiom \ref{ax:norm_independence}, satisfies $G(p)\to-\infty$ as $p\to0^+$ and $G(p)\to+\infty$ as $p\to1^-$; being strictly increasing and $C^1$, $G$ therefore maps $(0,1)$ \emph{onto} $\mathbb R$. By \textsf{Normative Richness}~(3) (Axiom \ref{richness_binary_below_ceiling}), applied with ${\bm a}:={\bm r}$, \emph{every} $p\in(0,1)$ is realized as $\rho({\bm z},\{{\bm r},{\bm z}\})$ for some ${\bm z}$ with ${\bm r}\succsim_n^\rho{\bm z}$, i.e., $\varphi({\bm z})\le t_0$. Hence the range of $M$ contains---indeed equals---$G((0,1)) = \mathbb R$. A strictly increasing function from $\mathbb R$ onto $\mathbb R$ is continuous (a monotone function can only fail continuity at a jump, and a jump would leave an interval of values unattained, contradicting surjectivity). Thus $M$ is continuous.

We emphasize that this argument replaces a tempting but incorrect one: continuity of the map ${\bm z}\mapsto\rho({\bm z},\{{\bm r},{\bm z}\})$ (Axiom \ref{continuity}) makes ${\bm z}\mapsto M(z_1)$ continuous \emph{as a function of the allocation ${\bm z}$}, which is not the same as continuity of $M$ \emph{as a function of the real variable $z_1$}: the latter would additionally require a continuous selection of allocations below the ceiling with prescribed first coordinate, which Axiom \ref{richness_below_ceiling} does not provide (it gives existence pointwise, not a selection). The monotone-surjection argument above avoids the issue entirely.

Define $c_1(\cdot,t) := g_t(\cdot)$ (no rescaling is needed: unlike the argument this replaces, $m_t$ and $M$ differ only \emph{additively}, never multiplicatively) and let $c(\cdot,t)$ be any antiderivative; by the boundary requirement on $g_t$ in Axiom \ref{ax:norm_independence}, $c_1(p,t)\to-\infty$ as $p\to0^+$ and $c_1(p,t)\to+\infty$ as $p\to1^-$, which is exactly Definition \ref{norm-APU}(ii). Condition \ref{norm-APU}(iv) needs no separate verification: the tuple built here represents a $\rho$ that satisfies all the axioms of Theorem \ref{representation_norm}(a), and by Remark \ref{rmk:iv_free} any such representation satisfies (iv) automatically. For any menu $A'$ with $\varphi(A')=t$ and any ${\bm x},{\bm y}\in A'$, the Chaining identity together with $m_t = M+\kappa(t)$ give
\begin{align*}
    c_1(\rho({\bm x},A'),t) - c_1(\rho({\bm y},A'),t) \;=\; g_t(\rho({\bm x},A'))-g_t(\rho({\bm y},A')) \;=\; m_t(x_1)-m_t(y_1) \;=\; M(x_1)-M(y_1),
\end{align*}
so that $M(x_1) - c_1(\rho({\bm x},A'),t)$ takes the same value, call it $\lambda_{A'}$, for every ${\bm x}\in A'$: exactly the interior first-order conditions of \eqref{model_norm} with selfish utility $M$. Since $\rho(\cdot,A')$ already satisfies the budget constraint $\sum_{{\bm x}\in A'}\rho({\bm x},A')=1$, and $c(\cdot,t)$ is strictly convex (as $c_1(\cdot,t)=g_t$ is strictly increasing), these first-order conditions are both necessary and sufficient for optimality in \eqref{model_norm}, so $\rho(\cdot,A')$ is exactly the maximizer of \eqref{model_norm} under $(M,c(\cdot,t))$; in particular $C_{A'}=c(\cdot,t)$ for every menu $A'$ of ceiling $t$. Since $t$ was an arbitrary element of $\varphi(X)$, this proves the claim, with $u:=M$.
\end{proof}

Henceforth we write $u:=M$ for this globally shared, continuous selfish utility, superseding the provisional utility of Claim \ref{claim:selfish_repn} (Step 1); the two are consistent with the same order on $X$ (Step 1's sole role was to license the within-menu monotonicity fact used in Monotonicity of $m_t$ above), but need not coincide numerically, and only $M$ is used from this point on. In particular, $u$ is continuous by construction, which is why no separate continuity step is needed later.

The norm-dependent cost function $c(p,t)$, $p\in(0,1)$, $t\in\varphi(X)$, is exactly the one constructed in the proof of Claim \ref{claim:well_defined_c}, directly from Axiom \ref{ax:norm_independence}'s functions $g_t$, before ever referring to a particular menu; by the claim, $C_A=c(\cdot,\varphi(A))$ for every $A\in\mathcal A$.

\subsubsection*{Step 4: Proposition \ref{prop:marginal_cost}}
\label{app:closing_marginal_cost}

\paragraph{Why a level statement is not well-posed} By Theorem \ref{unique_norm}(iii), any two representations of the same $\rho$ satisfy $\hat c(p,t) = \alpha\, c(p,t) + \gamma(t)p + \delta(t)$ for some $\alpha>0$ and completely unrestricted functions $\gamma,\delta:\varphi(X)\to\mathbb R$. Differentiating once, $\hat c_1(p,t) = \alpha\, c_1(p,t) + \gamma(t)$: the \emph{level} of $c_1(\cdot,t)$, as $t$ varies, can be shifted by an arbitrary function of $t$ alone without changing $\rho$ at all---so whether $c_1(p,t)$ is non-decreasing in $t$ is not, by itself, a fact about $\rho$; it depends on which representative $\gamma(t)$ one happens to choose. Fixing any particular normalization (e.g., a reference probability $p_0\in(0,1)$ and anchoring $c_1(p_0,t)=0$ for every $t$, always available by taking $\gamma(t):=-c_1(p_0,t)$ in Theorem \ref{unique_norm}(iii)) does not remove this indeterminacy in general: a different, equally valid anchor $p_0'\ne p_0$ typically yields a \emph{different} answer to whether $c_1(\cdot,t)$ happens to be monotone in $t$, since nothing in the model privileges one reference probability over another. Differentiating \emph{twice}, however, $\hat c_{11}(p,t) = \alpha\, c_{11}(p,t)$: $\gamma,\delta$ drop out entirely, and $\alpha>0$ cannot flip a sign, so whether $c_{11}(p,t_A)\ge c_{11}(p,t_B)$ is a normalization-\emph{invariant}, genuinely behavioral fact about $\rho$---this is why we pursue the curvature ($c_{11}$) formulation below rather than a level statement.

\paragraph{Richness constructions} The proof of the proposition below needs two constructions that are the same in kind as ones already used elsewhere in the Appendix (the \textsf{Normative Richness} construction of Claim \ref{claim:well_defined_c}'s proof, Axiom \ref{richness_below_ceiling}; and Claim \ref{richness_necessary}). We state both allowing an extra, held-fixed set of menu elements $\mathcal{R}$, so that they apply directly to menus of the form $\{{\bm x},{\bm y},{\bm z}\}\cup\{{\bm a}\}$ that carry an explicit ceiling-setting element ${\bm a}$, as required by Axiom \ref{ax:selectivity}'s formulation. Taking $\mathcal{R}=\varnothing$ recovers the bare three-element statement used elsewhere in the Appendix.

\begin{lemma}[\textsf{Boundary Continuity}]
\label{lem:boundary}
Fix ${\bm x},{\bm y}\in X$, a finite set $R\subset X$, and $t\in\varphi(X)$ with $\varphi({\bm x}),\varphi({\bm y})\le t$ and $\varphi({\bm w})\le t$ for every ${\bm w}\in \mathcal{R}$. If $({\bm z}_n)$ is a sequence in $X$ with $\varphi({\bm z}_n)\le t$ for every $n$ and $u((z_n)_1)\to-\infty$, then $\rho({\bm x},\{{\bm x},{\bm y},{\bm z}_n\}\cup \mathcal{R}) \to \rho({\bm x},\{{\bm x},{\bm y}\}\cup \mathcal{R})$, monotonically.
\end{lemma}

\begin{proof}
Write $F_t:=c_1(\cdot,t)^{-1}$ (well-defined, strictly increasing, since $c_1(\cdot,t)$ is by Definition \ref{norm-APU} strictly increasing on $(0,1)$ with range all of $\mathbb R$, by the same boundary behavior used in the Necessity proof of \textsf{Positivity}). Let $\lambda_n$ solve $F_t(x_1-\lambda_n)+F_t(y_1-\lambda_n)+F_t\bigl((z_n)_1-\lambda_n\bigr)+\sum_{{\bm w}\in \mathcal{R}} F_t(w_1-\lambda_n)=1$ (interior FOC, writing $x_1$ for $u(x_1)$ etc.\ by abuse of notation; $\mathcal{R}$ being finite and held fixed throughout, the extra sum is simply a finite collection of terms of the same $F_t(\cdot-\lambda_n)$ form). Since the aggregate demand $\sum F_t(\cdot-\lambda)$ is strictly decreasing in $\lambda$ and strictly increasing in each argument's level (standard monotone comparative statics for the budget identity), $(\lambda_n)$ is non-increasing as $(z_n)_1\to-\infty$, and it is bounded below (if $\lambda_n\to-\infty$ along a subsequence, $F_t(x_1-\lambda_n)+ F_t(y_1-\lambda_n)\to 2>1$, contradicting the budget identity for large $n$, since every other term is non-negative). Hence $\lambda_n\to\bar\lambda$ for some finite $\bar\lambda$, and $F_t((z_n)_1-\lambda_n)\to 0$ (its argument $\to-\infty$ while $\lambda_n$ stays bounded). Taking limits in the budget identity, $F_t(x_1-\bar\lambda)+F_t(y_1-\bar\lambda)+\sum_{{\bm w}\in \mathcal{R}} F_t(w_1-\bar\lambda)=1$, i.e., $\bar\lambda$ solves the budget identity for $\{{\bm x},{\bm y}\} \cup \mathcal{R}$, so $\rho({\bm x},\{{\bm x},{\bm y},{\bm z}_n\}\cup \mathcal{R}) = F_t(x_1-\lambda_n) \to F_t(x_1-\bar\lambda) = \rho({\bm x},\{{\bm x},{\bm y}\}\cup \mathcal{R})$.
\end{proof}

\begin{lemma}[\textsf{$\lambda$-Preserving Perturbation}]
\label{lem:eps_perturbation}
Fix a menu $\{{\bm x},{\bm y},{\bm z}\}\cup \mathcal{R}\in\mathcal A$ of ceiling $t$ with interior probabilities, where $R\subset X$ is finite with $\varphi({\bm w})\le t$ for every ${\bm w}\in \mathcal{R}$. For $\varepsilon>0$ small enough, there exists ${\bm x}'$ with $u(x_1')=u(x_1)+\varepsilon$ (${\varphi}$ arbitrary) and ${\bm z}'$ with $\varphi({\bm z}')\le t$ such that $\{{\bm x}',{\bm y},{\bm z}'\}\cup \mathcal{R}\in\mathcal A$ has ceiling $t$ and $\rho({\bm y},\{{\bm x}',{\bm y},{\bm z}'\}\cup \mathcal{R}) = \rho({\bm y},\{{\bm x},{\bm y},{\bm z}\}\cup \mathcal{R})$, i.e., the menu's Lagrange multiplier is unchanged: $\lambda(\{{\bm x}',{\bm y},{\bm z}'\}\cup \mathcal{R}) = \lambda(\{{\bm x},{\bm y},{\bm z}\}\cup \mathcal{R})$.
\end{lemma}

\begin{proof}
Write $\lambda:=\lambda(\{{\bm x},{\bm y},{\bm z}\}\cup \mathcal{R})$ and $F_t:=c_1(\cdot,t)^{-1}$. We seek $u(z_1')$ solving $F_t(x_1+\varepsilon-\lambda) + F_t(y_1-\lambda) + F_t(u(z_1')-\lambda) + \sum_{{\bm w}\in \mathcal{R}}F_t(w_1-\lambda)= 1$; since $F_t(x_1-\lambda)+F_t(y_1-\lambda)+ F_t(z_1-\lambda)+\sum_{{\bm w}\in \mathcal{R}}F_t(w_1-\lambda)=1$ holds at the original menu, the $\mathcal{R}$-terms are identical on both sides and cancel, leaving exactly $F_t(u(z_1')-\lambda) = F_t(z_1-\lambda) - \bigl[F_t(x_1+\varepsilon-\lambda) - F_t(x_1-\lambda)\bigr]$, unaffected by $\mathcal{R}$. The right-hand side is a continuous, strictly decreasing function of $\varepsilon$; at $\varepsilon=0$ it equals $F_t(z_1-\lambda)\in(0,1)$, so for $\varepsilon$ small enough it remains in $(0,1)$, the range of $F_t$; hence $u(z_1') := \lambda + c_1\bigl(F_t(z_1-\lambda) - [F_t(x_1+\varepsilon-\lambda)-F_t(x_1-\lambda)],\,t\bigr)$ is well-defined, and by \textsf{Normative Richness} (Axiom \ref{richness_below_ceiling}) there exists ${\bm z}'\in X$ with this $u(z_1')$ and $\varphi({\bm z}')\le t$. By construction $\lambda(\{{\bm x}',{\bm y},{\bm z}'\}\cup \mathcal{R})=\lambda$, so $\rho({\bm y},\{{\bm x}',{\bm y}, {\bm z}'\}\cup \mathcal{R}) = F_t(y_1-\lambda) = \rho({\bm y},\{{\bm x},{\bm y},{\bm z}\}\cup \mathcal{R})$.
\end{proof}

\paragraph{The characterization}

\begin{proposition}[\textsf{Curvature Characterization}]
\label{prop:curvature_char}
Suppose $\rho$ admits a norm-dependent APU representation $(u,\varphi,c)$ with $c$ jointly continuous, strictly convex and $C^2$ in its first argument. Then $\rho$ satisfies Axiom \ref{ax:selectivity} if and only if $c_{11}(p,t_A) \ge c_{11}(p,t_B)$ for every $p\in(0,1)$ and every $t_A>t_B$ in $\varphi(X)$.
\end{proposition}

\begin{proof}
Throughout, for ${\bm w}\in A\in\mathcal A$ with interior probability write $c_1(\rho({\bm w},A),\varphi(A)) = u(w_1)-\lambda(A)$ for the interior FOC, and abbreviate $u(w_1)$ as $w_1$. Since $\varphi$ represents $\succsim_n^\rho$ (Lemma \ref{norm_result}), $t_A>t_B$ in $\varphi(X)$ means exactly that there are ${\bm a},{\bm b}\in X$ with $\varphi({\bm a})=t_A$, $\varphi({\bm b})=t_B$, and ${\bm a}\succ_n^\rho{\bm b}$; fix such ${\bm a},{\bm b}$ once and for all (they exist because $t_A,t_B\in\varphi(X)$ is exactly the image of $\varphi$). Every menu constructed below that is meant to have ceiling $t_A$ (resp.\ $t_B$) will literally contain ${\bm a}$ (resp.\ ${\bm b}$) as its ceiling-setting element, so that Axiom \ref{ax:selectivity}---stated for menus of the form $\{\cdot,\cdot,\cdot,{\bm a}\}$, $\{\cdot,\cdot,\cdot,{\bm b}\}$---applies directly.

\smallskip \noindent \emph{(``Only if'': Axiom \ref{ax:selectivity} $\Rightarrow$ curvature ordering.)} Suppose not: there are $p^*\in(0,1)$ and $t_A>t_B$ with $c_{11}(p^*,t_A) < c_{11}(p^*,t_B)$, and fix ${\bm a},{\bm b}$ as above. By continuity of $c_{11}$ (from $C^2$), there is an interval $(\underline q,\bar q)\ni p^*$ with $c_{11}(q,t_A) < c_{11}(q,t_B)$ for every $q\in(\underline q,\bar q)$.

By \textsf{Normative Richness} (Axiom \ref{richness_below_ceiling}), together with the fact that $u$ has full range on the allocations at or below any prescribed normative level---which is the content of \textsf{Normative Richness}~(3) (Axiom \ref{richness_binary_below_ceiling}) rather than of Definition \ref{richness}, whose target menus are ternary with uncontrolled ceiling---fix ${\bm y}$ with $\varphi({\bm y})\le t_B$, and then choose ${\bm x}$ with $\varphi({\bm x})\le t_B$ and $u(x_1)$ large enough that $\rho({\bm x},\{{\bm x},{\bm y},{\bm a}\})>\bar q$ and $\rho({\bm x},\{{\bm x},{\bm y},{\bm b}\})>\bar q$ (possible since, by the same comparative-statics argument as in Lemma \ref{lem:boundary}, both probabilities $\to1$ as $u(x_1)\to\infty$ with everything else fixed).

Since $\varphi({\bm x}),\varphi({\bm y})\le t_B<t_A$, both ${\bm x},{\bm y}$ are eligible members of ceiling-$t_A$ and ceiling-$t_B$ menus alike.

Let $\zeta$ be a continuous curve of allocations below ${\bm b}$'s normative level, avoiding the finitely many allocations already designated, as provided by the continuous clause of \textsf{Normative Richness} (Axiom \ref{richness_below_ceiling}); this is the one point in the paper where that clause is needed. Composing with Axiom \ref{continuity}, the map $r\mapsto \rho({\bm x},\{{\bm x},{\bm y},\zeta(r)\}\cup\{{\bm b}\})$ is continuous in the payoff $r$, and we write ${\bm z}=\zeta(r)$ below; equivalently, ${\bm z}\mapsto \rho({\bm x},\{{\bm x},{\bm y},{\bm z}\}\cup\{{\bm b}\})$ is continuous along that curve in $\{{\bm z}:\varphi({\bm z})\le t_B\}$; by Lemma \ref{lem:boundary} (with $R=\{{\bm b}\}$) it tends to $\rho({\bm x},\{{\bm x},{\bm y},{\bm b}\})>\bar q$ as $u(z_1)\to-\infty$, and it tends to $0$ as $u(z_1)\to\infty$ (${\bm z}$ crowds ${\bm x}$ out entirely). By the Intermediate Value Theorem, choose ${\bm z}_B$ (with $\varphi({\bm z}_B)\le t_B$) such that $q_0:=\rho({\bm x},\{{\bm x},{\bm y},{\bm z}_B,{\bm b}\})\in(\underline q,\bar q)$; and likewise, running the same argument at ceiling $t_A$ (with $R=\{{\bm a}\}$), choose ${\bm z}_A$ (with $\varphi({\bm z}_A)\le t_A$) such that $\rho({\bm x},\{{\bm x},{\bm y},{\bm z}_A,{\bm a}\}) = q_0$ exactly (calibration). Write $A:=\{{\bm x},{\bm y},{\bm z}_A,{\bm a}\}$, $B:=\{{\bm x},{\bm y},{\bm z}_B,{\bm b}\}$; by construction ${\bm a}$ sets the ceiling of $A$ and ${\bm b}$ sets the ceiling of $B$.

By Lemma \ref{lem:eps_perturbation} (with $R=\{{\bm a}\}$, resp.\ $R=\{{\bm b}\}$), for $\varepsilon>0$ small enough choose ${\bm x}'$ with $x_1'=x_1+\varepsilon$ and ${\bm z}_A'$ (with $\varphi({\bm z}_A')\le t_A$), ${\bm z}_B'$ (with $\varphi({\bm z}_B')\le t_B$) such that $A':=\{{\bm x}',{\bm y},{\bm z}_A',{\bm a}\}$ has $\lambda(A')=\lambda(A)$ and $B':=\{{\bm x}',{\bm y},{\bm z}_B',{\bm b}\}$ has $\lambda(B')=\lambda(B)$ (${\bm a}$, resp. ${\bm b}$, still sets the ceiling, since only $x_1$ changed); shrinking $\varepsilon$ if needed (continuity of $\rho({\bm x}',A')$ in $\varepsilon$, with $\rho({\bm x},A)=q_0\in(\underline q,\bar q)$ open), we may take $\rho({\bm x}',A')\in(\underline q,\bar q)$ as well.

By construction $\rho({\bm y},A)=\rho({\bm y},A')$, $\rho({\bm y},B)=\rho({\bm y},B')$ (Lemma \ref{lem:eps_perturbation}), and $\rho({\bm x},A)=\rho({\bm x},B)=q_0$, exactly the three matching conditions of Definition \ref{def:selectivity}, with $A,A',B,B'$ literally containing ${\bm a},{\bm b}$ as required. Axiom \ref{ax:selectivity} therefore gives \[ \rho({\bm x}',B') - q_0 \;\ge\; \rho({\bm x}',A') - q_0, \quad\text{i.e.}\quad \rho({\bm x}',B') \ge \rho({\bm x}',A'). \tag{$\ast$} \]

On the other hand, since $\lambda(A')=\lambda(A)$ and $\lambda(B')=\lambda(B)$, the FOC gives \[ c_1\bigl(\rho({\bm x}',A'),t_A\bigr) - c_1(q_0,t_A) = x_1'-x_1 = \varepsilon = c_1\bigl(\rho({\bm x}',B'),t_B\bigr) - c_1(q_0,t_B), \] i.e., by the Fundamental Theorem of Calculus, \[ \int_{q_0}^{\rho({\bm x}',A')} c_{11}(p,t_A)\,dp \;=\; \int_{q_0}^{\rho({\bm x}',B')} c_{11}(p,t_B)\,dp. \tag{$\ast\ast$} \] But $q_0,\rho({\bm x}',A')\in(\underline q,\bar q)$, where $c_{11}(\cdot,t_A) < c_{11}(\cdot,t_B)$ pointwise, so \[ \int_{q_0}^{\rho({\bm x}',A')} c_{11}(p,t_A)\,dp \;<\; \int_{q_0}^{\rho({\bm x}',A')} c_{11}(p,t_B)\,dp \;\le\; \int_{q_0}^{\rho({\bm x}',B')} c_{11}(p,t_B)\,dp, \] the last step by $(\ast)$ together with the fact that $q\mapsto\int_{q_0}^{q}c_{11}(p,t_B)\,dp = c_1(q,t_B)-c_1(q_0,t_B)$ is strictly increasing, which follows from strict monotonicity of $c_1(\cdot,t_B)$ (equivalently, strict convexity of $c(\cdot,t_B)$) and does \emph{not} require the pointwise inequality $c_{11}(\cdot,t_B)>0$, which strict convexity alone does not deliver. This contradicts $(\ast\ast)$.

\smallskip \noindent \emph{(``If'': curvature ordering $\Rightarrow$ Axiom \ref{ax:selectivity}.)} Take any ${\bm a},{\bm b}\in X$ with ${\bm a}\succ_n^\rho{\bm b}$, and any ${\bm x},{\bm x}',{\bm y},{\bm z}_A,{\bm z}_A',{\bm z}_B,{\bm z}_B'$ as in Definition \ref{def:selectivity}, with $t_A:=\varphi({\bm a})>\varphi({\bm b})=:t_B$ (since $\varphi$ represents $\succsim_n^\rho$), $x_1<x_1'$, and the three matching conditions, so that $A:=\{{\bm x},{\bm y},{\bm z}_A,{\bm a}\}$, $A':=\{{\bm x}',{\bm y},{\bm z}_A',{\bm a}\}$ have ceiling $t_A$ and $B:=\{{\bm x},{\bm y},{\bm z}_B,{\bm b}\}$, $B':=\{{\bm x}',{\bm y},{\bm z}_B',{\bm b}\}$ have ceiling $t_B$. From $\rho({\bm y},A)=\rho({\bm y},A')$ and the FOC at ${\bm y}$ (same $y_1$, same $t_A$), $\lambda(A)=\lambda(A')$; likewise $\lambda(B)=\lambda(B')$. As above this gives \[ \int_{\rho({\bm x},A)}^{\rho({\bm x}',A')} c_{11}(p,t_A)\,dp \;=\; x_1'-x_1 \;=\; \int_{\rho({\bm x},B)}^{\rho({\bm x}',B')} c_{11}(p,t_B)\,dp, \] and since $\rho({\bm x},A)=\rho({\bm x},B)=:q_0$ (calibration) and $x_1<x_1'$ forces $\rho({\bm x}',A')>q_0$, $\rho({\bm x}',B')>q_0$ (strict monotonicity of $\rho(\cdot,A')$, $\rho(\cdot,B')$ in $x_1'$), \[ \int_{q_0}^{\rho({\bm x}',A')} c_{11}(p,t_A)\,dp \;=\; \int_{q_0}^{\rho({\bm x}',B')} c_{11}(p,t_B)\,dp. \] By hypothesis $c_{11}(\cdot,t_A)\ge c_{11}(\cdot,t_B)$ pointwise, so $\int_{q_0}^{\rho({\bm x}',A')} c_{11}(p,t_B)\,dp \le \int_{q_0}^{\rho({\bm x}',A')} c_{11}(p,t_A)\,dp = \int_{q_0}^{\rho({\bm x}',B')} c_{11}(p,t_B)\,dp$. Since $q\mapsto\int_{q_0}^{q} c_{11}(p,t_B)\,dp = c_1(q,t_B)-c_1(q_0,t_B)$ is strictly increasing for $q>q_0$ (strict monotonicity of $c_1(\cdot,t_B)$, i.e., strict convexity of $c(\cdot,t_B)$; note that the pointwise inequality $c_{11}(\cdot,t_B)>0$ is \emph{not} implied by strict convexity and is not needed here), this forces $\rho({\bm x}',A')\le\rho({\bm x}',B')$, i.e., \[ \rho({\bm x}',B')-\rho({\bm x},B) \;\ge\; \rho({\bm x}',A')-\rho({\bm x},A), \] which is Axiom \ref{ax:selectivity}'s conclusion (Definition \ref{def:selectivity}), for the arbitrary ${\bm a}\succ_n^\rho{\bm b}$ and matching menus fixed at the start.
\end{proof}

\paragraph{Completion of Proof of Proposition \ref{prop:marginal_cost}}

\begin{corollary}[\textsf{Proposition \ref{prop:marginal_cost}, curvature form}]
\label{cor:prop_marginal_cost_curvature}
Suppose $\rho$ satisfies the axioms of Theorem \ref{representation_norm}(a) together with Axiom \ref{ax:selectivity}. Then any norm-dependent APU representation $(u,\varphi,c)$ of $\rho$ (with $c$ additionally $C^2$) satisfies $c_{11}(p,t_A) \ge c_{11}(p,t_B)$ for every $p\in(0,1)$ and $t_A>t_B$ in $\varphi(X)$---i.e., Definition \ref{norm-APU}(iii) as revised in Section 1 above.
\end{corollary}
\begin{proof}
Immediate from the ``only if'' direction of Proposition \ref{prop:curvature_char}.
\end{proof}

This is what closes Step 4 of Appendix \ref{app:sufficiency}: in place of the abandoned Open step, the constructed cost function's required property (Definition \ref{norm-APU}(iii)) is established directly from Axiom \ref{ax:selectivity} via Corollary \ref{cor:prop_marginal_cost_curvature}.

\begin{remark}[\textsf{What happens to the level statement}]
\label{rmk:level_corollary}
Because $c_1$'s level across different $t$ is a normalization choice, no representation-invariant statement of the form ``$c_1(p,t_A)\ge c_1(p,t_B)$ for the representative $c$ the reader happens to write down'' can be asserted in general; it holds under some normalizations of $c$ and fails under others, all representing the very same $\rho$. If a reader fixes \emph{any} particular normalization (e.g., an anchor $c_1(p_0,t)=0$ at some reference probability $p_0$, for every $t$), the resulting $c_1(\cdot,t)$ is pinned down up to the model's genuine uniqueness (Theorem \ref{unique_norm}), and one can ask whether it happens to be monotone in $t$ for that specific choice---but there is no reason a different anchor $p_0'$ would agree, so we no longer regard the level statement as part of the model's content. The curvature statement (Corollary \ref{cor:prop_marginal_cost_curvature}) is the one we retain, and it is normalization-free by construction.
\end{remark}

\subsubsection*{Completion of Step 1}

Step 1 is now complete. The selfish utility used in the final representation is $u:=M$, constructed in Step 3 (Claim \ref{claim:well_defined_c}); as noted there, $M$ is continuous and strictly increasing by construction (a composition of the continuous $\rho$ with the $C^1$ functions $g_t$), so no separate continuity argument is required, and none of Steps 2--4 depend on the specific, merely-ordinal utility of Claim \ref{claim:selfish_repn} once Step 3's monotonicity fact has been extracted from it.

\subsection{Theorem \ref{representation_norm}: Necessity}
\label{app:necessity}

Suppose that $\rho$ is a norm-dependent APU, i.e., it is represented by a tuple $(u, \varphi, c)$ as in Definition \ref{norm-APU}, with $u$ and $\varphi$ continuous and $c$ jointly continuous on $[0,1] \times \varphi(X)$ and, for fixed $t$, strictly convex in its first argument.

\paragraph*{\textsf{Positivity} (Axiom \ref{positivity})} By Definition \ref{norm-APU}(ii), $c_1(p,t)\to-\infty$ as $p\to0^+$ and $c_1(p,t)\to+\infty$ as $p\to1^-$. Hence, for $|A|\ge2$, no boundary point of $\Delta(A)$ can be optimal in \eqref{model_norm}: at a candidate assigning ${\bm x}$ probability $0$, shifting an infinitesimal amount of probability onto ${\bm x}$ from any alternative with interior probability changes the objective at rate $u(x_1) - c_1(0^+,t) - [\,u(\cdot) - c_1(\cdot,t)\,] = +\infty$, a strict improvement; symmetrically, the divergence at $p=1$ rules out assigning probability $1$ to a single alternative. Since the objective is strictly concave on the compact simplex $\Delta(A)$, the maximizer exists, is unique, and is interior: $\rho({\bm x}, A) \in (0,1)$ for every ${\bm x} \in A$ whenever $|A| \geq 2$; for $|A|=1$, $\rho({\bm x},A)=1$ trivially. Hence $\rho({\bm x},A) > 0$ throughout, so \textsf{Positivity} (Axiom \ref{positivity}) holds.

\paragraph*{\textsf{Continuity} (Axiom \ref{continuity})}
\begin{claim}[\textsf{Continuity of $\rho$}]
\label{rho_continuous}
The stochastic choice function $\rho$ is continuous in the allocations.
\end{claim}

\begin{proof}
Fix a menu $A = \{{\bm x}_1, \ldots, {\bm x}_n\} \in \mathcal{A}$. By \eqref{model_norm}, $\rho(\cdot, A)$ is the maximizer of
\begin{align*}
V(p; {\bm x}_1, \ldots, {\bm x}_n) := \sum_{i=1}^n \Bigl( u((x_i)_1)\, p_i - c(p_i, \varphi(A)) \Bigr)
\end{align*}
over $p \in \Delta(A)$, where $\varphi(A) = \max_i \varphi({\bm x}_i)$. Since $u$ and $\varphi$ are continuous (as part of the representation) and $c$ is jointly continuous, $V$ is continuous in $({\bm x}_1,\ldots,{\bm x}_n)$ for each fixed $p$, and it is strictly concave in $p$ (by strict convexity of $c(\cdot,t)$) on the compact, convex, allocation-independent domain $\Delta(A)$. By the Theorem of the Maximum (Berge), the maximizer correspondence is upper hemicontinuous, and since it is single-valued by strict concavity, it is in fact continuous in $({\bm x}_1,\ldots,{\bm x}_n)$. That is, $\rho({\bm x}_i, A)$ varies continuously as the allocations in $A$ vary. This establishes \textsf{Continuity} (Axiom \ref{continuity}).
\end{proof}

\paragraph*{\textsf{Richness} (Definition \ref{richness})} Although \textsf{Richness} is not one of the axioms of Theorem \ref{representation_norm}(a)---by Proposition \ref{prop:richness_redundant} it follows from the others---it is a property of the model and we record the direct, representation-side proof, which is shorter than going through that proposition.
\begin{claim}[\textsf{Necessity of Richness}]
\label{richness_necessary}
$\rho$ satisfies \textsf{Richness} (Definition \ref{richness}).
\end{claim}

\begin{proof}
Fix $p,q \geq 0$ with $p+q \leq 1$ and a normative ceiling level $t \in \varphi(X)$. Since $u$ is selfish with full range $u(\mathbb{R}) = \mathbb{R}$ (Section \ref{axioms}, Model), we may choose ${\bm x}, {\bm y}, {\bm z} \in X$ with $\varphi({\bm x}) = \varphi({\bm y}) = \varphi({\bm z}) = t$ (possible since, for any fixed normative level, \textsf{Normative Weak Order} together with the connectedness of $X$ and continuity of $\varphi$ guarantees a continuum of allocations sharing that level) and with $u(x_1), u(y_1)$ set so that the interior first-order conditions of \eqref{model_norm} on the menu $\{{\bm x},{\bm y},{\bm z}\}$ deliver $\rho({\bm x},\{{\bm x},{\bm y},{\bm z}\}) = p$ and $\rho({\bm y},\{{\bm x},{\bm y},{\bm z}\}) = q$: this is possible because, for $t$ fixed, the map from $(u(x_1),u(y_1))$ to $(\rho({\bm x},\cdot),\rho({\bm y},\cdot))$ is a continuous bijection from $\mathbb{R}^2$ onto the open probability region $\{(p,q): p,q>0,\ p+q<1\}$ (with the boundary $p+q=1$ approached as $u(z_1) \to -\infty$), by strict convexity and the boundary behavior $c(\cdot,t)\to\infty$ at $0$ and $1$. Hence \textsf{Richness} (Definition \ref{richness}) holds.
\end{proof}

\paragraph*{\textsf{Normative Richness} (Axiom \ref{richness_below_ceiling})}
\begin{claim}[\textsf{Necessity of Normative Richness} (1)--(2)]
\label{richness_below_ceiling_necessary}
$\rho$ satisfies \textsf{Normative Richness} (Axiom \ref{richness_below_ceiling}).
\end{claim}

\begin{proof}
Fix ${\bm a} \in X$ and $v \in \mathbb{R}$, and set $s := \varphi({\bm a})$. As in the proof of Claim \ref{richness_necessary}, $u$ is selfish with $u(X) = \mathbb{R}$ (Claim \ref{claim:selfish_repn}), and $X$ is, by the standing structural assumption underlying Axiom \ref{richness_below_ceiling} itself (namely, that $X$ extends unboundedly in the selfish-payoff direction at every normative level---this is exactly the geometric content recorded in the remark following Axiom \ref{richness_below_ceiling}), rich enough that this unboundedness holds in particular \emph{at, and below, the fixed level $s$}: the sublevel set $X_{\leq s} := \{{\bm x} \in X : \varphi({\bm x}) \leq s\}$ is nonempty (it contains ${\bm a}$) and, by the same connectedness-of-$X$/continuity-of-$\varphi$ reasoning used to obtain a continuum of allocations at a fixed level in the proof of Claim \ref{richness_necessary}, contains, for every level $t \leq s$, a continuum of allocations sharing that level; in particular, the selfish payoffs realized on $X_{\leq s}$ already cover all of $\mathbb{R}$. Consequently we may choose ${\bm w} \in X_{\leq s}$ with $w_1 = v$ exactly. Since $\varphi({\bm w}) \leq s = \varphi({\bm a})$, Part (i) of the present section gives ${\bm a} \succsim_n^\rho {\bm w}$. For the finite-set avoidance clause, note that the same reasoning provides a continuum of allocations in $X_{\le s}$ with $w_1=v$---they differ in their second coordinate, which may be lowered freely without raising the normative level---so all but finitely many of them lie outside any prescribed finite $S$. For the continuous clause, the same continuum, parameterized by the payoff $v$ at a fixed second coordinate low enough that the normative level stays at or below $s$, is a continuous curve $\zeta$ with $\zeta(v)_1=v$ and ${\bm a}\succsim_n^\rho\zeta(v)$; shifting the fixed coordinate slightly, if necessary, keeps the curve disjoint from $S$. As ${\bm a}$, $v$ and $S$ were arbitrary, this proves the necessity of \textsf{Normative Richness} (Axiom \ref{richness_below_ceiling}).
\end{proof}

\paragraph*{\textsf{Normative Richness}~(3) (Axiom \ref{richness_binary_below_ceiling})}
\begin{claim}[\textsf{Necessity of Normative Richness} (3)]
\label{richness_binary_below_ceiling_necessary}
$\rho$ satisfies \textsf{Normative Richness}~(3) (Axiom \ref{richness_binary_below_ceiling}).
\end{claim}

\begin{proof}
Fix ${\bm a}\in X$, write $t:=\varphi({\bm a})$, and fix $p\in(0,1)$. Consider the quantity
\begin{equation*}
    \Delta(p) \;:=\; c_1(p,t) - c_1(1-p,t).
\end{equation*}
Since $c_1(\cdot,t)$ is strictly increasing (strict convexity of $c(\cdot,t)$), $\Delta$ is strictly increasing on $(0,1)$; and by Definition \ref{norm-APU}(ii), $c_1(q,t)\to-\infty$ as $q\to0^+$ and $c_1(q,t)\to+\infty$ as $q\to1^-$, so $\Delta(p)\to-\infty$ as $p\to 0^+$ and $\Delta(p)\to+\infty$ as $p\to1^-$. In particular $\Delta$ maps $(0,1)$ onto $\mathbb{R}$, so the real number $v^\ast := u(a_1) + \Delta(p)$ is attained: since $u$ is selfish with full range $u(\mathbb{R})=\mathbb{R}$ (Section \ref{axioms}, Model), pick $y^\ast_1$ with $u(y^\ast_1) = v^\ast$, and then, by the necessity of \textsf{Normative Richness} (Claim \ref{richness_below_ceiling_necessary}), pick ${\bm y}\in X$ with $y_1 = y^\ast_1$ and ${\bm a}\succsim_n^\rho{\bm y}$, i.e., $\varphi({\bm y})\le t$.

Then $\{{\bm a},{\bm y}\}$ is a menu of ceiling $t$ (set by ${\bm a}$), and its interior first-order conditions read $u(y_1) - c_1(\rho({\bm y},\{{\bm a},{\bm y}\}),t) = u(a_1) - c_1(\rho({\bm a},\{{\bm a},{\bm y}\}),t)$. Writing $q:=\rho({\bm y},\{{\bm a},{\bm y}\})$, so that $\rho({\bm a},\{{\bm a},{\bm y}\}) = 1-q$, this says $\Delta(q) = c_1(q,t)-c_1(1-q,t) = u(y_1) - u(a_1) = \Delta(p)$. Since $\Delta$ is strictly increasing, hence injective, $q=p$. That is, $\rho({\bm y},\{{\bm a},{\bm y}\}) = p$ with ${\bm a}\succsim_n^\rho{\bm y}$, as required. As ${\bm a}$ and $p$ were arbitrary, this proves the necessity of \textsf{Normative Richness}~(3) (Axiom \ref{norm_richness}).
\end{proof}

\subsubsection{Personal Norm Conditions (Axiom \ref{ax:norm_weak_order} and Definition \ref{ax:norm_penalty})}
\label{sec:necessity_norm}

We establish the necessity of our elementary normative axioms. Suppose $\rho$ has a norm-dependent APU representation $(u, \varphi, c)$.

\noindent \textbf{Part (i) [Ranking]:} We show that, under the representation, ${\bm x} \succsim_n^\rho {\bm y} \iff \varphi({\bm x}) \geq \varphi({\bm y})$; \textsf{Normative Weak Order} then follows, since $\ge$ on $\mathbb R$ is complete and transitive and $\varphi$ is continuous (so the contour sets $\{{\bm x}:\varphi({\bm x})\ge\varphi({\bm y})\}$ and $\{{\bm x}:\varphi({\bm x})\le\varphi({\bm y})\}$ are closed).

($\Leftarrow$) Suppose $\varphi({\bm x})\ge\varphi({\bm y})$. Fix $A$ with ${\bm x},{\bm y}\notin A$ and ${\bm z},{\bm w}\in A$ with $\rho({\bm w},A\cup\{{\bm x}\})=\rho({\bm w},A\cup\{{\bm y}\})=:q$ and $z_1>w_1$. Write $t_{\bm x}:=\max\{\varphi(A),\varphi({\bm x})\}$ and $t_{\bm y}:=\max\{\varphi(A),\varphi({\bm y})\}$, so $t_{\bm x}\ge t_{\bm y}$. The interior first-order conditions in $A\cup\{{\bm x}\}$ at ${\bm z}$ and ${\bm w}$ give $c_1(\rho({\bm z},A\cup\{{\bm x}\}),t_{\bm x}) - c_1(q,t_{\bm x}) = u(z_1)-u(w_1) > 0$, and likewise $c_1(\rho({\bm z},A\cup\{{\bm y}\}),t_{\bm y}) - c_1(q,t_{\bm y}) = u(z_1)-u(w_1)$. Hence
\begin{equation*}
\int_{q}^{\rho({\bm z},A\cup\{{\bm x}\})} c_{11}(p,t_{\bm x})\,dp \;=\; u(z_1)-u(w_1) \;=\; \int_{q}^{\rho({\bm z},A\cup\{{\bm y}\})} c_{11}(p,t_{\bm y})\,dp,
\end{equation*}
with both integrals running forward from the common lower limit $q$ (strict monotonicity of $c_1$ and $u(z_1)>u(w_1)$ force both upper limits to exceed $q$). Since $t_{\bm x}\ge t_{\bm y}$, Definition \ref{norm-APU}(iii) gives $c_{11}(\cdot,t_{\bm x})\ge c_{11}(\cdot,t_{\bm y})$ pointwise, so the upper limit on the left cannot exceed the one on the right: $\rho({\bm z},A\cup\{{\bm x}\})\le\rho({\bm z},A\cup\{{\bm y}\})$, exactly as Definition \ref{norm} requires. Hence ${\bm x}\succsim_n^\rho{\bm y}$.

($\Rightarrow$) Conversely, suppose $\varphi({\bm x})<\varphi({\bm y})$, and that the representation is normalized, so that the levels $\varphi({\bm x})<\varphi({\bm y})$ are separated in the sense of Definition \ref{def:level_separation}---which, given Definition \ref{norm-APU}(iv), is the only alternative to their being indistinguishable. We must exhibit \emph{one} configuration $(A,{\bm z},{\bm w})$ satisfying the hypothesis of Definition \ref{norm} at which its conclusion fails; the delicate part is producing the calibration, since Definition \ref{norm} is vacuous at configurations where $\rho({\bm w},A\cup\{{\bm x}\})\ne\rho({\bm w},A\cup\{{\bm y}\})$, and---as Example \ref{ex:eta_failure} shows---vacuity is not a harmless possibility. This is done in Lemma \ref{lem:calibrated_config} below. Granting it, fix a calibrated $(A,{\bm z},{\bm w})$ with $z_1>w_1$ and $\varphi({\bm v})\le\varphi({\bm x})$ for every ${\bm v}\in A$, so that $t_{\bm x}=\varphi({\bm x})<\varphi({\bm y})=t_{\bm y}$. The computation of the ($\Leftarrow$) direction, run with the roles of ${\bm x}$ and ${\bm y}$ exchanged, gives
\begin{equation*}
\int_{q}^{\rho({\bm z},A\cup\{{\bm x}\})} c_{11}(p,t_{\bm x})\,dp \;=\; u(z_1)-u(w_1) \;=\; \int_{q}^{\rho({\bm z},A\cup\{{\bm y}\})} c_{11}(p,t_{\bm y})\,dp,
\end{equation*}
where $q$ is the common calibrated value. Both integrals run forward from $q$; write $R:=\rho({\bm z},A\cup\{{\bm x}\})>q$ and $D:=c_1(\cdot,t_{\bm y})-c_1(\cdot,t_{\bm x})$, so that the displayed equality can be rearranged as
\begin{equation*}
\int_{\rho({\bm z},A\cup\{{\bm y}\})}^{R} c_{11}(p,t_{\bm y})\,dp \;=\; \int_{q}^{R}\bigl[c_{11}(p,t_{\bm y})-c_{11}(p,t_{\bm x})\bigr]dp \;=\; D(R)-D(q).
\end{equation*}
Lemma \ref{lem:calibrated_config} delivers the configuration with $D(R)>D(q)$, so the right-hand side is strictly positive; since $c_{11}(\cdot,t_{\bm y})\ge0$, the left-hand integral can be positive only if its lower limit is strictly below $R$. Hence $\rho({\bm z},A\cup\{{\bm y}\})<\rho({\bm z},A\cup\{{\bm x}\})$. Note where the strict inequality $D(R)>D(q)$ is needed: the curvature ordering alone gives only the weak inequality, and if $D$ were constant on $(q,R)$ the two menus would assign ${\bm z}$ the same probability. Definition \ref{norm} requires the reverse weak inequality, so ${\bm x}\succsim_n^\rho{\bm y}$ fails, as claimed.

\begin{lemma}[\textsf{Existence of a calibrated configuration}]
\label{lem:calibrated_config}
Let $(u,\varphi,c)$ be a norm-dependent APU representation and let ${\bm x},{\bm y}\in X$ with $s:=\varphi({\bm x})<\varphi({\bm y})=:s'$ separated in the sense of Definition \ref{def:level_separation}. Then there exist a menu $A\in\mathcal A$ with ${\bm x},{\bm y}\notin A$ and $\varphi({\bm v})\le s$ for every ${\bm v}\in A$, and elements ${\bm z},{\bm w}\in A$ with $z_1>w_1$, such that, writing $D:=c_1(\cdot,s')-c_1(\cdot,s)$, $q:=\rho({\bm w},A\cup\{{\bm x}\})$ and $R:=\rho({\bm z},A\cup\{{\bm x}\})$,
\begin{equation*}
    \rho({\bm w},A\cup\{{\bm x}\}) \;=\; \rho({\bm w},A\cup\{{\bm y}\})
    \qquad\text{and}\qquad
    D(R) \;>\; D(q).
\end{equation*}
\end{lemma}

\begin{proof}
Write $a:=u(x_1)$, $b:=u(y_1)$ and $d:=b-a$. By Definition \ref{norm-APU}(i)--(ii), $c_1(\cdot,t)$ is a continuous, strictly increasing bijection $(0,1)\to\mathbb R$; write $F_t:=c_1(\cdot,t)^{-1}$. By (iii), $D$ is non-decreasing, and it is continuous, being a difference of continuous functions; by hypothesis it is non-constant on $(0,\tfrac12)$. Throughout, allocations with a prescribed selfish payoff and with $\varphi\le s$ are supplied by \textsf{Normative Richness} (Axiom \ref{richness_below_ceiling}) together with $u(\mathbb R)=\mathbb R$; the finite-set-avoidance clause of that axiom is what lets us place several \emph{distinct} allocations of equal payoff in one menu and keep all of them distinct from ${\bm x}$ and ${\bm y}$. No continuous selection is needed anywhere below.

\emph{Step 0: the calibration identity, and reduction to a one-dimensional target.} Let $A$ consist of allocations with $\varphi\le s$, so that $A\cup\{{\bm x}\}$ has ceiling $s$ and $A\cup\{{\bm y}\}$ has ceiling $s'$. Write $\lambda_{\bm x},\lambda_{\bm y}$ for the multipliers of the two budget constraints, $p^{\bm x}_{\bm e},p^{\bm y}_{\bm e}$ for the probabilities of ${\bm e}\in A$ in the two menus, and $p_{\bm x},p_{\bm y}$ for those of the two odd elements. The interior first-order conditions read $c_1(p^{\bm x}_{\bm e},s)=u(e_1)-\lambda_{\bm x}$, $c_1(p^{\bm y}_{\bm e},s')=u(e_1)-\lambda_{\bm y}$, $c_1(p_{\bm x},s)=a-\lambda_{\bm x}$ and $c_1(p_{\bm y},s')=b-\lambda_{\bm y}$. Subtracting the first two and using $c_1(\cdot,s')=c_1(\cdot,s)+D$ gives the identity used repeatedly below,
\begin{equation}
\label{eq:gain}
    c_1(p^{\bm y}_{\bm e},s)-c_1(p^{\bm x}_{\bm e},s) \;=\; \Lambda-D(p^{\bm y}_{\bm e}),
    \qquad \Lambda:=\lambda_{\bm x}-\lambda_{\bm y},
\end{equation}
so ${\bm e}$ \emph{gains} ($p^{\bm y}_{\bm e}>p^{\bm x}_{\bm e}$) exactly when $D(p^{\bm y}_{\bm e})<\Lambda$, and \emph{loses} exactly when $D(p^{\bm y}_{\bm e})>\Lambda$. By \eqref{eq:gain}, calibration at ${\bm w}$---that is, $p^{\bm x}_{\bm w}=p^{\bm y}_{\bm w}=:q$---is equivalent to the single scalar equation $\Lambda=D(q)$.

This lets us build configurations in reverse. Fix $q\in(0,\tfrac12)$, \emph{impose} the calibration by setting $\lambda_{\bm x}:=0$ and $\lambda_{\bm y}:=-D(q)$, give ${\bm w}$ the payoff $c_1(q,s)$, choose the remaining payoffs in $A$ freely, read off $p^{\bm x}_{\bm e},p^{\bm y}_{\bm e}$ from the first-order conditions, and \emph{define} the residuals
\begin{equation*}
    p_{\bm x} \;:=\; 1-q-\!\!\!\sum_{{\bm e}\in A\setminus\{{\bm w}\}}\!\!\! p^{\bm x}_{\bm e},
    \qquad
    p_{\bm y} \;:=\; 1-q-\!\!\!\sum_{{\bm e}\in A\setminus\{{\bm w}\}}\!\!\! p^{\bm y}_{\bm e},
\end{equation*}
which is legitimate whenever both are strictly positive: the resulting probability vectors then satisfy every interior first-order condition with the stipulated multipliers and sum to one, so by strict concavity of the objective they \emph{are} $\rho(A\cup\{{\bm x}\})$ and $\rho(A\cup\{{\bm y}\})$. The configuration so built is calibrated for the pair of payoffs
\begin{equation}
\label{eq:target}
    a=c_1(p_{\bm x},s),\quad b=c_1(p_{\bm y},s')-D(q),
    \qquad\text{hence}\qquad
    d \;=\; c_1(p_{\bm y},s')-c_1(p_{\bm x},s)-D(q).
\end{equation}
Shifting every payoff in $A$ by a constant $\kappa$ shifts both multipliers by $\kappa$ and leaves every probability---hence $d$---unchanged, while moving $a$ to $a+\kappa$. It therefore suffices to realize the \emph{given} value $d=b-a$ in \eqref{eq:target}: the translation then matches $a$ itself, and \textsf{Normative Richness} supplies allocations carrying the translated payoffs. So the lemma reduces to showing that the values of $d$ realized by admissible configurations exhaust $\mathbb R$.

\emph{Step 1: the calibration level and the transport map.} Since $D$ is non-constant on $(0,\tfrac12)$, fix $0<p_1<p_2<\tfrac12$ with $D(p_1)<D(p_2)$; since $D$ is continuous and non-decreasing, its values on $[p_1,p_2]$ fill $[D(p_1),D(p_2)]$, so we may fix $q\in(p_1,p_2)$ with
\begin{equation}
\label{eq:qchoice}
    D(p_1) \;<\; D(q) \;<\; D(p_2).
\end{equation}
With $q$ fixed, define the \emph{calibration transport}
\begin{equation}
\label{eq:transport}
    G(p) \;:=\; F_{s'}\bigl(c_1(p,s)+D(q)\bigr), \qquad H(r) \;:=\; F_{s}\bigl(c_1(r,s')-D(q)\bigr),
\end{equation}
both continuous, strictly increasing bijections of $(0,1)$ onto itself and mutually inverse. The map $G$ sends the $x$-menu share of an element to its $y$-menu share, and $H$ does the reverse: by \eqref{eq:gain}, $c_1(p^{\bm y}_{\bm e},s')=c_1(p^{\bm x}_{\bm e},s)+D(q)$ at a calibrated configuration. Three properties are immediate from \eqref{eq:gain} and the definition of $D$:
\begin{equation}
\label{eq:Gfacts}
    G(q)=q, \qquad G(p)<p \iff D(p)>D(q), \qquad G(p)>p \iff D(p)<D(q).
\end{equation}
In particular $G(p_1)>p_1$ (a \emph{gaining} level) and $G(p_2)<p_2$ (a \emph{losing} level), and $G(p_1)<G(q)=q$.

\emph{Step 2: the menu template, and the choice of $N$ and $R$.} Fix an integer $N\ge 1/p_1$. A configuration is now specified by two numbers: the $x$-menu share $R\in(q,p_2]$ of a single \emph{losing} element ${\bm z}$, and the common $y$-menu share $r$ of $N$ pairwise distinct \emph{gaining} elements. The two menus carry the shares
\begin{alignat*}{3}
    &A\cup\{{\bm x}\}: \quad && q \;(\,{\bm w}\,), \quad R \;(\,{\bm z}\,), \quad H(r)\ \text{each},
    \quad && p_{\bm x}=1-q-R-N H(r); \\
    &A\cup\{{\bm y}\}: \quad && q \;(\,{\bm w}\,), \quad G(R) \;(\,{\bm z}\,), \quad r\ \text{each},
    \quad && p_{\bm y}=1-q-G(R)-N r,
\end{alignat*}
and are admissible exactly when $p_{\bm x},p_{\bm y}>0$. Writing $L(R):=R-G(R)>0$ for ${\bm z}$'s loss and $\Gamma(r):=N(r-H(r))>0$ for the gaining elements' total gain,
\begin{equation}
\label{eq:budgetgap}
    p_{\bm x}-p_{\bm y} \;=\; \Gamma(r)-L(R).
\end{equation}
Let $r_{\max}(R):=(1-q-G(R))/N$, the value of $r$ at which $p_{\bm y}$ vanishes, and $\Gamma_N(R):=\Gamma(r_{\max}(R))$. Because $N\ge1/p_1$ we have $r_{\max}(R)\le1/N\le p_1$, so $D(r_{\max}(R))\le D(p_1)<D(q)$ and $r_{\max}(R)$ is a gaining level by \eqref{eq:Gfacts}; hence $\Gamma_N(R)>0$ for every $R$. Now consider
\begin{equation*}
    k(R) \;:=\; 2L(R)-\Gamma_N(R), \qquad R\in[q,p_2],
\end{equation*}
which is continuous, with $k(q)=-\Gamma_N(q)<0$ since $L(q)=0$ by \eqref{eq:Gfacts}. If $k(p_2)\le0$, put $R_1:=p_2$; then $L(R_1)>0$ because $D(p_2)>D(q)$, and $2L(R_1)\le\Gamma_N(R_1)$. Otherwise the intermediate value theorem gives $R_1\in(q,p_2)$ with $\Gamma_N(R_1)=2L(R_1)$, and $\Gamma_N(R_1)>0$ forces $L(R_1)>0$. Either way
\begin{equation}
\label{eq:LGam}
    0 \;<\; L(R_1) \;<\; \Gamma_N(R_1), \qquad\text{hence}\qquad D(R_1)>D(q).
\end{equation}
This is the step that has to be arranged with care: the loss of ${\bm z}$ must be dominated by the total gain of the adjoined elements, and \eqref{eq:LGam} secures it without any asymptotic estimate.

\emph{Step 3: two one-parameter families.} Fix $R:=R_1$ and let $r$ range over $(0,r_{\max}(R_1))$. Here $p_{\bm y}=N(r_{\max}(R_1)-r)>0$, and $p_{\bm x}$ is decreasing in $r$ with $p_{\bm x}\ge p_{\bm x}\big|_{r=r_{\max}(R_1)}=\Gamma_N(R_1)-L(R_1)>0$ by \eqref{eq:budgetgap} and \eqref{eq:LGam}; so the whole family is admissible. As $r\uparrow r_{\max}(R_1)$ we get $p_{\bm y}\downarrow0$ while $p_{\bm x}$ stays bounded away from $0$, so by \eqref{eq:target} and Definition \ref{norm-APU}(ii),
\begin{equation*}
    d \;=\; c_1(p_{\bm y},s')-c_1(p_{\bm x},s)-D(q) \;\longrightarrow\; -\infty .
\end{equation*}

For the other direction fix $r:=r_a$ and let $R$ range over $[R_1,R_{\max})$, where $R_{\max}:=1-q-N H(r_a)$ is the value of $R$ at which $p_{\bm x}$ vanishes. Since $H(r_a)\downarrow0$ as $r_a\downarrow0$, we have $R_{\max}\uparrow1-q$; and $G(1-q)<1-q$ because $D(1-q)\ge D(p_2)>D(q)$, so $(1-q)-G(R_{\max})-N r_a$ is bounded below by a positive constant for $r_a$ small enough. Fix such an $r_a$, small enough also that $r_a<r_{\max}(R_1)$ and $R_{\max}>R_1$. Then $p_{\bm y}=1-q-G(R)-N r_a>0$ throughout, while $p_{\bm x}=R_{\max}-R\downarrow0$ as $R\uparrow R_{\max}$, so $c_1(p_{\bm x},s)\to-\infty$ and $d\to+\infty$.

\emph{Step 4: every $d$ is realized, with the strict inequality.} The two families of Step 3 share the configuration $(R_1,r_a)$, so their union is a connected one-dimensional path on which $d$ is continuous, with values tending to $-\infty$ at one end and to $+\infty$ at the other. By the intermediate value theorem $d$ attains the prescribed value $b-a$ somewhere on it. Translating all payoffs as in Step 0 then matches $a$, and hence $b$, exactly, and \textsf{Normative Richness} supplies the required allocations: ${\bm w}$ with $u(w_1)=c_1(q,s)+\lambda_{\bm x}$, ${\bm z}$ with $u(z_1)=c_1(R,s)+\lambda_{\bm x}$, and $N$ pairwise distinct gaining elements with $u(e_1)=c_1(H(r),s)+\lambda_{\bm x}$, all with $\varphi\le s$ and none equal to ${\bm x}$ or ${\bm y}$.

Finally, the two conclusions hold at every point of the path. Along both families $R\ge R_1$, so $D(R)\ge D(R_1)>D(q)$ by \eqref{eq:LGam}, which is the required strict inequality at $R=\rho({\bm z},A\cup\{{\bm x}\})$. And $R\ge R_1>q$ gives $c_1(R,s)>c_1(q,s)$, i.e., $u(z_1)>u(w_1)$, hence $z_1>w_1$ because $u$ is monotonic.

\end{proof}

\begin{remark}[\textsf{Scale families: an explicit shortcut}]
\label{rmk:scale_shortcut}
When the cost family has the \emph{scale} (temperature) form $c_1(p,t)=\gamma(t)\psi(p)$ for a strictly increasing $\gamma$ and a fixed kernel $\psi$---Example \ref{ex:base_model} is the case $\psi=\kappa'$---a calibrated configuration can be exhibited in a menu of exactly three elements, with an explicit threshold. For $v\in\mathbb R$ pick, by the finite-set-avoiding clause of \textsf{Normative Richness}, \emph{distinct} ${\bm w},{\bm z}\notin\{{\bm x},{\bm y}\}$ with $\varphi\le s$ and $u(w_1)=u(z_1)=v$, and consider $\{{\bm w},{\bm z},{\bm x}\}$ (ceiling $s$) and $\{{\bm w},{\bm z},{\bm y}\}$ (ceiling $s'$). Equal utilities give ${\bm w}$ and ${\bm z}$ a common probability $q$, the third element receives $1-2q$, and
\begin{equation*}
    \Phi_t(q) \;:=\; c_1(1-2q,\,t) - c_1(q,\,t) \;=\; u(\text{third element}) - v .
\end{equation*}
On $(0,\tfrac12)$ the map $\Phi_t$ is continuous, strictly decreasing and onto $\mathbb R$, and $\Phi_t(\tfrac13)=0$ for \emph{every} $t$, since $1-2q=q$ at $q=\tfrac13$. The tie shares are therefore $q_x(v)=\Phi_{s}^{-1}(a-v)$ and $q_y(v)=\Phi_{s'}^{-1}(b-v)$, both continuous in $v$; write $\Delta^{\mathrm{tie}}(v):=q_x(v)-q_y(v)$. In the scale case $\Phi_t=\gamma(t)\Psi$ with $\Psi:=\psi(1-2\cdot)-\psi(\cdot)$, so, $\Psi^{-1}$ being decreasing,
\begin{equation*}
    \Delta^{\mathrm{tie}}(v)>0
    \iff \frac{a-v}{\gamma(s)} < \frac{b-v}{\gamma(s')}
    \iff v > v^\ast := \frac{\gamma(s')\,a-\gamma(s)\,b}{\gamma(s')-\gamma(s)},
\end{equation*}
a strict, explicit threshold. Choosing $v_-<v^\ast<v_+$ gives $\Delta^{\mathrm{tie}}(v_-)<0<\Delta^{\mathrm{tie}}(v_+)$. Setting, for $\beta\le T$ and allocations ${\bm w}_\beta,{\bm z}_T$ with $\varphi\le s$ and $u$-values $\beta,T$,
\begin{equation*}
    f(\beta,T) := \rho({\bm w}_\beta,\{{\bm w}_\beta,{\bm z}_T,{\bm x}\}) - \rho({\bm w}_\beta,\{{\bm w}_\beta,{\bm z}_T,{\bm y}\}),
\end{equation*}
which is continuous (the multipliers solve budget identities whose left-hand sides are $C^1$ and strictly decreasing in $\lambda$) and satisfies $f(v,v)=\Delta^{\mathrm{tie}}(v)$, and traversing $(v_-,v_-)\to(v_-,v_+)\to(v_+,v_+)$ within $\{\beta\le T\}$, the intermediate value theorem produces a zero of $f$ at an interior point, where $\beta<T$, i.e., $w_1<z_1$. Strictness of $D(R)>D(q)$ is automatic here, since $D=(\gamma(s')-\gamma(s))\psi$ is non-constant on every subinterval. This route is worth recording separately because it uses three-element menus with a \emph{prescribed} pair of tied payoffs, which is what the independence examples of Appendix \ref{app:independence} need in order to pin a menu-dependent cost index.
\end{remark}

\begin{remark}[\textsf{What the calibration argument rests on}]
\label{rmk:calibration_tail}
Two points about Lemma \ref{lem:calibrated_config} deserve emphasis, because a shorter and more natural-looking route to it fails.

First, one might hope to obtain the lemma from the tie configurations of Remark \ref{rmk:scale_shortcut} in general, by showing that the curves $v\mapsto\Phi_{s}^{-1}(a-v)$ and $v\mapsto\Phi_{s'}^{-1}(b-v)$ always cross. They need not. The two curves meet at a common tie share $q$ precisely when $\Phi_{s'}(q)-\Phi_{s}(q)=b-a$, i.e., when
\begin{equation*}
    D(1-2q)-D(q) \;=\; b-a \qquad\text{for some } q\in(0,\tfrac12).
\end{equation*}
The left-hand side is continuous and non-increasing in $q$, running from $V:=D(1^-)-D(0^+)$---the total variation of $D$---down to $D(0^+)-D(\tfrac12)\le0$, so a calibrated tie configuration exists if and only if $b-a\in[D(0^+)-D(\tfrac12),\,V)$. For bounded $D$ the quantity $V$ is finite: for $c_1(p,t)=\kappa'(p)+\theta(t)p$, say, $V$ is just the spread of $\theta$. So whenever the two selfish payoffs are far enough apart, \emph{no} tie configuration is calibrated and the three-element route is unavailable. This is why the proof above enlarges the menu: the adjoined elements, all less selfish than ${\bm w}$, are what make the residual probability $p_{\bm y}$ collapse while $p_{\bm x}$ stays bounded away from zero. In the scale families---hence in the canonical model of Example \ref{ex:base_model} and in every example of Appendix \ref{app:independence}---one has $V=+\infty$ and the shortcut applies.

Second, the hypothesis that the levels be \emph{separated} cannot be weakened to mere distinguishability, and this is what Definition \ref{norm-APU}(iv) secures. Step 1 needs two levels $p_1<p_2$ \emph{below} $\tfrac12$ with $D(p_1)<D(p_2)$: the first supplies the gaining elements and the second the losing one, through \eqref{eq:Gfacts}. So $D$ must be non-constant on the \emph{lower} half of the probability range, because the calibration level $q$ of an admissible configuration is necessarily below $\tfrac12$: ${\bm z}$ and ${\bm w}$ both lie in $A\cup\{{\bm x}\}$ and ${\bm z}$ has the strictly larger probability. Example \ref{ex:eta_failure} is exactly a case in which $D$ is flat below $\tfrac12$ and rises only above it; there the two levels differ behaviorally, no calibrated configuration compares the relevant pair, $\succsim_n^\rho$ acquires a vacuous ranking, and transitivity fails.
\end{remark}

If instead the levels $\varphi({\bm x})<\varphi({\bm y})$ are indistinguishable, then they generate the same randomization behavior in every menu and $\succsim_n^\rho$ ranks ${\bm x}$ and ${\bm y}$ as indifferent. This is not a defect of the argument but a genuine limit of identification: no condition stated in terms of $\rho$ can order two levels that induce the same behavior. By Definition \ref{norm-APU}(iv) indistinguishable and separated are the only two cases, so merging indistinguishable levels (Definition \ref{def:level_separation}) yields a normalized representation, and under it ${\bm x} \succsim_n^\rho {\bm y} \iff \varphi({\bm x}) \geq \varphi({\bm y})$ holds without qualification. This proves \textsf{Normative Weak Order} (Axiom \ref{ax:norm_weak_order}) for normalized representations, and hence for all representations once $\varphi$ is taken in normalized form.

\textsf{Normative Penalty}, the compression property, is likewise implied rather than imposed (Proposition \ref{prop:penalty_implied}).
\begin{definition}[\textsf{Normative Penalty}]
\label{ax:norm_penalty}
For any menus $A, B \in \mathcal{A}$, suppose $A$ normatively dominates $B$ in the sense that there exists ${\bm a} \in A$ such that ${\bm a} \succsim_n^\rho {\bm z}$ for all ${\bm z} \in B$. Then, for any ${\bm x}, {\bm y} \in A \cap B$ such that $\rho({\bm y},A) = \rho({\bm y},B)$ and $x_1 > y_1$,
\begin{align*}
    \rho({\bm x}, A) \leq \rho({\bm x}, B).
\end{align*}
\end{definition}

Definition \ref{ax:norm_penalty} captures the core idea that stronger normative environments impose greater pressure on choice behavior, but---unlike a naive ``adding a normatively strong option always suppresses everyone else'' statement---it does so \emph{relative to a calibration point}. The hypothesis $\rho({\bm y},A)=\rho({\bm y},B)$ singles out an allocation ${\bm y}$ that is \emph{equally likely to be chosen} in both menus, so that ${\bm y}$'s realized probability serves as a shared reference; the condition then says that, relative to this shared reference, an alternative ${\bm x}$ with a \emph{higher} selfish payoff than ${\bm y}$ is chosen with (weakly) lower probability in the more normatively demanding menu $A$ than in $B$. In words: stronger normative pressure compresses choice probabilities toward the calibration point, and this compression suppresses alternatives that are more selfish than the reference. (We do not separately assert the mirror-image statement for $x_1 < y_1$---that stronger pressure \emph{raises} the probability of alternatives less selfish than the calibration point---because Definition \ref{norm-APU}(iii), once established below, entails it as a consequence rather than requiring it as a primitive hypothesis; see Remark \ref{rmk:compression_direction}.) This calibrated formulation is what lets the condition be simultaneously true of the model (Section \ref{app:closing_marginal_cost}) and consistent with the central mechanism of this paper---namely, that raising the normative ceiling can \emph{increase} the probability of a sufficiently selfish alternative even as it suppresses intermediate ones (Section \ref{intro}), rather than uniformly suppressing every alternative regardless of its own payoff. Importantly, this condition does not impose any structure on the cost function directly, and it is not what delivers the monotonicity of the \emph{curvature} of the marginal cost with respect to the normative ceiling: that is \textsf{Normative Selectivity}, through Proposition \ref{prop:marginal_cost}. The compression property is a consequence of the axioms rather than an input to them (Proposition \ref{prop:penalty_implied}).

\begin{remark}[\textsf{The mirror direction: compression, not one-sided suppression}]
\label{rmk:compression_direction}
Once $\rho$ has a norm-dependent APU representation satisfying the revised Definition \ref{norm-APU}(iii) (curvature of $c$ non-decreasing in $t$), the \emph{same} hypotheses as in Definition \ref{ax:norm_penalty} but with $x_1 < y_1$ in place of $x_1>y_1$ give the reverse conclusion, $\rho({\bm x},A) \geq \rho({\bm x},B)$: an alternative \emph{less} selfish than the calibration point ${\bm y}$ is chosen with (weakly) \emph{higher} probability under the stronger normative pressure. The argument is the mirror image of the ``if'' direction of Proposition \ref{prop:curvature_char} (Appendix \ref{app:closing_marginal_cost}): writing $q:=\rho({\bm y},A)=\rho({\bm y},B)$, the interior FOCs give $\int_{q}^{\rho({\bm x},A)} c_{11}(p,t_A)\,dp = x_1-y_1 = \int_{q}^{\rho({\bm x},B)} c_{11}(p,t_B)\,dp$, and since $x_1<y_1$ forces $\rho({\bm x},A),\rho({\bm x},B) < q$, both integrals run \emph{backward} from $q$; $c_{11}(\cdot,t_A)\ge c_{11}(\cdot,t_B)$ pointwise, together with strict monotonicity of $c_1(\cdot,t_B)$, then forces $\rho({\bm x},A)\ge\rho({\bm x},B)$ by the same monotonicity argument used for Proposition \ref{prop:curvature_char}, with the inequality flipped because the interval of integration lies below $q$ rather than above it. So the full behavioral content of higher normative pressure is a two-sided \emph{compression} of choice probabilities toward the calibration point ${\bm y}$: alternatives more selfish than ${\bm y}$ are suppressed, and alternatives less selfish than ${\bm y}$ are correspondingly boosted. This is exactly the pattern underlying the Introduction's example (Section \ref{intro}): raising the normative ceiling suppresses the intermediate option while \emph{increasing} the probability of the purely selfish one, which sits on the far side of whatever calibration point separates them. Definition \ref{ax:norm_penalty} records only the suppression half, both because it is the direction with the clearer normative interpretation (higher pressure disciplining more-selfish deviations) and because, once Definition \ref{norm-APU}(iii) is established (Corollary \ref{cor:prop_marginal_cost_curvature}), the boosting half follows automatically and need not be separately assumed.
\end{remark}

\begin{proposition}[\textsf{Normative Penalty is implied}]
\label{prop:penalty_implied}
If $\rho$ satisfies \textit{Positivity} and the axioms of Theorem \ref{representation_norm}(a) ---equivalently, if $\rho$ is a norm-dependent additive perturbed utility---then $\rho$ satisfies \textsf{Normative Penalty}. None of the axioms of part (a) is \textsf{Normative Penalty} itself, and the sufficiency argument of Appendix \ref{app:sufficiency} never invokes it.
\end{proposition}

The proof is the necessity argument of Appendix \ref{sec:necessity_norm}, which derives \textsf{Normative Penalty} from a representation; what makes the statement a result rather than a tautology is that Theorem \ref{representation_norm} produces the representation without it. So the compression property is not something we ask the reader to accept. It is a prediction, on the same footing as the Regularity pattern of Proposition \ref{prop:regularity}: a stronger normative environment compresses choice probabilities toward whatever alternative is equally likely in both menus, and the compression suppresses the alternatives more selfish than that reference.

\noindent \textbf{Part (ii) [Penalty]:} Let $A, B \in \mathcal{A}$ with ${\bm x},{\bm y} \in A \cap B$, $\rho({\bm y},A)=\rho({\bm y},B)=:q$ and $x_1>y_1$. Suppose $A$ normatively dominates $B$ in the sense that there exists ${\bm a} \in A$ such that ${\bm a} \succsim_n^\rho {\bm z}$ for all ${\bm z} \in B$. Under the representation, this implies $t_A:=\max_{{\bm w} \in A} \varphi({\bm w}) \geq \varphi({\bm a}) \geq \max_{{\bm z} \in B} \varphi({\bm z})=:t_B$. Since $(u,\varphi,c)$ is, by hypothesis, a norm-dependent APU representation, $c$ satisfies the revised Definition \ref{norm-APU}(iii) directly, i.e., $c_{11}(p,t_A)\ge c_{11}(p,t_B)$ for every $p\in(0,1)$ (this is part of what it means for $(u,\varphi,c)$ to represent $\rho$ at all---no separate hypothesis is needed here). Since $\rho({\bm y},A)=\rho({\bm y},B)=q$, the interior FOCs at ${\bm y}$ give $c_1(q,t_A)-c_1(q,t_B) = \lambda_B-\lambda_A$; combined with the interior FOCs at ${\bm x}$, this yields (exactly as in the ``only if'' direction of Proposition \ref{prop:curvature_char})
\begin{equation*}
\int_{q}^{\rho({\bm x},A)} c_{11}(p,t_A)\,dp \;=\; x_1-y_1 \;=\; \int_{q}^{\rho({\bm x},B)} c_{11}(p,t_B)\,dp.
\end{equation*}
Since $x_1>y_1$, the interior FOC at ${\bm y}$ together with strict monotonicity of $c_1(\cdot,t_A)$ (resp.\ $c_1(\cdot,t_B)$) forces $\rho({\bm x},A)>q$ and $\rho({\bm x},B)>q$. Both integrals therefore run forward from $q$ over intervals where $c_{11}(\cdot,t_A)\ge c_{11}(\cdot,t_B)$ pointwise; if $\rho({\bm x},A)>\rho({\bm x},B)$, then $\int_q^{\rho({\bm x},A)}c_{11}(\cdot,t_A) \ge \int_q^{\rho({\bm x},A)}c_{11}(\cdot,t_B) > \int_q^{\rho({\bm x},B)}c_{11}(\cdot,t_B)$, contradicting the displayed equality; the strict inequality in the last step is
\begin{equation*}
    \int_{q}^{\rho({\bm x},A)}c_{11}(\cdot,t_B) - \int_{q}^{\rho({\bm x},B)}c_{11}(\cdot,t_B) = c_1(\rho({\bm x},A),t_B)-c_1(\rho({\bm x},B),t_B) > 0,
\end{equation*}
which uses only strict monotonicity of $c_1(\cdot,t_B)$ (strict convexity of $c(\cdot,t_B)$) and not the pointwise inequality $c_{11}(\cdot,t_B)>0$, which strict convexity does not imply. Hence $\rho({\bm x},A)\le\rho({\bm x},B)$, proving the necessity of \textsf{Normative Penalty} (Definition \ref{ax:norm_penalty}). (This is, in substance, the same integration argument as the ``if'' direction of Proposition \ref{prop:curvature_char}; see also Remark \ref{rmk:compression_direction} for the mirror statement when $x_1<y_1$.)



\subsubsection{\textsf{Selfishness-Acyclicity} (Axiom \ref{ax:selfish_acyclicity})}
Suppose, for contradiction, that $\rho$ violates Axiom \ref{ax:selfish_acyclicity}: there exist $n \geq 2$, allocations ${\bm x}^1,\ldots,{\bm x}^n \in X$, and menus $A^1,\ldots,A^n \in \mathcal{A}$ with ${\bm x}^k,{\bm x}^{k+1} \in A^k$ (indices mod $n$) such that, for every $k$, at least one of $\rho({\bm x}^k,A^k) \geq \rho({\bm x}^{k+1},A^k)$ or $x_1^k \geq x_1^{k+1}$ holds, with at least one strict.

Fix $k$ and write $t^k := \varphi(A^k)$. If $x_1^k \geq x_1^{k+1}$ (with equality unless this is the strict link), then since $u$ is strictly increasing, $u(x_1^k) \geq u(x_1^{k+1})$ (strictly, if $x_1^k>x_1^{k+1}$). If instead $\rho({\bm x}^k,A^k) \geq \rho({\bm x}^{k+1},A^k)$, the interior first-order conditions of \eqref{model_norm} give $u(x_1^k) - c_1(\rho({\bm x}^k,A^k),t^k) = u(x_1^{k+1}) - c_1(\rho({\bm x}^{k+1},A^k),t^k)$, and since $c_1(\cdot,t^k)$ is strictly increasing (strict convexity of $c(\cdot,t^k)$), $\rho({\bm x}^k,A^k) \geq \rho({\bm x}^{k+1},A^k)$ implies $c_1(\rho({\bm x}^k,A^k),t^k) \geq c_1(\rho({\bm x}^{k+1},A^k),t^k)$, so $u(x_1^k) \geq u(x_1^{k+1})$ (again strictly whenever the probability inequality is strict). Either way, $u(x_1^k) \geq u(x_1^{k+1})$ for every $k=1,\ldots,n$, strictly for at least one $k$. Summing around the cycle,
\begin{equation*}
u(x_1^1) \geq u(x_1^2) \geq \cdots \geq u(x_1^n) \geq u(x_1^1),
\end{equation*}
with at least one strict inequality, a contradiction. Hence no such cycle exists, and $\rho$ satisfies \textsf{Selfishness-Acyclicity} (Axiom \ref{ax:selfish_acyclicity}).

\subsubsection{\textsf{Normative Independence} (Axiom \ref{ax:norm_independence})}
\label{sec:necessity_independence}
Suppose $\rho$ has a norm-dependent APU representation $(u,\varphi,c)$. For every allocation ${\bm a}\in X$, set $g_{\bm a} := c_1(\cdot,\varphi({\bm a}))$. By Definition \ref{norm-APU}(i), $c_1(\cdot,t)$ is $C^1$ on $(0,1)$ for every $t$, and by strict convexity of $c(\cdot,t)$ it is strictly increasing, so each $g_{\bm a}$ satisfies the required regularity; by Definition \ref{norm-APU}(ii), $c_1(p,t)\to-\infty$ as $p\to0^+$ and $c_1(p,t)\to+\infty$ as $p\to1^-$, so each $g_{\bm a}$ also has the boundary behavior required by Axiom \ref{ax:norm_independence}. Since $\varphi$ represents $\succsim_n^\rho$ (Lemma \ref{norm_result}), ${\bm a}\sim_n^\rho{\bm b}$ implies $\varphi({\bm a})=\varphi({\bm b})$, hence $g_{\bm a}=g_{\bm b}$: the family $\{g_{\bm a}\}_{{\bm a}\in X}$ satisfies the $\sim_n^\rho$-invariance required by Axiom \ref{ax:norm_independence}. Let $A,B\in\mathcal A$ be \emph{any} two menus, let ${\bm a}\in A$ set the ceiling of $A$ and ${\bm b}\in B$ set the ceiling of $B$ (so $\varphi({\bm a})=\varphi(A)=:t_A$ and $\varphi({\bm b})=\varphi(B)=:t_B$, since $\varphi$ represents $\succsim_n^\rho$), and let ${\bm x},{\bm y}\in A\cap B$. The interior first-order conditions of \eqref{model_norm} at $A$ give $u(x_1) - c_1(\rho({\bm x},A),t_A) = u(y_1) - c_1(\rho({\bm y},A),t_A)$, i.e.,
\begin{equation*}
    g_{\bm a}(\rho({\bm x},A)) - g_{\bm a}(\rho({\bm y},A)) = u(x_1) - u(y_1),
\end{equation*}
and likewise, from the first-order conditions at $B$ (using $c(\cdot,t_B)$, whatever $t_B$ happens to be, since $B$'s own ceiling need not equal $A$'s),
\begin{equation*}
    g_{\bm b}(\rho({\bm x},B)) - g_{\bm b}(\rho({\bm y},B)) = u(x_1) - u(y_1).
\end{equation*}
The right-hand sides coincide (both depend on ${\bm x},{\bm y}$ only through the fixed, globally shared function $u$---Definition \ref{norm-APU} does not let $u$ depend on $t$), so $g_{\bm a}(\rho({\bm x},A)) - g_{\bm a}(\rho({\bm y},A)) = g_{\bm b}(\rho({\bm x},B)) - g_{\bm b}(\rho({\bm y},B))$, regardless of whether $A,B$ share a ceiling. Since $A,B,{\bm x},{\bm y}$ were entirely arbitrary, this proves the necessity of \textsf{Normative Independence} (Axiom \ref{ax:norm_independence}), with the same family $\{g_{\bm a}\}_{{\bm a}\in X}$ working simultaneously for every pair of menus, same-ceiling or not.

\subsubsection{\textsf{Normative Selectivity} (Axiom \ref{ax:selectivity})}
\label{sec:necessity_selectivity}
Suppose $\rho$ has a norm-dependent APU representation $(u,\varphi,c)$ satisfying the (revised) cost-function requirement of Definition \ref{norm-APU}(i) and (iii), i.e., $c$ is $C^2$ in its first argument and $c_{11}(p,t_A) \geq c_{11}(p,t_B)$ for every $p \in (0,1)$ and every $t_A > t_B$ in $\varphi(X)$. This is exactly the hypothesis of the ``if'' direction of Proposition \ref{prop:curvature_char} (curvature ordering $\Rightarrow$ Axiom \ref{ax:selectivity}), proved in Appendix \ref{app:closing_marginal_cost}; we refer the reader there rather than repeat the argument. That direction establishes that $\rho$ satisfies \textsf{Normative Selectivity} (Axiom \ref{ax:selectivity}).

Consequently the equivalence in Theorem \ref{representation_norm} closes in both directions at Axiom \ref{ax:selectivity}: (b) $\Rightarrow$ Axiom \ref{ax:selectivity} holds by the necessity argument just given, and Axiom \ref{ax:selectivity} (together with the other five conditions of part (a)) $\Rightarrow$ (b) holds by Corollary \ref{cor:prop_marginal_cost_curvature}, which closes the previously open step in the Sufficiency Part (Step 4, Appendix \ref{app:sufficiency}). Axiom \ref{ax:selectivity} is necessary \emph{by construction}: it was reverse-engineered from exactly the first-order-condition identity that any representation must satisfy, so that both directions of Proposition \ref{prop:curvature_char} hold simultaneously---mirroring \citet{FIS_2015}'s own Proposition 3, which is likewise an if-and-only-if characterization rather than a one-directional sufficient condition.

\subsection{Theorem \ref{unique_norm}}
\label{app:unique}

Suppose $(u, \varphi, c)$ and $(\hat{u}, \hat{\varphi}, \hat{c})$ both represent the same stochastic choice function $\rho$.

\textit{Part (ii).} By Definition \ref{norm}, the behavioral relation $\succsim_n^\rho$ is defined directly from $\rho$, independently of any particular representation; hence both $\varphi$ and $\hat\varphi$ represent this same $\succsim_n^\rho$. Define $f: \varphi(X) \to \mathbb{R}$ by $f(\varphi({\bm x})) := \hat\varphi({\bm x})$. This is well-defined and strictly increasing: for ${\bm x}, {\bm y} \in X$, $\varphi({\bm x}) = \varphi({\bm y}) \iff {\bm x} \sim_n^\rho {\bm y} \iff \hat\varphi({\bm x}) = \hat\varphi({\bm y})$, and $\varphi({\bm x}) > \varphi({\bm y}) \iff {\bm x} \succ_n^\rho {\bm y} \iff \hat\varphi({\bm x}) > \hat\varphi({\bm y})$. Continuity of $f$ on $\varphi(X)$ follows because $\varphi$ and $\hat\varphi$ are both continuous and $X$ is connected: $f = \hat\varphi \circ \varphi^{-1}$ is a strictly increasing bijection between the two intervals $\varphi(X)$ and $\hat\varphi(X)$ (both connected subsets of $\mathbb{R}$, being continuous images of the connected space $X$), and a strictly increasing bijection between two real intervals is automatically continuous. This proves $\hat\varphi({\bm x}) = f(\varphi({\bm x}))$ for all ${\bm x} \in X$, which is condition (ii).

In particular, since $f$ is strictly increasing, for every menu $A$, $\hat\varphi(A) = \max_{{\bm x} \in A} \hat\varphi({\bm x}) = \max_{{\bm x} \in A} f(\varphi({\bm x})) = f\bigl(\max_{{\bm x}\in A}\varphi({\bm x})\bigr) = f(\varphi(A))$; that is, the normative ceiling transforms consistently: $\hat t = f(t)$ whenever $t = \varphi(A)$ for some menu $A$.

\textit{Parts (i) and (iii).} Fix an allocation ${\bm a}\in X$, put $t:=\varphi({\bm a})$ and $\hat t := \hat\varphi({\bm a}) = f(t)$, and abbreviate
\begin{equation*}
    h := c_1(\cdot,t), \qquad \hat h := \hat c_1(\cdot,\hat t), \qquad G := \hat h\circ h^{-1}.
\end{equation*}
By Definition \ref{norm-APU}(i)--(ii), $h$ and $\hat h$ are continuous, strictly increasing bijections from $(0,1)$ onto $\mathbb R$, so $G:\mathbb R\to\mathbb R$ is a continuous, strictly increasing bijection. We derive a functional equation for $G$ directly from the first-order conditions; no appeal to external uniqueness results for item-invariant models is needed, and no density property of the level set $\{\varphi = t\}$ is used.

Let ${\bm x},{\bm y}$ be allocations with ${\bm a}\succsim_n^\rho{\bm x}$ and ${\bm a}\succsim_n^\rho{\bm y}$, distinct from each other and from ${\bm a}$, as provided by \textsf{Normative Richness} (Axiom \ref{richness_below_ceiling}); by part (ii) the relation $\succsim_n^\rho$ is the same under both representations, so ${\bm a}$ sets the ceiling of $A:=\{{\bm a},{\bm x},{\bm y}\}$ under each. Writing $\sigma := h(\rho({\bm x},A))$ and $d := u(a_1)-u(x_1)$, the interior first-order conditions of \eqref{model_norm} under $(u,\varphi,c)$ give $h(\rho({\bm a},A)) = \sigma + d$, and those under $(\hat u,\hat\varphi,\hat c)$ give
\begin{equation*}
    \hat u(x_1) - \hat\lambda_A = \hat h(\rho({\bm x},A)) = G(\sigma), \qquad
    \hat u(a_1) - \hat\lambda_A = \hat h(\rho({\bm a},A)) = G(\sigma+d).
\end{equation*}
Subtracting,
\begin{equation}
\label{eq:uniq_functional}
    G(\sigma+d) - G(\sigma) \;=\; \hat u(a_1) - \hat u(x_1) \;=:\; k(d),
\end{equation}
where the right-hand side depends on $d$ alone: $u$ is strictly increasing and onto $\mathbb R$, so $x_1 = u^{-1}(u(a_1)-d)$ is determined by $d$, and hence so is $\hat u(x_1)$.

Equation \eqref{eq:uniq_functional} holds for every pair $(d,\sigma)$ for which the menu $A$ can be built, i.e., whenever $h^{-1}(\sigma)+h^{-1}(\sigma+d)<1$: given such a pair, \textsf{Normative Richness} supplies ${\bm x}$ with the required payoff below ${\bm a}$'s level, and ${\bm y}$ with the payoff $u^{-1}\bigl(\lambda_A + h(1-h^{-1}(\sigma)-h^{-1}(\sigma+d))\bigr)$, which absorbs the remaining probability; the finite-set avoidance clause makes the three allocations distinct. For any fixed $d$ the constraint holds once $\sigma$ is small enough, since $h^{-1}(\sigma),h^{-1}(\sigma+d)\to0$ as $\sigma\to-\infty$.

Consequently $k$ is additive: for $\sigma$ sufficiently negative,
\begin{equation*}
    k(d_1+d_2) = G(\sigma+d_1+d_2)-G(\sigma) = \bigl[G(\sigma+d_1+d_2)-G(\sigma+d_1)\bigr] + \bigl[G(\sigma+d_1)-G(\sigma)\bigr] = k(d_2)+k(d_1),
\end{equation*}
and $k$ is strictly increasing because $G$ is. An additive monotone function on $\mathbb R$ is linear, so $k(d)=\alpha d$ for some $\alpha>0$; no continuity or measurability input beyond monotonicity is required. Substituting into the definition of $k$ and letting $d$ range over $\mathbb R$ (possible since $u$ is onto) gives
\begin{equation*}
    \hat u(x_1) = \alpha\,u(x_1) + \beta_u \qquad\text{for all } x_1\in\mathbb R, \qquad \beta_u := \hat u(a_1)-\alpha\,u(a_1),
\end{equation*}
which is condition (i). Since $u$ and $\hat u$ are fixed functions on $\mathbb R$, evaluating at two points with distinct $u$-values determines $\alpha,\beta_u$ from $u,\hat u$ alone; in particular they do not depend on which ${\bm a}$, hence on which ceiling $t$, was used.

Returning to \eqref{eq:uniq_functional} with $k(d)=\alpha d$: fixing $\sigma_0$ small enough and letting $d$ range over $\mathbb R$ yields $G(\tau) = \alpha\tau + \beta'$ for all $\tau\in\mathbb R$, where $\beta' := G(\sigma_0)-\alpha\sigma_0$. Unwinding $G=\hat h\circ h^{-1}$, this is exactly
\begin{equation*}
\hat c_1(p, \hat t) = \alpha\, c_1(p, t) + \beta' \quad \text{for all } p \in (0,1);
\end{equation*}
integrating in $p$,
\begin{equation*}
\hat c(p, \hat t) = \alpha\, c(p, t) + \beta' p + \delta
\end{equation*}
for a constant of integration $\delta$. Both $\beta'$ and $\delta$ were obtained from the fixed ceiling pair $(t,\hat t)$ and do not depend on the menu used, since by Claim \ref{claim:well_defined_c} $c(\cdot,t)$ and $\hat c(\cdot,\hat t)$ depend on the menu only through its ceiling; write $\gamma(t):=\beta'$ and $\delta(t):=\delta$. Since ${\bm a}$, and hence $t$, was an arbitrary element of $\varphi(X)$ and $\hat t = f(t)$, i.e., $t = f^{-1}(\hat t)$, relabeling $\hat t$ as a generic variable $t \in \hat\varphi(X)$ and setting $\tilde\gamma(t) := \gamma(f^{-1}(t))$, $\tilde\delta(t) := \delta(f^{-1}(t))$ gives
\begin{equation*}
\hat c(p, t) = \alpha\, c(p, f^{-1}(t)) + \tilde\gamma(t) p + \tilde\delta(t) \quad \text{for all } p \in (0,1),\ t \in \hat\varphi(X),
\end{equation*}
which is condition (iii). This completes the proof. \qed


\subsection{Proposition \ref{prop:regularity}}
\label{app:regularity}

Throughout, $F_s:=c_1(\cdot,s)^{-1}:\mathbb R\to(0,1)$, which is a continuous, strictly increasing bijection by Definition \ref{norm-APU}(i)--(ii). Since $|A|\ge2$ and $|B|=|A|+1\ge3$, both menus contain at least two elements, so \textsf{Positivity} together with the budget constraints places every probability in the open interval $(0,1)$ and the interior first-order conditions hold in both menus:
\begin{equation}
\label{eq:reg_foc}
    c_1\bigl(\rho({\bm y},A),t\bigr) = u(y_1)-\lambda_A \ \ ({\bm y}\in A),
    \qquad
    c_1\bigl(\rho({\bm y},B),t'\bigr) = u(y_1)-\lambda_B \ \ ({\bm y}\in B),
\end{equation}
with $\lambda_A,\lambda_B$ the unique reals making the respective probabilities sum to one. Note $t'=\max\{t,\varphi({\bm z})\}\ge t$, so $D$ is non-decreasing by Definition \ref{norm-APU}(iii), and $D\equiv0$ when $t'=t$.

\emph{(i)} Fix ${\bm x}\in A$. Since $c_1(\cdot,t')$ is strictly increasing, $\rho({\bm x},B)>\rho({\bm x},A)$ holds if and only if $c_1(\rho({\bm x},B),t')>c_1(\rho({\bm x},A),t')$. By \eqref{eq:reg_foc} the left side equals $u(x_1)-\lambda_B$, and by the definition of $D$ the right side equals $c_1(\rho({\bm x},A),t)+D(\rho({\bm x},A)) = u(x_1)-\lambda_A+D(\rho({\bm x},A))$. Cancelling $u(x_1)$ and rearranging gives $\lambda_A-\lambda_B>D(\rho({\bm x},A))$, which is the assertion. The same computation with $>$ replaced by $=$ or $<$ gives the corresponding equivalences.

\emph{(ii)} We first record the counting fact used twice below: not every element of $A$ can gain probability, since that would give $\sum_{{\bm y}\in A}\rho({\bm y},B)>\sum_{{\bm y}\in A}\rho({\bm y},A)=1$, whereas $\sum_{{\bm y}\in A}\rho({\bm y},B)=1-\rho({\bm z},B)<1$ by \textsf{Positivity}. Hence a non-violator exists.

Next, \eqref{eq:reg_foc} and the strict monotonicity of $c_1(\cdot,t)$ show that $\rho(\cdot,A)$ is a strictly increasing function of $u(y_1)$, hence of $y_1$ since $u$ is selfish. So if ${\bm y}\in A$ has $y_1\le x_1$ then $\rho({\bm y},A)\le\rho({\bm x},A)$, and since $D$ is non-decreasing, $D(\rho({\bm y},A))\le D(\rho({\bm x},A))$; if ${\bm x}$ violates Regularity then $D(\rho({\bm x},A))<\Lambda$ by (i), so $D(\rho({\bm y},A))<\Lambda$ and ${\bm y}$ violates it too. The violators therefore form a lower set in the selfish payoff, and because a non-violator exists, the element of $A$ with the largest payoff is one.

For the quantitative form, let $\pi^\ast:=\sup\{p\in(0,1):D(p)<\Lambda\}$. If $\rho({\bm x},A)<\pi^\ast$, then some $p>\rho({\bm x},A)$ satisfies $D(p)<\Lambda$, and monotonicity of $D$ gives $D(\rho({\bm x},A))<\Lambda$: a violation by (i). If $\rho({\bm x},A)>\pi^\ast$, then $D(\rho({\bm x},A))\ge\Lambda$ by the definition of the supremum, and (i) gives $\rho({\bm x},B)\le\rho({\bm x},A)$.

\emph{(iii)} Suppose $\varphi({\bm z})\le\varphi(A)$, so $t'=t$ and $D\equiv0$. We claim $\lambda_B>\lambda_A$. If instead $\lambda_B\le\lambda_A$, then $F_t(u(y_1)-\lambda_B)\ge F_t(u(y_1)-\lambda_A)$ for every ${\bm y}\in A$, whence
\begin{equation*}
    1 \;=\; \sum_{{\bm y}\in A}F_t(u(y_1)-\lambda_B)+F_t(u(z_1)-\lambda_B)
      \;>\; \sum_{{\bm y}\in A}F_t(u(y_1)-\lambda_A) \;=\; 1,
\end{equation*}
the strict inequality because $F_t>0$. This is a contradiction, so $\Lambda=\lambda_A-\lambda_B<0=D$ and part (i) gives $\rho({\bm x},B)<\rho({\bm x},A)$ for every ${\bm x}\in A$.

\emph{(iv)} If the levels are indistinguishable then $D\equiv\bar D$ on $(0,1)$ for a constant $\bar D$, so the criterion in (i) reads $\bar D<\Lambda$ and does not depend on ${\bm x}$: either every element of $A$ gains probability, or none does. The first alternative is ruled out by the counting fact established in (ii). Hence $\bar D\ge\Lambda$ and $\rho({\bm x},B)\le\rho({\bm x},A)$ for every ${\bm x}\in A$. \qed

\subsection{Proposition \ref{prop:exante}}
\label{app:exante}

Write $\Gamma_t(p):=g_t(p)+g_t(1-p)$ and $h(r):=w({\bm x}_0-r{\bm d})$ for $r\in[0,\bar s]$, so that $h$ is continuous. Fix $s\in(0,\bar s]$ and consider the binary menu $A_s:=\{{\bm x}_0,{\bm x}_0-s{\bm d}\}$, whose ceiling is $t$ by \eqref{eq:exante_ray}. With $q:=\rho({\bm x}_0-s{\bm d})$ the induced mean is ${\bm x}_0-qs{\bm d}$, and the right-hand side of \eqref{eq:absorb} is $\Gamma_t(q)$ plus an affine function of $q$, so that
\begin{equation}
\label{eq:exante_chain}
    h(qs) \;=\; a_s + b_s\,q + \Gamma_t(q),
    \qquad q\in(0,1),
\end{equation}
for constants $a_s,b_s$. The decisive point is that $\Gamma_t$ depends on the normative ceiling alone and not on the menu, so \eqref{eq:exante_chain} holds with the \emph{same} $\Gamma_t$ for every admissible $s$. Differencing it at two scales $s,\sigma\in(0,\bar s]$ therefore eliminates $\Gamma_t$ and leaves $h(qs)-h(q\sigma)$ affine in $q$; letting $q\to0^{+}$ and using the continuity of $h$ at $0$ kills the constant term, so with $u:=q\sigma$ and $\lambda:=s/\sigma$,
\begin{equation}
\label{eq:exante_scaling}
    h(\lambda v) \;=\; h(v) + c_{\lambda}\,v,
    \qquad v\in(0,\bar s),
\end{equation}
for a constant $c_{\lambda}$ depending on $\lambda$ alone---two pairs with the same ratio give the same increment, as comparing them at a common argument shows. Expanding $h(\lambda\mu v)$ in the two possible orders gives $c_{\lambda\mu}=c_{\mu}+\mu c_{\lambda}=c_{\lambda}+\lambda c_{\mu}$, hence $c_{\mu}(1-\lambda)=c_{\lambda}(1-\mu)$; taking $\mu=\tfrac12$ yields $c_{\lambda}=-k(1-\lambda)$ with $k:=-2c_{1/2}$. Substituting into \eqref{eq:exante_scaling} gives $h(\lambda v)-k\lambda v=h(v)-kv$, so $h(v)-kv$ is constant on $(0,\bar s)$ and, by continuity, on $[0,\bar s]$. Thus $h$ is affine, which is the first claim, and \eqref{eq:exante_chain} then makes $\Gamma_t$ affine on $(0,1)$, which is the second. Note that no differentiability of $w$ has been used.

Suppose now that $w$ is twice continuously differentiable, and write $Q(r):={\bm d}^{\top}H_w({\bm x}_0-r{\bm d}){\bm d}$. The first claim says that $h$ is affine on $[0,\bar s]$, which is to say $Q\equiv0$ there. For (a), evaluate this at $r=0$ along each of the three rays: this gives ${\bm d}_i^{\top}H_w({\bm x}_0){\bm d}_i=0$ for $i=1,2,3$, and a symmetric bilinear form on $\mathbb R^{2}$ that annihilates three pairwise non-parallel vectors vanishes. For (b), fix a unit vector ${\bm d}$; the first claim says that $s\mapsto w({\bm x}_0-s{\bm d})$ is affine on $[0,\bar s]$, with slope $-\,Dw({\bm x}_0)\cdot{\bm d}$ read off at $s=0$. Hence $w({\bm z})=w({\bm x}_0)+Dw({\bm x}_0)\cdot({\bm z}-{\bm x}_0)$ for every ${\bm z}\in B({\bm x}_0,\bar s)$, which is the assertion; no statement about $H_w$ away from ${\bm x}_0$ is needed.

\subsection{Condition \ref{norm-APU}(iv): the three cases, an example, and Remark \ref{rmk:iv_free}}
\label{app:iv_free}

Condition \ref{norm-APU}(iv) is the one requirement of Definition \ref{norm-APU} with no counterpart in the perturbed utility literature. This appendix sets out what it excludes and shows that excluding it costs nothing.

By (iii) the difference $D_{s,s'}$ is non-decreasing, so for a pair of levels $s<s'$ exactly three cases can arise.
\begin{enumerate}
    \item[(a)] $D_{s,s'}$ is constant on all of $(0,1)$. The two levels then induce the same randomization behavior in every menu and are behaviorally indistinguishable; the coarser index obtained by merging them represents the same $\rho$.
    \item[(b)] $D_{s,s'}$ is non-constant on the lower half $(0,\tfrac12)$. The comparisons of Definition \ref{norm} can then detect the difference between the levels, and the necessity of \textsf{Normative Weak Order} goes through (Lemma \ref{lem:calibrated_config}).
    \item[(c)] $D_{s,s'}$ is constant on $(0,\tfrac12)$ but not on $(0,1)$. The two levels now \emph{do} differ behaviorally, yet for certain pairs of allocations no calibrated configuration comparing them exists at all, so Definition \ref{norm} ranks those pairs as indifferent \emph{vacuously} and the resulting relation can fail transitivity.
\end{enumerate}
Condition (iv) excludes case (c). It is costless in applications: every \emph{scale} family $c_1(p,t)=\gamma(t)\psi(p)$---in particular the entropy cost of Example \ref{ex:base_model} and every example of Appendix \ref{app:independence}---satisfies it, since there $D_{s,s'}=(\gamma(s')-\gamma(s))\psi$ is either identically zero or non-constant on every subinterval.

The following example underlies both the necessity of condition \ref{norm-APU}(iv) and the argument below.

\begin{example}[\textsf{Why \ref{norm-APU}(iv) cannot be dropped}]
\label{ex:eta_failure}
Let $X=\mathbb R^2$, $u(x_1)=x_1$, $\varphi({\bm x})=x_2$, and set
\begin{equation*}
    \eta(p) := \bigl(\max\{p-\tfrac12,\,0\}\bigr)^2, \qquad
    H(p) := \tfrac13\bigl(\max\{p-\tfrac12,\,0\}\bigr)^3, \qquad
    c(p,t) := \kappa(p) + \max\{t,0\}\cdot H(p),
\end{equation*}
with $\kappa$ the entropy kernel of Example \ref{ex:base_model}, so that $H'=\eta$ and
\begin{equation*}
    c_1(p,t) \;=\; \ln\tfrac{p}{1-p} \;+\; \max\{t,0\}\cdot\eta(p),
    \qquad
    c_{11}(p,t) \;=\; \tfrac{1}{p(1-p)} + \max\{t,0\}\cdot\eta'(p),
    \quad \eta'(p)=2\max\{p-\tfrac12,0\}.
\end{equation*}
Then $c$ is $C^2$ and strictly convex on $(0,1)$ (indeed $c_{11}>0$), giving (i); $c_1(\cdot,t)$ is strictly increasing and, since $0\le\eta\le\tfrac14$, maps $(0,1)$ onto $\mathbb R$, giving (ii); and $c_{11}(\cdot,t)$ is non-decreasing in $t$ because $\eta'\ge0$, giving (iii). Hence the objective in \eqref{model_norm} is strictly concave with a unique interior maximizer in every menu, so $\rho$ is well defined and satisfies \textsf{Positivity}. For the levels $s=0<s'=1$ we get $D_{0,1}=\eta$, which vanishes on $(0,\tfrac12]$ and is strictly increasing on $(\tfrac12,1)$: the levels are neither indistinguishable nor separated, so (iv) fails. Now take ${\bm x}=(1,0)$, ${\bm y}=(0,1)$ and ${\bm y}'=(2,1)$.

\emph{No calibrated configuration compares ${\bm x}$ with ${\bm y}$.} Suppose $(A,{\bm z},{\bm w})$ were one, with $\rho({\bm w},A\cup\{{\bm x}\})=\rho({\bm w},A\cup\{{\bm y}\})=q$ and $z_1>w_1$. Then $q<\tfrac12$, because ${\bm z}$ and ${\bm w}$ both lie in $A\cup\{{\bm x}\}$ and ${\bm z}$ has the strictly larger probability. Write $t:=\varphi(A)$, so that the two menus have ceilings $\max\{t,0\}$ and $\max\{t,1\}$, and let $D$ be the difference of the corresponding marginal costs. The menu $A$ is arbitrary, so all three cases must be covered, and in each of them $D$ is non-negative, non-decreasing, and vanishes on $(0,\tfrac12]$: if $t\le0$ then $D=\eta$; if $0<t<1$ then $D=(1-t)\eta$; and if $t\ge1$ the two ceilings coincide and $D\equiv0$. The first-order conditions at ${\bm w}$ give $\lambda_{\bm x}-\lambda_{\bm y}=D(q)=0$, since $q<\tfrac12$, so the two multipliers coincide; the first-order conditions at any ${\bm e}\in A$ then give $c_1(p^{\bm y}_{\bm e},\max\{t,0\})+D(p^{\bm y}_{\bm e})=c_1(p^{\bm x}_{\bm e},\max\{t,0\})$, whence $p^{\bm y}_{\bm e}\le p^{\bm x}_{\bm e}$ because $D\ge0$; and at the odd elements, $c_1(p_{\bm y},\max\{t,0\})+D(p_{\bm y})=u(y_1)-\lambda_{\bm y}<u(x_1)-\lambda_{\bm x} =c_1(p_{\bm x},\max\{t,0\})$, whence $p_{\bm y}<p_{\bm x}$. Summing, the probabilities in $A\cup\{{\bm y}\}$ add to strictly less than $1$---a contradiction. Hence ${\bm x}\succsim_n^\rho{\bm y}$ holds \emph{vacuously}, while ${\bm y}\succsim_n^\rho{\bm x}$ holds by the elementary direction of Part (i) in Appendix \ref{app:necessity}. Since $\varphi({\bm y})=\varphi({\bm y}')$, the same elementary direction gives ${\bm y}\sim_n^\rho{\bm y}'$.

\emph{But ${\bm x}\succsim_n^\rho{\bm y}'$ fails.} Here $u(y'_1)>u(x_1)$, and the flatness of $D$ below $\tfrac12$ no longer blocks the construction: what is needed is a single \emph{losing} share, and $D$ does rise above $\tfrac12$. Concretely, the second family of Step 3 in the proof of Lemma \ref{lem:calibrated_config}, run with $N=0$ and a two-element $A=\{{\bm w},{\bm z}\}$, applies to any pair $q<\tfrac12<R$ with $D(R)>D(q)=0$---for instance $q=\tfrac14$ and $R=\tfrac35$, where $D(R)=\tfrac1{100}$---and realizes every sufficiently large value of $u(y'_1)-u(x_1)$ with $D(R)>D(q)$, hence with strict compression at ${\bm z}$. The value $u(y'_1)-u(x_1)=1$ is already large enough: taking $A=\{{\bm w},{\bm z}\}$ with $w_1\approx7.239$, $z_1\approx8.044$ and $w_2=z_2=0$ gives $\rho({\bm w},A\cup\{{\bm x}\})=\rho({\bm w},A\cup\{{\bm y}'\})=0.4$ together with $\rho({\bm x},A\cup\{{\bm x}\})\approx1.3\times10^{-3}$, $\rho({\bm z},A\cup\{{\bm x}\})\approx0.5987$ and $\rho({\bm z},A\cup\{{\bm y}'\})\approx0.5965<0.5987$, so the configuration is calibrated and the conclusion of Definition \ref{norm} fails.

Thus ${\bm x}\succsim_n^\rho{\bm y}$ and ${\bm y}\succsim_n^\rho{\bm y}'$ but not ${\bm x}\succsim_n^\rho{\bm y}'$: $\succsim_n^\rho$ is intransitive and \textsf{Normative Weak Order} fails. The mechanism is exactly the one isolated above---the calibration $q$ of any admissible configuration is forced below $\tfrac12$, which is where $D_{0,1}$ is flat, so the difference between the levels is real but unobservable in precisely the comparisons that Definition \ref{norm} inspects.

We claim no more about this example than what is used above: $(u,\varphi,c)$ satisfies Definition \ref{norm-APU}(i)--(iii) and generates a well-defined $\rho$ with full support, and $\succsim_n^\rho$ is intransitive. Whether the remaining axioms of Theorem \ref{representation_norm}(a) hold here is immaterial---one violation already shows that (i)--(iii) cannot deliver part (a), which is what (iv) repairs.
\end{example}

We show that a tuple satisfying Definition \ref{norm-APU}(i)--(iii) and representing a $\rho$ that satisfies the axioms of Theorem \ref{representation_norm}(a) must satisfy (iv). Suppose not, so that for two levels $s<s'$ the difference $D_{s,s'}$ is constant on $(0,\tfrac12)$ and non-constant on $(0,1)$. Every level in $\varphi(X)$ is attained by allocations of arbitrary selfish payoff (the richness proviso). Fix ${\bm x}$ at level $s$ and ${\bm y}$ at level $s'$ with $u(x_1)>u(y_1)$.

First, the flatness propagates to the intermediate levels, which is what lets us quantify over \emph{all} menus. Definition \ref{norm} ranges over every $A$, and the two menus $A\cup\{{\bm x}\}$ and $A\cup\{{\bm y}\}$ have ceilings $\max\{t,s\}$ and $\max\{t,s'\}$ where $t:=\varphi(A)$; write $D$ for the corresponding difference of marginal costs. If $t\le s$ then $D=D_{s,s'}$; if $t\ge s'$ the two ceilings coincide and $D\equiv0$; and if $s<t<s'$ then $D=D_{t,s'}$, which is again constant on $(0,\tfrac12)$, because $D_{s,s'}=D_{s,t}+D_{t,s'}$ with both summands non-decreasing by (iii), so neither can vary where the sum does not. In every case, then, $D$ is non-decreasing and constant on $(0,\tfrac12)$, say $\equiv\bar D$ there, whence $D\ge\bar D$ throughout.

Now, in any configuration calibrated at ${\bm w}$ the calibration level satisfies $q<\tfrac12$, so $\lambda_{\bm x}-\lambda_{\bm y}=D(q)=\bar D$; the first-order conditions then give $p^{\bm y}_{\bm e}\le p^{\bm x}_{\bm e}$ for every common element and $p_{\bm y}<p_{\bm x}$ for the odd ones, exactly as in Example \ref{ex:eta_failure}, so the probabilities under the higher ceiling would sum to less than one. Hence no calibrated configuration compares ${\bm x}$ and ${\bm y}$, and ${\bm x}\succsim_n^\rho{\bm y}$ holds vacuously, while ${\bm y}\succsim_n^\rho{\bm x}$ holds because $\varphi({\bm y})>\varphi({\bm x})$. Now take ${\bm y}'$ at level $s'$ with $u(y'_1)-u(x_1)$ large---this is where the richness proviso is used, a second allocation at the level $s'$ being needed. The second family of Step 3 in the proof of Lemma \ref{lem:calibrated_config} needs only a \emph{losing} level---some $R$ with $D(R)>D(q)$---and no gaining elements at all: taking $N=0$ and $A=\{{\bm w},{\bm z}\}$ with $R$ in the rising range of $D$ above $\tfrac12$, one has $p_{\bm y}=p_{\bm x}+L(R)$ with $L(R)>0$, so raising $R$ until $p_{\bm x}\downarrow0$ drives $d\to+\infty$. This realizes every sufficiently large value of $u(y'_1)-u(x_1)$, with $D(R)>D(q)$. So ${\bm x}\succsim_n^\rho{\bm y}'$ fails, while $\varphi({\bm y})=\varphi({\bm y}')$ gives ${\bm y}\sim_n^\rho{\bm y}'$. The relation $\succsim_n^\rho$ is therefore intransitive, contradicting \textsf{Normative Weak Order}. Adding (iv) to Definition \ref{norm-APU} thus leaves both directions of Theorem \ref{representation_norm} intact: it excludes exactly those tuples $(u,\varphi,c)$ whose induced behavior violates the axioms. Without the richness proviso the argument still shows that a tuple violating (iv) produces a vacuous ranking of some pair at the two levels; it is only the construction of the third allocation, and hence the explicit cycle, that needs a second element at a given level. In that generality (iv) is simply imposed, at the cost noted above of being satisfied by every scale family.

\section{Relation to the Conditions of \citet{FIS_2015}}
\label{app:fis}

Two of the conditions in Section \ref{axioms} are close relatives of conditions in \citet{FIS_2015}, and in both cases the difference is easy to misread. This appendix records them, so that the main text can state the axioms without pausing over their provenance.

\subsection*{\textsf{Normative Independence} and Ordinal IIA}

\citet{FIS_2015} show (their Theorem 1) that the \emph{unrestricted} version of this condition, stated with a \emph{single} function $g$ (not a family $g_{\bm a}$) applying identically to every menu regardless of any ceiling structure, is equivalent to their baseline, fully menu-\emph{independent} APU representation (a single cost function $c(\cdot)$, with no $t$-dependence at all). That stronger, single-$g$ statement would indeed collapse our model's ceiling-dependent cost curvature (Normative Selectivity) and is not what we impose. Axiom \ref{ax:norm_independence} instead keeps a separate rescaling for every $\sim_n^\rho$-class of ceiling-setting allocations---so the \emph{shape} of $g_{\bm a}$ as a function of probability is completely unrestricted across classes, and each $c(\cdot,t)$ may have its own curvature---while still requiring cross-menu comparisons of common elements to line up through whichever rescaling is relevant, \emph{including} across menus of different ceilings. It is exactly this cross-ceiling discipline, together with \textsf{Normative Richness} (Axiom \ref{richness_below_ceiling}), that lets us construct, in the proof of Claim \ref{claim:well_defined_c} (Appendix \ref{app:sufficiency}), a single selfish utility $u$ valid simultaneously at every ceiling, without constraining how $c(\cdot,t)$'s curvature varies with $t$.

\subsection*{\textsf{Richness}}

We next record a technical condition on the range of observed stochastic behavior. It is not ours: it is verbatim the condition \citet[Definition 8]{FIS_2015} impose, alongside Continuity, to characterize the additive perturbed utility model from Positivity and Acyclicity. We state it because it is the most transparent of the richness requirements---it speaks directly about attainable choice probabilities---and because the discussion of the domain in Section \ref{axioms} appeals to it.
\begin{definition}[{\textsf{Richness}; Definition 8 of \citealp{FIS_2015}}]
\label{richness}
For any ${\bm x} \in X$ and $p, q \in (0, 1)$ such that $p + q \leq 1$, there exist ${\bm y}, {\bm z} \in X$ (not necessarily distinct) such that
\begin{itemize}
	\item $\rho({\bm x}, \lbrace {\bm x}, {\bm y}, {\bm z} \rbrace) = p$, and
	\item $\rho({\bm y}, \lbrace {\bm x}, {\bm y}, {\bm z} \rbrace) = q$.
\end{itemize}
\end{definition}
Two remarks on this statement. First, the menu $\{{\bm x},{\bm y},{\bm z}\}$ must have (up to) three \textit{distinct} elements, not two: in a binary menu $\{{\bm x},{\bm y}\}$, Positivity and the budget constraint force $\rho({\bm x},\{{\bm x},{\bm y}\}) + \rho({\bm y},\{{\bm x},{\bm y}\}) = 1$ exactly, so $p+q<1$ could never be achieved there; the third element ${\bm z}$ is what absorbs the remaining probability $1-p-q$ when $p+q<1$ is required. Second, because $u$ is selfish (Section \ref{axioms}, Model), Definition \ref{richness} implies that the range of $u$, restricted to the values $u(x_1)$ actually attained on $X$, equals $\mathbb{R}$, i.e., $u(X) = \mathbb{R}$, and that there are at least three allocations with each utility level; this is why $X$ is not assumed bounded above.

A third remark, on the status of this condition, and it is one of our results. \textsf{Richness} is \emph{not} one of the axioms of Theorem \ref{representation_norm}(a): Proposition \ref{prop:richness_redundant} shows that it \emph{follows} from \textsf{Positivity}, \textsf{Selfishness-Acyclicity}, \textsf{Normative Independence} and \textsf{Normative Richness} (Axiom \ref{norm_richness}). In other words, the technical condition that \citet{FIS_2015} assume is implied by our normative axioms---the structure that the normative ceiling imposes is enough to generate the probabilistic richness their characterization needs. Listing a provably redundant condition among the axioms of a characterization would be a defect; deriving it is a result.

The two richness conditions are easy to confuse, so we record the difference. \textsf{Richness} is about attainable \emph{choice probabilities}: it asks for a ternary menu realizing any prescribed pair $(p,q)$, and places no restriction on the normative ceiling of the menu it produces. \textsf{Normative Richness} (Axiom \ref{norm_richness}) is about attainable \emph{selfish payoffs at a prescribed normative level}: it asks for an allocation with a given $w_1$ lying weakly below the level of a given ${\bm a}$, and it is stated relative to the derived order $\succsim_n^\rho$ rather than to $\rho$ itself. Neither implies the other directly, and the direction we need---payoff richness at a controlled ceiling---is the one \textsf{Richness} does not deliver. Conversely, \textsf{Richness} is implied by our axioms, which is the content of Proposition \ref{prop:richness_redundant}; we therefore state it but do not assume it.

\section{Independence of the Axioms}
\label{app:independence}
In this section we examine the logical independence of the axioms of Theorem \ref{representation_norm}(a). Throughout, \textsf{Positivity} (Axiom \ref{positivity}) is a maintained hypothesis of the theorem rather than one of the axioms in part (a), so it is not at issue here.

\paragraph*{Method.} Each example below is a perturbation, in exactly one dimension, of the canonical model of Example \ref{ex:base_model}: $X=\mathbb{R}^2$, $u(x_1)=x_1$, $\varphi({\bm x})=x_2$, and $c(p,t)=\gamma(t)\kappa(p)$ with the entropy kernel $\kappa(p)=p\ln p+(1-p)\ln(1-p)$ and $\gamma$ continuous, positive and strictly increasing. That model satisfies every axiom (Example \ref{ex:base_model}), so it is the natural baseline: we specify $\rho$ by the same first-order conditions, changing one ingredient---the utility index, the normative index, the ceiling-dependence of the cost, the cost kernel, or the geometry of $X$---and then check which axioms survive. Note that a perturbed $\rho$ of this kind is generally \emph{not} a norm-dependent APU; that is the point, since by Theorem \ref{representation_norm} any $\rho$ satisfying all the axioms would be one.

A word on what follows. There is an example for each of the six axioms of Theorem \ref{representation_norm}(a), with the two halves of \textsf{Normative Richness} separated from one another as well, so seven in all; in each, the remaining conditions survive. The one qualification concerns \textsf{Normative Weak Order}, where the calibrated ranking becomes incomplete and we verify only that the other axioms have no evident conflict with it, as explained in that example. Two further conditions used in the paper are not candidates for separation at all, because they are implied by the axioms rather than independent of them: the compression property of Definition \ref{ax:norm_penalty} (Proposition \ref{prop:penalty_implied}) and \citeauthor{FIS_2015}'s \textsf{Richness} (Proposition \ref{prop:richness_redundant}). The final subsection records both, together with a structural restriction on how \textsf{Normative Independence} can fail. The perturbations below act on four ingredients: the utility index, the normative index, the ceiling-dependence of the cost, and the cost kernel itself---the last two being the least obvious, and used in the \textsf{Continuity} and \textsf{Normative Selectivity} examples respectively.

The following observation does most of the routine work in the examples that perturb the utility index, and is worth stating once.

\begin{remark}[\textsf{Transport along a monotone reparametrization of the selfish coordinate}]
\label{rmk:transport}
Let $\rho$ be generated by the first-order conditions of Example \ref{ex:base_model} with the selfish payoff $x_1$ replaced by $v(x_1)$, where $v:\mathbb{R}\to\mathbb{R}$ is continuous, non-decreasing and onto. Write $x_1' := v(x_1)$ for the transported coordinate. Then:
\begin{enumerate}
    \renewcommand{\labelenumi}{(\roman{enumi})}
    \item In the transported coordinate, $\rho$ \emph{is} the canonical model of Example \ref{ex:base_model} with identity utility, so every axiom holds when stated in terms of $x_1'$.
    \item Allocations with the same $x_1'$ receive the same probability in any menu (immediate from the first-order conditions and strict convexity); call this the \emph{tie property}.
    \item Consequently the calibrated ranking $\succsim_n^\rho$ of Definition \ref{norm} is the same whether its payoff comparisons are read in $x_1$ or in $x_1'$. Indeed, the two readings differ only on pairs with $z_1 > w_1$ but $v(z_1)=v(w_1)$, and for such pairs the tie property makes the required conclusion hold with equality.
\end{enumerate}
Hence \textsf{Normative Weak Order}, \textsf{Normative Penalty}, \textsf{Normative Selectivity}, \textsf{Normative Independence} and \textsf{Normative Richness}~(3)---all of which are statements about $\succsim_n^\rho$ and about probabilities---transfer from the transported model to the original one. What does \emph{not} transfer automatically is \textsf{Selfishness-Acyclicity}, whose hypothesis compares raw payoffs $x_1$ directly and is therefore triggered by exactly the pairs on which the two readings differ; nor \textsf{Richness}, whose quantifier fixes one allocation (see the example below).
\end{remark}

\subsection*{Violation of Axiom \ref{continuity} (Continuity)}
Here the perturbation acts on the ceiling-dependence of the cost, rather than on the utility or the normative index, both of which stay continuous. Keep $X=\mathbb R^2$, $u(x_1)=x_1$, $\varphi({\bm x})=x_2$ and $c(p,t)=\gamma(t)\kappa(p)$, but take $\gamma$ positive and strictly increasing yet \emph{discontinuous}:
\begin{equation*}
    \gamma(t) := \begin{cases} e^{t}, & t<0,\\ e^{t}+1, & t\ge0.\end{cases}
\end{equation*}
Every other axiom is untouched, because nothing outside \textsf{Continuity} uses more of $\gamma$ than positivity and strict monotonicity: \textsf{Positivity} and \textsf{Selfishness-Acyclicity} are properties of the logistic first-order condition, \textsf{Normative Weak Order} holds with $\varphi=x_2$ representing $\succsim_n^\rho$ (its closedness clause needs the continuity of $\varphi$, not of $\gamma$), \textsf{Normative Penalty} and \textsf{Normative Selectivity} follow from $\gamma$ being increasing, \textsf{Normative Independence} holds with $g_{\bm a}=\gamma(a_2)\kappa'$, and the three richness axioms hold because $u$ is onto $\mathbb R$ and $\kappa'$ maps $(0,1)$ onto $\mathbb R$ at every ceiling. But \textsf{Continuity} (Axiom \ref{continuity}) fails at a level crossing. Take
\begin{equation*}
    {\bm x}=(0,-1), \qquad {\bm y}^n=\bigl(1,-\tfrac{1}{n+1}\bigr) \longrightarrow {\bm y}=(1,0).
\end{equation*}
The menu $\{{\bm x},{\bm y}^n\}$ has ceiling $-\tfrac1{n+1}<0$, where $\gamma=e^{-1/(n+1)}\to1$, and the menu $\{{\bm x},{\bm y}\}$ has ceiling $0$, where $\gamma(0)=2$. By the binary formula $\rho({\bm y},\{{\bm x},{\bm y}\})=\sigma\bigl((y_1-x_1)/2\gamma\bigr)$ of Example \ref{ex:base_model}, with $y_1-x_1=1$,
\begin{equation*}
    \rho({\bm y}^n,\{{\bm x},{\bm y}^n\}) = \sigma\Bigl(\tfrac12 e^{1/(n+1)}\Bigr) \longrightarrow \sigma\bigl(\tfrac12\bigr)\approx0.622,
    \qquad
    \rho({\bm y},\{{\bm x},{\bm y}\}) = \sigma\bigl(\tfrac14\bigr)\approx0.562 .
\end{equation*}
The limit of the probabilities is not the probability of the limit, so \textsf{Continuity} is violated while everything else holds. Two points stand out. First, the closedness clause of \textsf{Normative Weak Order} controls only the continuity of the \emph{normative index}; the continuity of the cost's dependence on the ceiling is a separate requirement, and \textsf{Continuity} is what supplies it. Second, the failure is invisible in any single menu---each menu here is a perfectly well-behaved perturbed-utility model---and shows up only in how menus vary, which is exactly the kind of condition a continuity axiom is for.

\subsection*{Violation of Axiom \ref{ax:norm_weak_order} (Normative Weak Order)}
Consider a model in which the menu index driving the cost is a sum of \emph{two} ceiling-like indices. For any menu $A \in \mathcal{A}$, let \[ K(A) \;=\; \max_{{\bm a} \in A} a_2 \;+\; \max_{{\bm a} \in A} (a_1+a_2), \] and let the cost be $C_A(p) = e^{K(A)}\,\kappa(p)$ with $\kappa$ as in Example \ref{ex:base_model}; keep $u(x_1)=x_1$ and $X=\mathbb R^2$. The exponential is not cosmetic: $K$ takes negative values, and a negative multiple of the convex kernel $\kappa$ would not be a cost function at all. Any positive, strictly increasing transform of $K$ would serve equally well, since only the \emph{ordering} of the two menu indices matters below. Each summand of $K$ is the maximum over the menu of a continuous index, but the two indices order allocations differently, and it is this conflict---rather than any failure of regularity---that breaks the axiom. Because $K$ is continuous in the allocations and $\kappa$ is unchanged, \textsf{Continuity}, \textsf{Positivity} and \textsf{Selfishness-Acyclicity} hold exactly as before.

For a calibrated pair ${\bm z},{\bm w}$ in a menu $A$ with $\rho({\bm w},A\cup\{{\bm x}\})=\rho({\bm w},A\cup\{{\bm y}\})=:q$ and $z_1>w_1$, the first-order conditions give $\kappa'(\rho({\bm z},A\cup\{\cdot\})) - \kappa'(q) = (z_1-w_1)/e^{K(A\cup\{\cdot\})}$, so \[ \rho({\bm z},A\cup\{{\bm x}\}) \le \rho({\bm z},A\cup\{{\bm y}\}) \quad\Longleftrightarrow\quad K(A\cup\{{\bm x}\}) \ge K(A\cup\{{\bm y}\}). \] Thus Definition \ref{norm} reduces here to a comparison of menu indices, and we claim \[ {\bm a}\succsim_n^\rho{\bm w} \quad\Longleftrightarrow\quad w_2\le a_2 \ \text{ and } \ w_1+w_2 \le a_1+a_2 . \] The implication ``$\Leftarrow$'' is immediate from the displayed equivalence, since both maxima are then weakly larger with ${\bm a}$ in place of ${\bm w}$, whatever the rest of the menu.

For ``$\Rightarrow$'' we must exhibit, for each of the two inequalities separately, a calibrated configuration detecting its failure, and here the citation of Remark \ref{rmk:scale_shortcut} requires care. (The shortcut is what we use, rather than the general construction of Lemma \ref{lem:calibrated_config}: the cost family here is a scale family in the index $K$, and the shortcut keeps the menu down to three elements, which is what lets $K$ be pinned.) Unlike a ceiling scaling $\gamma(t)$, the index $K$ depends on \emph{all} the menu's elements, including the pair ${\bm z},{\bm w}$ whose payoffs the tie-configuration argument varies; were a maximum attained at that pair, the two ``scalings'' would move along the calibration path and the argument would not apply. What is needed is only that each maximum be \emph{pinned} over the range the argument traverses, and the calibration pair can do this itself---no additional menu element is required. Take ${\bm x}=(0,4)$ and ${\bm y}=(5,0)$, which are incomparable in the displayed order ($x_2>y_2$ while $x_1+x_2<y_1+y_2$), and consider the three-element menus $\{{\bm z},{\bm w}\}\cup\{{\bm x}\}$ and $\{{\bm z},{\bm w}\}\cup\{{\bm y}\}$ in the following two configurations.
\begin{itemize}
    \item \emph{Pinning $\max a_2$.} Let ${\bm z}=(T,5)$ and ${\bm w}=(\beta,w_2)$ with $w_2\le5$, $w_1+w_2\le4$ and $\beta<T\le-1$. Then $\max_{{\bm a}} a_2 = 5$ is pinned at ${\bm z}$, whose second coordinate does not move, while $\max_{{\bm a}}(a_1+a_2)$ is attained at the added element. The indices are $5+4=9$ with ${\bm x}$ and $5+5=10$ with ${\bm y}$. Since $9<10$, the ${\bm x}$-menu compresses less, so at a calibrated pair $\rho({\bm z},\{{\bm z},{\bm w},{\bm x}\})>\rho({\bm z},\{{\bm z},{\bm w},{\bm y}\})$ and ${\bm x}\succsim_n^\rho{\bm y}$ fails.
    \item \emph{Pinning $\max(a_1+a_2)$.} Let ${\bm w}=(\beta,5-\beta)$ and ${\bm z}=(T,5-T)$ lie on the line $a_1+a_2=5$, with $5.05\le\beta<T$. Then $\max_{{\bm a}}(a_1+a_2)=5$ is pinned by the pair itself---moving $\beta$ and $T$ along the line does not disturb it---while both second coordinates are negative, so $\max_{{\bm a}} a_2$ is attained at the added element. The indices are $4+5=9$ with ${\bm x}$ and $0+5=5$ with ${\bm y}$. Now $5<9$, so the ${\bm y}$-menu compresses less and ${\bm y}\succsim_n^\rho{\bm x}$ fails.
\end{itemize}
In each configuration both maxima are constant over the stated parameter range, so the two indices are genuine constants there and the tie-configuration computation of Remark \ref{rmk:scale_shortcut} applies with $e^{9},e^{10}$ (respectively $e^{9},e^{5}$) in the role of the cost scalings. One point deserves note: the indices are pinned only on a bounded window of parameters, so the two tie points must be chosen inside it. This is possible because the threshold $v^\ast$ of Remark \ref{rmk:scale_shortcut} falls in the interior of the window. Taking the windows $[-6,-1]$ and $[\tfrac{101}{20},6]$, the required statements are
\begin{equation*}
    -6 \;<\; v^\ast_1 = -\frac{5}{e-1} \;<\; -1
    \qquad\text{and}\qquad
    \frac{101}{20} \;<\; v^\ast_2 = \frac{5e^4}{e^4-1} \;<\; 6 ,
\end{equation*}
which reduce to $11<6e$ and $e<6$ for the first, and to $e^4<101$ and $6<e^4$ for the second---all four immediate from $2<e<3$. Numerically $v^\ast_1\approx-2.91$ and $v^\ast_2\approx5.093$; note how little room the second leaves, which is why the window is stated rather than left implicit.

Two features here are easy to get wrong, so we flag them. First, the calibration pair may serve as its own carrier: in the second configuration it is the constraint $a_1+a_2=5$ on the pair, rather than any extra element, that pins the maximum. Second, some such device is indispensable: with ${\bm z},{\bm w}$ of very negative coordinates, so that the added element attains \emph{both} maxima, the index difference collapses to $(x_1+2x_2)-(y_1+2y_2)$, which is a \emph{complete} order on allocations and can therefore refute at most one of the two directions. Incompleteness is visible only to configurations that isolate the two indices separately.

Hence $\succsim_n^\rho$ is the intersection of two distinct orders and is incomplete.

The richness axioms survive this change, which is why the example isolates \textsf{Normative Weak Order} rather than colliding with them. Given ${\bm a}$ and a target payoff $v$, the allocation ${\bm w}=(v,w_2)$ satisfies both conditions above once $w_2$ is small enough, so \textsf{Normative Richness} holds---with infinitely many witnesses and a continuous selection $w_2\mapsto(v,w_2)$, giving all three of its clauses. Since $w_1$ remains free below any prescribed ceiling, $u(y_1)-u(a_1)$ sweeps $\mathbb R$ and \textsf{Normative Richness}~(3) holds; \textsf{Richness} holds as in Example \ref{ex:base_model}. (This is exactly what fails if one uses $\max a_1 + \max a_2$ as the index: the first summand makes $\succsim_n^\rho$ rank by $x_1$, so ${\bm a}\succsim_n^\rho{\bm w}$ forces $w_1\le a_1$ and both \textsf{Normative Richness} and \textsf{Normative Richness}~(3) break down. The index must leave the selfish coordinate free below a ceiling.)

The mechanism is that $K$ is not the maximum of any single index over the menu, so no ranking of allocations can summarize its menu-level behavior---exactly what \textsf{Normative Weak Order} rules out. We claim only what the construction establishes, namely the failure of \textsf{Normative Weak Order}. \textsf{Normative Penalty}, \textsf{Normative Selectivity} and \textsf{Normative Independence} are stated for configurations built from $\succsim_n^\rho$---dominating menu pairs, and menus possessing a ceiling-setting element---and since $\succsim_n^\rho$ is here incomplete, such configurations are correspondingly scarce; but verifying that those axioms hold requires checking the configurations that \emph{do} exist, which we have not carried out. This example should therefore be read as isolating \textsf{Normative Weak Order} only in the sense that the other axioms have no evident conflict with it, not as a verified one-axiom violation.

\subsection*{Violation of Axiom \ref{ax:selectivity} (Normative Selectivity)}
Here the perturbation acts on the cost kernel, rather than on either index or on the ceiling-dependence of the scale. Keep $X=\mathbb R^2$, $u(x_1)=x_1$ and $\varphi({\bm x})=x_2$, and take
\begin{equation*}
    c(p,t) \;=\; \gamma(t)\kappa(p) + \theta(t)H(p),
    \qquad\text{where}\quad H'(p) \;:=\; -\bigl(p-\tfrac12\bigr)^2 ,
\end{equation*}
with $\gamma(t)=1+\tfrac1{20}\sigma(t)$ and $\theta(t)=\sigma(t)$ for $\sigma$ the logistic function. The cost therefore differs from the canonical one of Example \ref{ex:base_model} by the single term $\theta(t)H(p)$, and the design is dictated by the division of labour recorded in the final subsection below: the compression property of Definition \ref{ax:norm_penalty} reaches only pairs $(q,p)$ whose lower endpoint is the calibration point, and these satisfy $q+p\le1$, whereas \textsf{Normative Selectivity} reaches every pair in $(0,1)$. A perturbation that is well behaved on the first family and ill behaved outside it will keep the former and break the latter.

The cost is admissible. Its marginal and curvature are
\begin{equation*}
    c_1(p,t) = \gamma(t)\ln\frac{p}{1-p} - \theta(t)\bigl(p-\tfrac12\bigr)^2,
    \qquad
    c_{11}(p,t) = \frac{\gamma(t)}{p(1-p)} - 2\theta(t)\bigl(p-\tfrac12\bigr),
\end{equation*}
and since $0<\theta<1\le\gamma$ we have $c_{11}(p,t) \ge 4\gamma(t) - 2\theta(t)\lvert p-\tfrac12\rvert > 3$, so $c(\cdot,t)$ is strictly convex; $c_1(\cdot,t)$ is onto $\mathbb R$ because the logistic term diverges at both endpoints while the added term is bounded. So Definition \ref{norm-APU}(i)--(ii) hold and $\rho$ is well defined by the interior first-order conditions.

What fails is Definition \ref{norm-APU}(iii). For $t'>t$ write $\Delta\gamma,\Delta\theta>0$ for the increments of $\gamma$ and $\theta$. Then
\begin{equation*}
    c_{11}(p,t')-c_{11}(p,t) \;=\; \frac{\Delta\gamma}{p(1-p)} - 2\Delta\theta\bigl(p-\tfrac12\bigr),
\end{equation*}
which at $p=\tfrac34$ equals $\tfrac{16}{3}\Delta\gamma-\tfrac12\Delta\theta$; since $\Delta\gamma=\tfrac1{20}\Delta\sigma$ and $\Delta\theta=\Delta\sigma$, this is $-\tfrac7{30}\Delta\sigma<0$. A higher ceiling therefore makes the cost \emph{less} curved at $p=\tfrac34$, and \textsf{Normative Selectivity} fails by the contrapositive of Proposition \ref{prop:curvature_char}, which derives curvature monotonicity on all of $(0,1)$ from it.

That strict convexity and non-monotone curvature are compatible is the point of the particular coefficients, and the two requirements turn out to constrain different things. Convexity needs $2\theta(t)/\gamma(t) < \min_{0<m<1/2}\bigl[m(\tfrac14-m^2)\bigr]^{-1} = 12\sqrt3$, a bound on the \emph{levels}; non-monotone curvature needs $2\Delta\theta/\Delta\gamma > 12\sqrt3$, a bound on the \emph{increments}. Both hold here, since $\theta/\gamma<1$ while $\Delta\theta/\Delta\gamma=20>6\sqrt3\approx10.39$. Any other pair of coefficient functions with this property would serve equally well.

The remaining axioms hold, and the computation that delivers them is the one the design anticipates. The level difference
\begin{equation*}
    D(p) \;:=\; c_1(p,t')-c_1(p,t) \;=\; \Delta\gamma\ln\frac{p}{1-p} + \Delta\theta\,H'(p)
\end{equation*}
is all that the compression property and the easy direction of the calibrated comparison use, and they use it only at pairs $q<p$ with $q+p\le1$. On those pairs
\begin{equation*}
    H'(p)-H'(q) \;=\; (q-p)(q+p-1) \;\ge\; 0 ,
\end{equation*}
both factors being non-positive, while the logistic term is strictly increasing; hence $D(p)>D(q)$ and higher ceilings compress throughout the family that the property reaches. For the strict direction---that $x_2<y_2$ precludes ${\bm x}\succsim_n^\rho{\bm y}$---the tie-crossing construction of Lemma \ref{lem:calibrated_config} applies, because $\Delta\gamma>0$ makes the logistic term dominate in the tails while the added term is bounded by $\tfrac14$ in absolute value. So $\succsim_n^\rho$ is represented by $x_2$ and \textsf{Normative Weak Order} holds. \textsf{Continuity} holds since $\gamma$ and $\theta$ are continuous; \textsf{Selfishness-Acyclicity} since $u$ is strictly increasing; \textsf{Normative Independence} with $g_{\bm a}=c_1(\cdot,a_2)$; and the richness axioms since $u$ is onto $\mathbb R$ and each $c_1(\cdot,t)$ maps $(0,1)$ onto $\mathbb R$. The example is thus the canonical model with one term added, and the single condition it costs is \textsf{Normative Selectivity}.

\subsection*{Violation of Axiom \ref{ax:selfish_acyclicity} (Selfishness-Acyclicity)}
Keep $X$, $\varphi$, $\gamma$ and $\kappa$ as in Example \ref{ex:base_model}, but replace the utility index by the \emph{plateau} function
\begin{align*}
    v(x_1) \;:=\;
    \begin{cases}
        x_1, & x_1\le 0,\\
        0, & 0\le x_1\le 1,\\
        x_1-1, & x_1\ge 1,
    \end{cases}
\end{align*}
which is continuous, non-decreasing, constant exactly on $[0,1]$, and---this is what matters below---a surjection onto $\mathbb R$, unbounded both above and below. Define $\rho$ by the first-order conditions $v(x_1) - \gamma(\varphi(A))\kappa'(\rho({\bm x},A)) = \lambda_A$ together with the budget identity.

\textsf{Selfishness-Acyclicity} fails: take any menu $A$ containing ${\bm x},{\bm y}$ with $x_1 = 1$ and $y_1 = 0$. Then $v(x_1)=v(y_1)$, so $\rho({\bm x},A)=\rho({\bm y},A)$. The two-cycle $({\bm x},{\bm y})$ with $A^1=A^2=A$ satisfies the hypothesis of Axiom \ref{ax:selfish_acyclicity}---the first link holds because $x_1 > y_1$ (strictly), the second because $\rho({\bm y},A)\ge\rho({\bm x},A)$---and therefore witnesses a violation.

All the other axioms survive. Since $v$ is continuous, non-decreasing and onto, Remark \ref{rmk:transport} applies with $x_1'=v(x_1)$: the transported model is the canonical one, so \textsf{Normative Weak Order}, \textsf{Normative Penalty}, \textsf{Normative Selectivity}, \textsf{Normative Independence} and \textsf{Normative Richness}~(3) transfer, and \textsf{Continuity} holds as in Example \ref{ex:base_model}. \textsf{Normative Richness} is unaffected because it constrains payoffs rather than utilities. Only \textsf{Richness} requires separate attention, because its quantifier fixes ${\bm x}$ and asks for ${\bm y},{\bm z}$ realizing target probabilities in $\{{\bm x},{\bm y},{\bm z}\}$: here surjectivity of $v$ is what saves it, since $v(y_1),v(z_1)$ may be taken arbitrarily far above or below the fixed $v(x_1)$ while $\kappa'$ maps $(0,1)$ onto $\mathbb R$. This is precisely where an index bounded above, such as $\min\{x_1,0\}$, would fail: for ${\bm x}$ at the top of its range, $\rho({\bm x},\cdot)$ could never be made small, so \textsf{Richness} and \textsf{Normative Richness}~(3) would break along with \textsf{Selfishness-Acyclicity}, and the example would isolate nothing---which is exactly the configuration the clause (3) example below puts to use. Since $v$ is non-decreasing, the calibrated comparison of Definition \ref{norm} yields, for $z_1>w_1$, $\kappa'(\rho({\bm z},\cdot)) - \kappa'(q) = [v(z_1)-v(w_1)]/\gamma(\cdot) \ge 0$, so higher ceilings still weakly compress: $\succsim_n^\rho$ is again represented by $x_2$ and \textsf{Normative Weak Order}, \textsf{Normative Penalty} and \textsf{Normative Selectivity} all hold (the relevant inequalities are weak, and hold with equality on the flat part of $v$). Finally $g_{\bm a}:=\gamma(a_2)\kappa'$ satisfies \textsf{Normative Independence}. Thus \textsf{Selfishness-Acyclicity} is not implied by the remaining axioms.

\subsection*{Violation of Axiom \ref{ax:norm_independence} (Normative Independence)}
Keep $X=\mathbb R^2$, $\varphi({\bm x})=x_2$ and $\kappa$ as in Example \ref{ex:base_model}. Let $\theta:\mathbb R\to[0,1]$ be continuous and non-decreasing with $\theta(t)=0$ for $t\le0$ and $\theta(t)=1$ for $t\ge1$, and define, for a menu $A$ of ceiling $t=\varphi(A)$, the \emph{ceiling-dependent} utility index and cost
\begin{align*}
    \tilde u_t(x_1) \;:=\; x_1 + \theta(t)\sin x_1, \qquad
    c(p,t) \;:=\; \bigl(1+\theta(t)\bigr)\kappa(p),
\end{align*}
with $\rho$ determined by $\tilde u_t(x_1) - c_1(\rho({\bm x},A),t) = \lambda_A$ and the budget identity. Each $\tilde u_t$ is continuous, non-decreasing (its derivative is $1+\theta(t)\cos x_1 \ge 0$, vanishing only at isolated points) and onto $\mathbb R$.

\textsf{Normative Independence} fails. Within any single ceiling class the model is an ordinary single-cost APU, so the axiom's within-class content holds with $g$ proportional to $\kappa'$; what fails is the cross-class requirement, because the induced utility indices are $\tilde u_0 = \mathrm{id}$ and $\tilde u_1 = \mathrm{id}+\sin$, which are not affinely related. The argument must be made through the axiom's own quantifiers, since a choice probability depends on the whole menu and not on $x_1-y_1$ alone. Suppose, for contradiction, that a family $\{g_{\bm a}\}$ as in Axiom \ref{ax:norm_independence} exists, and write
\begin{equation*}
    \Phi_{\bm a}({\bm x},{\bm y}) \;:=\; g_{\bm a}(\rho({\bm x},A))-g_{\bm a}(\rho({\bm y},A))
\end{equation*}
for a menu $A$ whose ceiling is set by ${\bm a}$ and which contains ${\bm x},{\bm y}$. Applying the axiom to two such menus shows $\Phi_{\bm a}$ is well defined---independent of which ceiling-${\bm a}$ menu is used---and applying it across classes gives $\Phi_{\bm a}=\Phi_{\bm b}$ whenever ${\bm a}$ has $\theta(a_2)=0$ and ${\bm b}$ has $\theta(b_2)=1$.

Now translate. In the $\theta=0$ class the utility index is the identity, so replacing ${\bm x},{\bm y}$ by allocations of payoffs $x_1+c,\ y_1+c$ and adjusting one further element of the menu leaves every probability unchanged (the multiplier simply shifts by $c$); hence $\Phi_{\bm a}$ is translation invariant, and therefore so is $\Phi_{\bm b}$. But in the $\theta=1$ class, matching the translated pair's probabilities to the original pair's forces
\begin{equation*}
    \sin(y_1+c)-\sin(x_1+c) \;=\; \sin y_1-\sin x_1 ,
\end{equation*}
which fails for $x_1=0$, $y_1=\pi/2$, $c=\pi/2$: the left-hand side is $\sin\pi-\sin(\pi/2)=-1$ and the right-hand side is $\sin(\pi/2)-\sin 0=1$. This contradiction needs no Cauchy-type or limiting argument; it uses only that the linear class is translation invariant while $\sin$ is not.

Three features of this refutation are worth noting, since they show how little it rests on. The witness menus are small---two elements in the $\theta=0$ class and three in the $\theta=1$ class---and, as in the \textsf{Normative Weak Order} example, the calibration pair serves as its own carrier, so no auxiliary element is needed to pin anything. The clause of Axiom \ref{ax:norm_independence} that does the work is the $\sim_n^\rho$-invariance of the family $\{g_{\bm a}\}$: it is what forces the same rescaling on the two menus of a class, and here $\theta(0)=\theta(0)$ and $\theta(1)=\theta(1)$ supply the required equivalences. Finally, the regularity clauses of the axiom---$C^1$, strict monotonicity aside, and the boundary behavior---are never used: the violation is a failure of the cross-menu identity itself, not an artifact of how well behaved the rescalings are asked to be.

The remaining axioms hold. Write $d:=z_1-w_1>0$ and $s:=\sin z_1-\sin w_1$, so $|s|\le d$ because $\sin$ is $1$-Lipschitz. For ceilings $t\ge t'$ the calibrated comparison of Definition \ref{norm} and of Definition \ref{ax:norm_penalty} reduces to
\begin{align*}
    \frac{d+\theta(t)s}{1+\theta(t)} \;\le\; \frac{d+\theta(t')s}{1+\theta(t')}
    \quad\Longleftrightarrow\quad
    s\,\bigl(\theta(t)-\theta(t')\bigr) \;\le\; d\,\bigl(\theta(t)-\theta(t')\bigr),
\end{align*}
which holds because $\theta(t)\ge\theta(t')$ and $s\le d$; moreover $s<d$ strictly whenever $z_1>w_1$, since $t\mapsto t-\sin t$ is strictly increasing. Hence higher ceilings compress, and $\succsim_n^\rho$ is represented by $\theta({x_2})$---\emph{not} by $x_2$ itself: since $\theta$ is constant on $(-\infty,0]$ and on $[1,\infty)$, two allocations whose second coordinates lie in the same flat stretch induce the same temperature and the same utility index, hence identical choice probabilities, and are therefore $\sim_n^\rho$. As $\theta$ is continuous and non-decreasing, $\theta(x_2)$ is a continuous representation, so \textsf{Normative Weak Order} holds (this is also why the model is normalized only after the flat stretches are merged, as in Definition \ref{def:level_separation}). \textsf{Normative Penalty} and---by the same computation applied to increments---\textsf{Normative Selectivity} hold as well. \textsf{Continuity} holds since $\theta$ is continuous; \textsf{Selfishness-Acyclicity} holds since each $\tilde u_t$ is non-decreasing and strictly increasing overall; and the three richness axioms hold since each $\tilde u_t$ is onto $\mathbb R$ and $\kappa'$ maps $(0,1)$ onto $\mathbb R$. This example is also instructive in its own right: by the argument in the last paragraph of this appendix, \emph{every} violation of \textsf{Normative Independence} that preserves \textsf{Normative Penalty} must be of this cross-ceiling, non-affine type.

\subsection*{Violation of Axiom \ref{norm_richness} (Normative Richness, clauses (1) and (2))}
Replace the domain by the vertical strip $X:=\{{\bm x}\in\mathbb R^2 : |x_1| < 1\}$, which is convex and unbounded, and take
\begin{align*}
    u(x_1) \;:=\; \tan\!\left(\frac{\pi x_1}{2}\right),
\end{align*}
which is continuous, strictly increasing and onto $\mathbb R$ on $(-1,1)$; keep $\varphi({\bm x})=x_2$ and $c(p,t)=\gamma(t)\kappa(p)$.

\textsf{Normative Richness} fails: no allocation in $X$ has $x_1 = 2$, so for any ${\bm a}$ and $v=2$ there is no ${\bm w}$ with $w_1=v$ at all, let alone one with ${\bm a}\succsim_n^\rho{\bm w}$. The other axioms are unaffected, because all of them constrain \emph{probabilities} rather than payoff levels, and $u$ still has full range on $X$: differences $u(x_1)-u(y_1)$ sweep $\mathbb R$, so \textsf{Richness} and \textsf{Normative Richness}~(3) hold; $\succsim_n^\rho$ is again represented by $x_2$, giving \textsf{Normative Weak Order}, \textsf{Normative Penalty} and \textsf{Normative Selectivity}; and \textsf{Continuity}, \textsf{Selfishness-Acyclicity} and \textsf{Normative Independence} hold as in Example \ref{ex:base_model}. This isolates what \textsf{Normative Richness} adds: control of \emph{payoff levels} below a prescribed ceiling, which no probability-based axiom delivers.

\subsection*{Violation of Axiom \ref{norm_richness} (Normative Richness, clause (3))}
The perturbation acts on the utility index, as in the \textsf{Selfishness-Acyclicity} example, but in the opposite direction: there the index was non-injective and onto, here it is injective and \emph{bounded}. Keep $X=\mathbb R^2$, $\varphi({\bm x})=x_2$ and $c(p,t)=\gamma(t)\kappa(p)$, and take
\begin{equation*}
    v(x_1) \;:=\; \arctan x_1 ,
\end{equation*}
continuous and strictly increasing, with range $(-\tfrac\pi2,\tfrac\pi2)$.

Clause (3) fails immediately. In a binary menu $\{{\bm a},{\bm y}\}$ of ceiling $t$ the first-order conditions give $v(y_1)-v(a_1) = 2\gamma(t)\ln\bigl(p/(1-p)\bigr)$ with $p:=\rho({\bm y},\{{\bm a},{\bm y}\})$, so the right-hand side would have to sweep $\mathbb R$ as $p$ sweeps $(0,1)$, while the attainable left-hand side is confined to the bounded interval $\bigl(-\tfrac\pi2-\arctan a_1,\ \tfrac\pi2-\arctan a_1\bigr)$. For $\gamma(t)=1$ and $a_1=0$ the probabilities that ${\bm y}$ can be given below ${\bm a}$ are exactly those in $(0.3132,0.6868)$, not all of $(0,1)$. \citeauthor{FIS_2015}'s \textsf{Richness} fails for the same reason---as it must, since Proposition \ref{prop:richness_redundant} derives it from clause (3) among others.

Everything else survives, and this is the part that is not routine. The worry is that a bounded index also destroys the calibrated \emph{configurations} on which Definition \ref{norm} rests, in which case $\succsim_n^\rho$ would be vacuously defined and \textsf{Normative Weak Order} would break along with clause (3)---the mechanism of Example \ref{ex:eta_failure}. It does not. Shares in a menu of ceiling $t$ are $\sigma\bigl(v(x_1)/\gamma(t)+\ell\bigr)$ for a menu-specific level $\ell$ pinned by the budget identity, and a calibrated configuration for a pair ${\bm x},{\bm y}$ with $x_2>y_2$ can be assembled by \emph{padding}: adjoin $n$ copies of an allocation of selfish payoff $r$, all normatively below ${\bm y}$. As $n$ grows the padding absorbs almost all the probability and both levels are driven to $-\infty$, so the discrepancy between the two menus is governed by the sign of $\arctan r$; taking $r$ of either sign gives discrepancies of opposite sign, and the intermediate value theorem supplies a calibrating $r^\ast$ between them. At that configuration the comparison is strict in the required direction. Hence $\succsim_n^\rho$ is represented by $x_2$, giving \textsf{Normative Weak Order}, the compression property, \textsf{Normative Selectivity} and \textsf{Normative Independence}; \textsf{Continuity} and \textsf{Selfishness-Acyclicity} hold as in Example \ref{ex:base_model}; and clauses (1) and (2) of \textsf{Normative Richness} hold because they constrain \emph{payoffs} $w_1$ rather than utilities, and $X=\mathbb R^2$ leaves those free. This last point is what distinguishes the example from the plateau one, where surjectivity of the index was what saved the richness conditions.

One feature here has no counterpart in the other examples and is worth recording, since it says something about calibration rather than about the axiom. The size of menu that a calibrated comparison requires is not bounded in this model: as the two normative levels approach one another, the number of padding elements needed grows without bound, at a rate of the order of the reciprocal of the level difference. The payoff difference does not have the same effect, being itself bounded by the range of the index. So the calibrated ranking of Definition \ref{norm} is identified only from arbitrarily large menus here, even though every individual menu is perfectly well behaved.

The example also settles a question that the structure of the axioms makes tempting. By \textsf{Selfishness-Acyclicity} and the mechanism of Proposition \ref{prop:richness_redundant}, clause (3) is equivalent to the surjectivity onto $\mathbb R$ of the binary gap below a ceiling-setter ${\bm a}$, namely of $f_{\bm a}(v):=g_{\bm a}(\rho({\bm a},\{{\bm a},{\bm y}\}))-g_{\bm a}(\rho({\bm y},\{{\bm a},{\bm y}\}))$ evaluated at any ${\bm y}$ with $y_1=v$ and ${\bm a}\succsim_n^\rho{\bm y}$, which \textsf{Selfishness-Acyclicity} makes independent of the representative. Its intermediate value half follows from the other axioms: the continuous-selection clause of \textsf{Normative Richness} with $S=\{{\bm a}\}$ gives a continuous $\zeta:\mathbb R\to X\setminus\{{\bm a}\}$ with $\zeta(v)_1=v$ and ${\bm a}\succsim_n^\rho\zeta(v)$, and \textsf{Continuity} applied to the two-element menus $\{{\bm a},\zeta(v)\}$ makes $f_{\bm a}\circ\zeta$ continuous. What is left is unboundedness, and it reduces cleanly to the increment comparison
\begin{equation*}
    {\bm s}\succsim_n^\rho{\bm a} \;\Longrightarrow\;
    g_{\bm s}(q)-g_{\bm s}(p) \;\ge\; g_{\bm a}(q)-g_{\bm a}(p) \qquad (p<q),
\end{equation*}
the axiom-level counterpart of Definition \ref{norm-APU}(iii). The compression property delivers that inequality, but only at pairs realized at a calibrated configuration, and the pair required is the extreme pair $(\varepsilon,1-2\varepsilon)$, whose calibrated configuration at the lower ceiling ${\bm a}$ would require realizing a probability near $1-2\varepsilon$ there---which is the surjectivity in question. One might hope the gap could be closed; the example above shows that it cannot, since clause (3) fails there while the other axioms hold. We record the conditional statement all the same, because it localizes what the clause contributes: \emph{if} the increment comparison holds for all ${\bm s}\succsim_n^\rho{\bm a}$ and all $p<q$, then clause (3) follows from \textsf{Positivity}, \textsf{Selfishness-Acyclicity}, \textsf{Continuity}, the weak-order clauses of \textsf{Normative Weak Order}, \textsf{Normative Independence}, clause (2) and \textsf{Richness}. There is no tension with the example, which violates \textsf{Richness}---one of those hypotheses---for the reason given above.

\subsection*{Derived conditions and a structural consequence}
Two conditions that the paper uses do not appear among the examples above, and the reason is the same in both cases: they are not independent of the axioms but implied by them. We record each, and then a restriction the same reasoning places on how \textsf{Normative Independence} can fail.

\emph{The compression property.} The compression property of Definition \ref{ax:norm_penalty} is implied by the axioms (Proposition \ref{prop:penalty_implied}), so there is no example separating it from them. It is not equivalent to \textsf{Normative Selectivity} either, and the difference matters for two reasons: it is why one is a consequence and the other an axiom, and it is what the \textsf{Normative Selectivity} example above is built on. Writing the interior first-order conditions in integral form, each requires $\int_a^b c_{11}(p,t_A)\,dp \ge \int_a^b c_{11}(p,t_B)\,dp$, but for different families of pairs $a<b$, and the difference is which allocation sits at the ends of the interval. In the compression property of Definition \ref{ax:norm_penalty} the calibration point is the \emph{lower} endpoint: the pair is $(q,p)$ with $q:=\rho({\bm y},A)=\rho({\bm y},B)$ and $p:=\rho({\bm x},B)$, and since $x_1>y_1$ forces $q<p$ while $q+p\le1$, every such pair satisfies
\begin{equation*}
    q \;<\; \tfrac12 .
\end{equation*}
The family therefore never reaches a pair such as $(0.6,0.8)$, and additivity of the integral does not recover it, since subtracting two inequalities of the same sense yields nothing. In \textsf{Normative Selectivity} the calibration point cancels when the first-order conditions are differenced, so the interval runs between the two \emph{compared} alternatives, $[\rho({\bm x},\cdot),\rho({\bm x}',\cdot)]$, and its position in $(0,1)$ is free. That is why \textsf{Normative Selectivity}, and not the compression property, is what Proposition \ref{prop:curvature_char} converts into Definition \ref{norm-APU}(iii) on all of $(0,1)$, and why the compression property is implied rather than imposed (Proposition \ref{prop:penalty_implied}).

\emph{\citeauthor{FIS_2015}'s \textsf{Richness}.} This one too is implied by the others, which is why it does not appear among the axioms of Theorem \ref{representation_norm}(a).

\begin{proposition}[\textsf{Richness is implied}]
\label{prop:richness_redundant}
If $\rho$ satisfies \textsf{Positivity}, \textsf{Selfishness-Acyclicity}, \textsf{Normative Independence}, \textsf{Normative Richness} and \textsf{Normative Richness}~(3), then $\rho$ satisfies \textsf{Richness}. Neither \textsf{Normative Weak Order}, nor \textsf{Normative Penalty}, nor \textsf{Normative Selectivity}, nor \textsf{Continuity} is used.
\end{proposition}

\begin{proof}
Let $\{g_{\bm a}\}$ be the family given by Axiom \ref{ax:norm_independence}. Fix ${\bm x}\in X$ and $p,q\in(0,1)$ with $p+q\le1$, and set $r:=1-p-q$. Note that whenever ${\bm x}\succsim_n^\rho{\bm y}$ and ${\bm x}\succsim_n^\rho{\bm z}$, the allocation ${\bm x}$ sets the ceiling of each of $\{{\bm x},{\bm y}\}$, $\{{\bm x},{\bm z}\}$ and $\{{\bm x},{\bm y},{\bm z}\}$, so the cross-menu identity of Axiom \ref{ax:norm_independence} applies to those menus with ${\bm a}={\bm b}={\bm x}$.

Write $\mathrm{gap}({\bm y}):=g_{\bm x}(\rho({\bm x},\{{\bm x},{\bm y}\}))-g_{\bm x}(\rho({\bm y},\{{\bm x},{\bm y}\}))$ for the binary gap below ${\bm x}$. Two facts are used throughout: $\rho(\cdot,A)$ is a probability distribution, so the probabilities in any menu sum to one, while \textsf{Positivity} is what places each of them in the open interval $(0,1)$, where alone $g_{\bm x}$ and $g_{\bm x}^{-1}$ may be applied. In particular $\mathrm{gap}({\bm y})=g_{\bm x}(1-\pi)-g_{\bm x}(\pi)$ where $\pi:=\rho({\bm y},\{{\bm x},{\bm y}\})\in(0,1)$; as $\pi$ ranges over $(0,1)$ this expression is continuous, strictly decreasing, and by the boundary behavior of $g_{\bm x}$ it sweeps all of $\mathbb R$. \textsf{Normative Richness}~(3) realizes every $\pi$ by some ${\bm y}$ with ${\bm x}\succsim_n^\rho{\bm y}$, so \emph{every} value of $\mathrm{gap}$ is attained below ${\bm x}$.

\emph{Case $r=0$.} Choose ${\bm y}$ with $\rho({\bm y},\{{\bm x},{\bm y}\})=q$ and put ${\bm z}:={\bm y}$ (Definition \ref{richness} does not require ${\bm y},{\bm z}$ distinct). Then $\{{\bm x},{\bm y},{\bm z}\}=\{{\bm x},{\bm y}\}$ and $\rho({\bm x},\cdot)=1-q=p$.

\emph{Case $r>0$, $q\ne r$.} Choose ${\bm y},{\bm z}$ below ${\bm x}$ with $\mathrm{gap}({\bm y})=g_{\bm x}(p)-g_{\bm x}(q)$ and $\mathrm{gap}({\bm z})=g_{\bm x}(p)-g_{\bm x}(r)$; these two values differ because $q\ne r$ and $g_{\bm x}$ is injective, so ${\bm y}\ne{\bm z}$, and both differ from ${\bm x}$. Put $B:=\{{\bm x},{\bm y},{\bm z}\}$ and $\sigma:=g_{\bm x}(\rho({\bm x},B))$. The cross-menu identity, applied to $\{{\bm x},{\bm y}\}$ and $B$ and then to $\{{\bm x},{\bm z}\}$ and $B$, gives
\begin{equation*}
    g_{\bm x}(\rho({\bm y},B)) = \sigma-\bigl[g_{\bm x}(p)-g_{\bm x}(q)\bigr],
    \qquad
    g_{\bm x}(\rho({\bm z},B)) = \sigma-\bigl[g_{\bm x}(p)-g_{\bm x}(r)\bigr].
\end{equation*}
Hence $\sigma$ solves $F(\sigma)=1$, where $F(\sigma):=g_{\bm x}^{-1}(\sigma)+g_{\bm x}^{-1}\bigl(\sigma-g_{\bm x}(p)+g_{\bm x}(q)\bigr) +g_{\bm x}^{-1}\bigl(\sigma-g_{\bm x}(p)+g_{\bm x}(r)\bigr)$ is strictly increasing. Since $F(g_{\bm x}(p))=p+q+r=1$, the solution is $\sigma=g_{\bm x}(p)$, so $\rho({\bm x},B)=p$ and $\rho({\bm y},B)=q$, as required.

\emph{Case $r>0$, $q=r$.} Take ${\bm y}$ as in the previous case. For ${\bm z}$ use \textsf{Normative Richness} with $S=\{{\bm x},{\bm y}\}$ and $v=y_1$: it supplies ${\bm z}\notin\{{\bm x},{\bm y}\}$ with $z_1=y_1$ and ${\bm x}\succsim_n^\rho{\bm z}$. Equal selfish payoffs force equal probabilities in $B=\{{\bm x},{\bm y},{\bm z}\}$: if $\rho({\bm y},B)>\rho({\bm z},B)$, then taking $n=2$, ${\bm x}^1={\bm y}$, ${\bm x}^2={\bm z}$ and $A^1=A^2=B$ in Axiom \ref{ax:selfish_acyclicity} yields a forbidden cycle, since $y_1\ge z_1$ and $z_1\ge y_1$ both hold while the first probability inequality is strict; the symmetric argument excludes $\rho({\bm z},B)>\rho({\bm y},B)$. With $\rho({\bm z},B)=\rho({\bm y},B)$ and the cross-menu identity applied to $\{{\bm x},{\bm y}\}$ and $B$, the budget identity becomes $\tilde F(\sigma)=1$ with $\tilde F(\sigma):=g_{\bm x}^{-1}(\sigma)+2g_{\bm x}^{-1}\bigl(\sigma-g_{\bm x}(p)+g_{\bm x}(q)\bigr)$, again strictly increasing, and $\tilde F(g_{\bm x}(p))=p+2q=1$. So $\rho({\bm x},B)=p$ and $\rho({\bm y},B)=q$.
\end{proof}

Two features of this proof deserve note. The mechanism is that \textsf{Normative Independence} converts binary data into ternary data: the cross-menu identity pins the probabilities in a three-element menu to the binary gaps, and \textsf{Normative Richness}~(3) says those gaps sweep $\mathbb R$. And the finite-set-avoidance clause of \textsf{Normative Richness} is what handles the tie case---two distinct allocations of equal selfish payoff below a common ceiling. That clause is also why one tempting independence witness fails: taking $X$ to be a line, say $\mathbb R\times\{0\}$, does defeat \textsf{Richness} at the ties, since no two distinct allocations share a selfish payoff, but for the same reason it defeats the avoidance clause, so it separates nothing. Our domain assumption permits such an $X$; the axioms themselves do not.

\emph{A structural consequence.} The same reasoning delimits how \textsf{Normative Independence} can fail. Two menus with $\sim_n^\rho$-equivalent ceiling-setters dominate one another, so \textsf{Normative Penalty}, applied in both directions to a calibrated common element, forces their choice probabilities to agree on all more-selfish alternatives; hence the cost cannot vary \emph{within} a ceiling class without violating \textsf{Normative Penalty}. Any violation of \textsf{Normative Independence} compatible with the other axioms must therefore be a \emph{cross-ceiling} one, in which each ceiling class is internally consistent but the induced utility indices at different ceilings are not affinely related---precisely the example constructed above.

\section{The Fechnerian Representation and Psychometric Functions}
\label{app:psychometric}

In binary menus, our model admits a natural connection to Fechnerian models of stochastic choice and psychometric functions. Consider a binary menu $A = \{{\bm x},{\bm y}\}$ and let $\Delta u := u(x_1) - u(y_1)$ denote the selfish utility difference. The model's binary choice probability is already Fechnerian in this difference: by the interior first-order conditions,
\begin{equation*}
    c_1\bigl(\rho({\bm x},A),t\bigr) - c_1\bigl(1-\rho({\bm x},A),t\bigr) \;=\; \Delta u,
    \qquad t := \max\{\varphi({\bm x}),\varphi({\bm y})\},
\end{equation*}
and the left-hand side is strictly increasing in $\rho({\bm x},A)$ and onto $\mathbb R$ by Definition \ref{norm-APU}(i)--(ii). So $\rho({\bm x},A) = F(\Delta u;t)$ for a function $F(\cdot;t):\mathbb R\to(0,1)$ that is strictly increasing with $F(0;t)=\tfrac12$---a psychometric function with the normative ceiling as its only menu input.

Standard parametric psychometric families, including the one used below, instead take inputs in $[0,1]$, so for the comparison we rescale. Here $X$ is unbounded and $u$ has full range, so there is no maximal utility difference to normalize by; we therefore fix a reference spread $W>0$, restrict attention to pairs with $|\Delta u|\le W$, and set
\begin{equation*}
    \hat{\delta}({\bm x},{\bm y}) \;=\; \frac{\Delta u + D}{2D} \;\in\; [0,1].
\end{equation*}
Nothing substantive turns on $W$: it fixes the window of utility differences displayed and rescales the horizontal axis, and the flattening comparison below is made at a common $W$ across ceilings. Here $\hat{\delta}({\bm x},{\bm y}) = 1/2$ corresponds to indifference ($u(x_1) = u(y_1)$), while $\hat{\delta}({\bm x},{\bm y}) > 1/2$ indicates $u(x_1) > u(y_1)$. Figure \ref{fig:psychometric} illustrates the effect. In this parametrization the binary choice probability reads
\begin{equation*}
    \rho({\bm x}, \{{\bm x},{\bm y}\}) = F\bigl(\hat{\delta}({\bm x},{\bm y}); t\bigr),
\end{equation*}
with $F(\cdot;t)$ now defined on $[0,1]$ and $F(\tfrac12;t)=\tfrac12$.

Crucially, Proposition \ref{prop:marginal_cost} is consistent with flattening in that as the normative pressure $t$ increases, the \emph{curvature} $c_{11}(\cdot,t)$ of the randomization cost increases---the level of $c_1$ across ceilings is a normalization choice and carries no behavioral content (Remark \ref{rmk:level_corollary}). In the Fechnerian representation, this translates to the psychometric function $F(\cdot; t)$ becoming flatter (less sensitive to the underlying utility difference $\hat{\delta}$) around the indifference point. Thus, normative pressure endogenously compresses the choice probabilities, structurally resembling the probability weighting functions commonly observed in Prospect Theory \citep{W_2010}, but parameterized by the menu's normative environment rather than exogenous cognitive biases.

\begin{figure}[ht]
    \centering
    \includegraphics[width=0.8\textwidth]{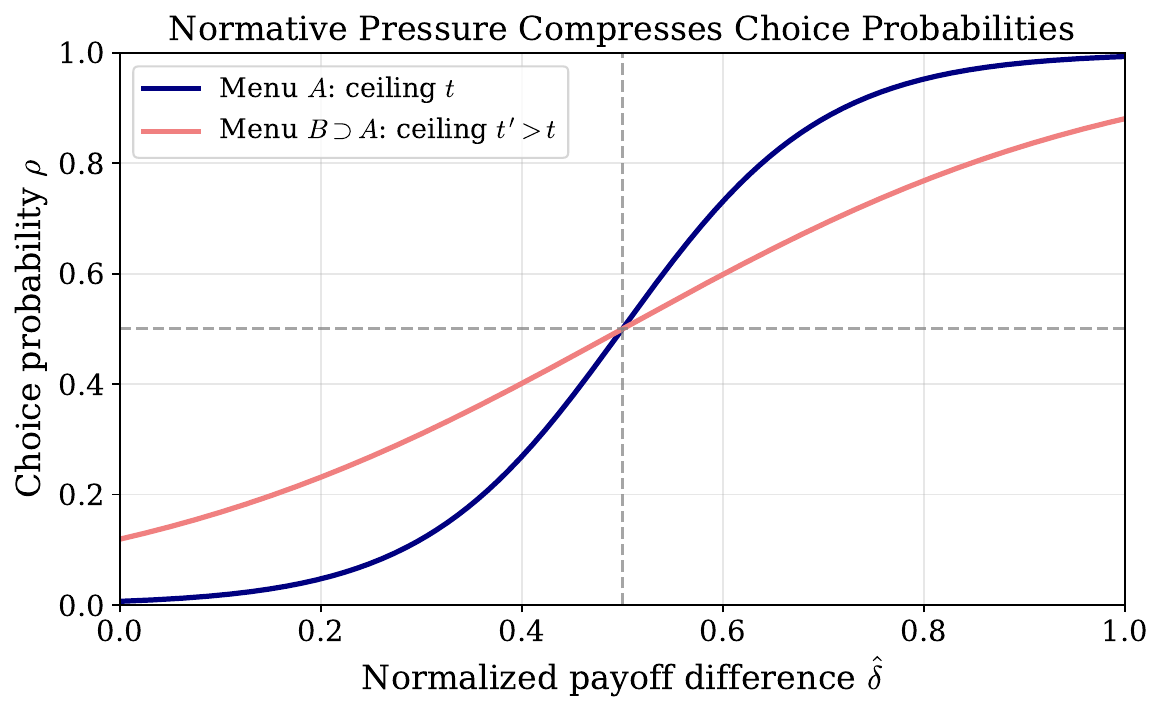}
    \caption{Illustrative psychometric functions under two normative ceilings $t<t'$. The flatter curve is the one at the \emph{higher} ceiling $t'$: greater normative pressure compresses choice probabilities toward the indifference point and so makes the response less sensitive to the normalized payoff difference $\hat\delta$. Both curves pass through $(\tfrac12,\tfrac12)$. The figure is a conceptual illustration; the horizontal axis uses the rescaling $\hat\delta=(\Delta u+W)/2W$ introduced above, for a fixed reference spread $W$ common to the two ceilings.}
    \label{fig:psychometric}
\end{figure}

\subsection*{Numerical examples}

To illustrate the flattening effect, we adopt a widely used parametric form from the literature on decision-making under risk \citep{W_2010}:
\begin{align}\label{pw_2}
F(\hat{\delta}; t) = \dfrac{b(t) \hat{\delta}^{a(t)}}{b(t) \hat{\delta}^{a(t)} + (1-\hat{\delta})^{a(t)}},
\end{align}
where the input $\hat{\delta} \in [0, 1]$ represents the normalized utility difference. The parameters depend on the normative ceiling $t$: $a(t)$ captures \textit{likelihood sensitivity}, and $b(t)$ captures \textit{pessimism} \citep{W_2010}.

Formally, we obtain several properties. As the sensitivity parameter $a(t)$ decreases (driven by higher normative pressure $t$), the level of likelihood insensitivity increases, flattening the curve. Moreover, as parameter $b(t)$ decreases, the decision maker becomes more pessimistic. Finally, to have the symmetric form of $F$, we obtain the condition of $b(t)^2 = 1$ (i.e., $b(t)=1$).

We explore the interpretation of $F$ in our context. Under the condition that $F$ is symmetric, fix an allocation ${\bm x} \in X$ and let $\hat{\delta} \propto u(x_1) - u(y_1)$ for each alternative allocation ${\bm y} \in X$. As $\hat{\delta}$ increases, allocation ${\bm y}$ becomes less selfish than the fixed allocation ${\bm x}$. In this case, if $F$ is inversely S-shaped (i.e., $a(t) < 1$ due to high $t$), the decision maker chooses stochastic prosocial behavior more frequently than in the case of the S-shaped form.\footnote{The stylized fact in stochastic choice shows that $F$ is generally S-shaped \citep{S_2025}.} This form captures pessimism, suggesting that uncertainty and pressure about norms can be key to altruistic behavior.

\begin{figure}[htbp]
	\begin{center}
		\includegraphics[clip, width=0.72\textwidth]{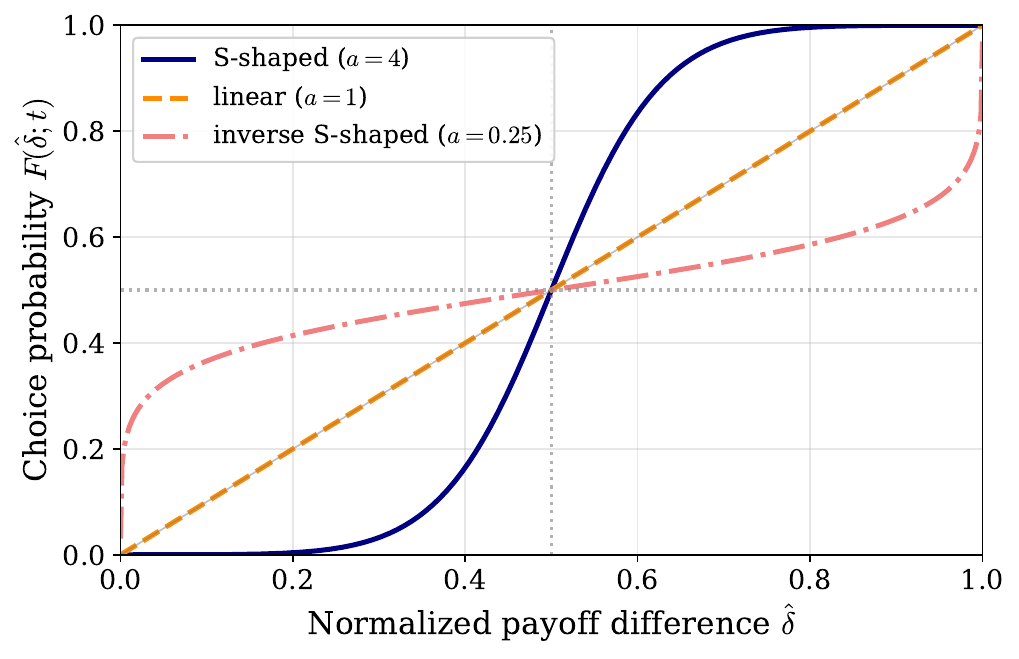}
		\caption{The three qualitative shapes of \eqref{pw_2}, at $b=1$: S-shaped for $a>1$, linear for $a=1$, and inverse S-shaped for $a<1$. All three pass through $(\tfrac12,\tfrac12)$; the thin diagonal is the linear benchmark. Higher normative pressure lowers $a(t)$ and so moves the curve from the S-shaped to the inverse S-shaped region, flattening the response around indifference.}
		\label{fig_example}
	\end{center}
\end{figure}

Figure \ref{fig_pw_2} demonstrates how varying these parameters alters the shape of the function.
\begin{figure}[htbp]
	\begin{center}
		\includegraphics[clip, width=0.78\textwidth]{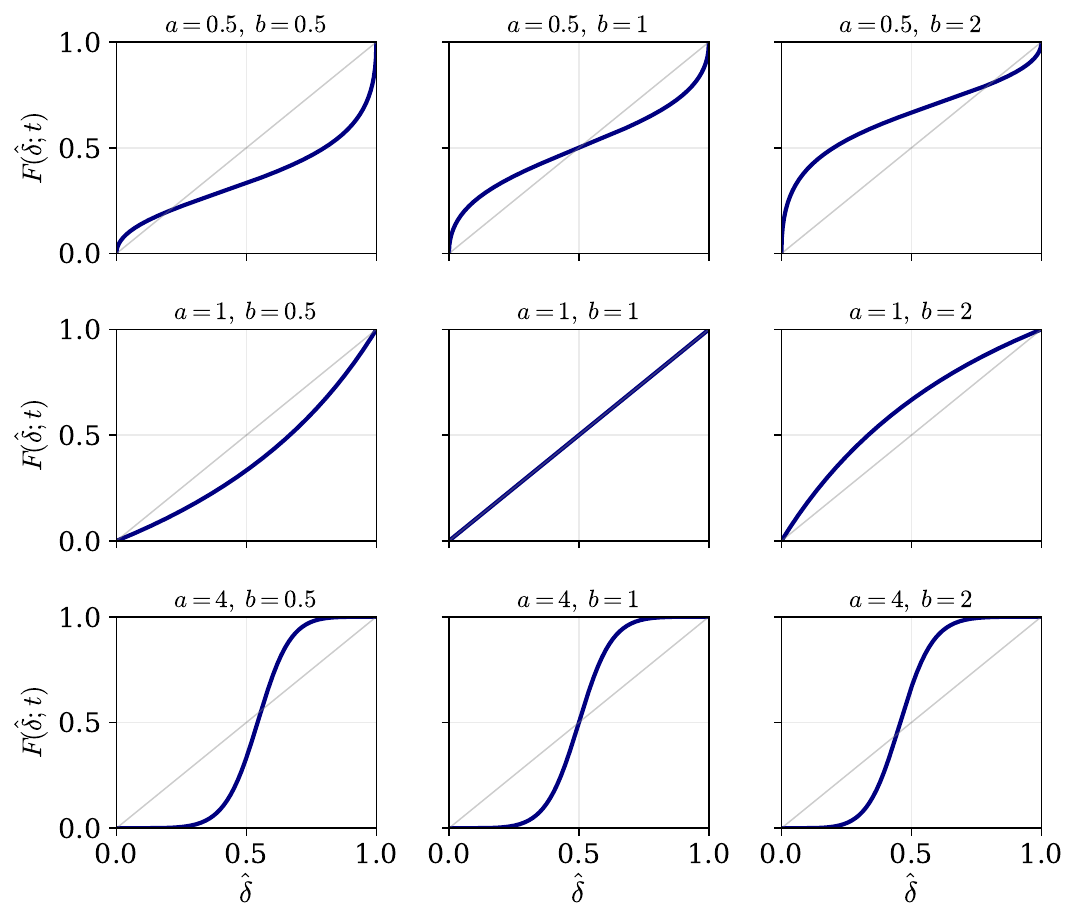}
		\caption{Equation \eqref{pw_2} across the two parameters. Rows vary the likelihood-sensitivity parameter $a$ (flat to steep); columns vary the pessimism parameter $b$, which shifts the curve without changing its curvature class. The thin diagonal in each panel is the linear benchmark $a=b=1$. The horizontal axis is the normalized payoff difference $\hat\delta$, not the normative ceiling $t$.}
		\label{fig_pw_2}
	\end{center}
\end{figure}

In the context of decision-making under risk or uncertainty, the parameters of $a$ and $b$ are estimated as $a = 0.69$ and $b = 0.77$ \citep{W_2010}. In summary, in the context of social decision-making, the parameters have not yet been explored. The inverse S-shaped form can capture stochastic prosocial behavior. Understanding the relationship between pessimism and altruism is important.

\bibliographystyle{econ-econometrica}
\bibliography{literature}

\end{document}